\documentclass[twocolumn,showpacs,preprintnumbers,amsmath,amssymb,nofootinbib,pra,floatfix]{revtex4-1}

\usepackage{mathrsfs,amsthm,clrscode,comment,booktabs,graphicx,subfig}

\newtheorem{theorem}{Theorem}
\newtheorem{proposition}{Proposition}
\newtheorem{lemma}{Lemma}
\newtheorem{corollary}{Corollary}

\newenvironment{definition}[1][Definition]{\begin{trivlist}
\item[\hskip \labelsep {\bfseries #1}]}{\end{trivlist}}

\newenvironment{remark}[1][Remark]{\begin{trivlist}
\item[\hskip \labelsep {\bfseries #1}]}{\end{trivlist}}

\newcommand{\lst}{\vec}
\newcommand{\set}{\tilde}

\newcommand{\genfun}{\tilde{\mathcal{G}}}
\newcommand{\pauligroup}{{\set{\mathfrak{P}}}}
\newcommand{\powerset}{\set{\mathcal{P}}}
\newcommand{\centralizer}{\set{\mathcal{C}}}
\newcommand{\pseudoproduct}{\set\Pi}
\newcommand{\unpack}{\set U}
\newcommand{\optimizer}{\lst{\mathcal{O}}}

\newcommand{\half}{\frac{1}{2}}

\newcommand{\om}{\omega}

\newcommand{\unit}[1]{\,\text{#1}}

\renewcommand{\min}{\unit{min}}

\renewcommand{\max}{\text{max}}

\renewcommand{\choose}[2]{\left(\begin{matrix}#1\\#2\end{matrix}\right)}

\newcommand\cancel{\bgroup \markoverwith{---}\ULon}

\newcommand{\paren}[1]{\left(#1\right)}

\begin{document}

\title{Automated Searching for Quantum Subsystem Codes}

\author{Gregory M. Crosswhite}
\affiliation{Department of Physics, University of Washington, Seattle, 98195}
\author{Dave Bacon}
\affiliation{Department of Computer Science \& Engineering, University of Washington, Seattle, 98195}
\affiliation{Department of Physics, University of Washington, Seattle, 98195}


\email{gcross@phys.washington.edu, dabacon@cs.washington.edu}

\begin{abstract}
Quantum error correction allows for faulty quantum systems to behave in an effectively error free manner.  One important class of techniques for quantum error correction is the class of \emph{quantum subsystem codes}, which are relevant both to active quantum error correcting schemes as well as to the design of self-correcting quantum memories.
Previous approaches for investigating these codes have focused on applying theoretical analysis to look for interesting codes and to investigate their properties.  In this paper we present an alternative approach that uses \emph{computational} analysis to accomplish the same goals.  Specifically, we present an algorithm that computes the optimal quantum subsystem code that can be implemented given an arbitrary set of measurement operators that are tensor products of Pauli operators.  We then demonstrate the utility of this algorithm by performing a systematic investigation of the quantum subsystem codes that exist in the setting where the interactions are limited to 2-body interactions between neighbors on lattices derived from the convex uniform tilings of the plane.
\end{abstract}

\maketitle

\newpage


Quantum computers are a technological possibility because there exist methods for building these computers out of physical components that fail to operate in an error-free manner.  The theory behind achieving this makes up the field of quantum error correction~\cite{Shor:95a,Steane:96a,Steane:96b,Steane:96c,Knill:97a,Gottesman:97a} and fault-tolerant quantum computing~\cite{Shor:96a,Aharonov:97a,Knill:98a,Knill:98b,Preskill:98a,Aliferis:05a}. Of particular note is the threshold theorem for fault-tolerant quantum computing~\cite{Aharonov:97a,Knill:98a,Knill:98b,Aliferis:05a}.  This theorem says that if a quantum system decoheres slowly enough, and sufficiently precise control is maintained over the system, then effectively arbitrary error-free quantum computations can be performed.  The way that this is achieved is through the use of quantum information which is encoded across multiple quantum subsystems into a quantum error correcting code.

Different quantum codes have different advantages and disadvantages for implementation in a fault-tolerant device~\cite{Cross:07a}.  In this paper we undertake a study of an important class of quantum codes, quantum stabilizer subsystem codes~\cite{Poulin:05a,Kribs:05a,Kribs:05b,Kribs:06a} generated by measurements that are tensor products of Pauli operators.  Part of the significance of this class of codes is that they can be used to implement \emph{passive} fault tolerance by turning the measurement operators into interaction terms forming a Hamiltonian that provide energetic protection against errors;  the first example of such an approach was the toric code and related models due to Kitaev~\cite{Kitaev:97c,Kitaev:03a}, and a plethora of related approaches have now been investigated~\cite{Barnes:00a,Bacon:01b,Jordan:05a,Weinstein:05b,Bacon:06a,Bacon:08b,Nayak:08a,Bombin:09a,Chesi:10a}.

Previous approaches for studying quantum subsystem codes have focused on using theoretical analysis to find and investigate new quantum subsystem codes.  While powerful, theoretical analysis has some disadvantages:  it is limited to the `cleverness' of the analyst, and it can be prohibitively expensive to perform systematic searches of large parameter spaces to pick out the gems in the dust.  In this paper, we present an alternative approach that uses \emph{computational} analysis to accomplish the same goals.  The advantage of this approach is that one becomes limited by the power of the computer rather than the brain of analyst\footnote{Of course, this is also the main \emph{disadvantage} of this approach.}.

In this paper we develop an algorithm that computes the optimal subsystem code for a given set of measurements consisting of tensors products of Pauli operators.  In the process of doing this we also develop a formalism that allows us to prove that the algorithm is correct and that the code it compute is indeed the optimal code for the given measurements.  We also prove bounds on the running time of the algorithm that show that the algorithm terminates (relatively) quickly when the optimal code is not very robust to errors.  Because of this property, the algorithm can be applied to sift through a class of possible measurements to determine which (if any) result in a robust code.

To demonstrate the use of this algorithm, we focus on classes of measurement operators where each measurement is limited in action to two qubits --- that is, to operators taking the form $P_i \cdot Q_j$, where $P_i$ and $Q_j$ are Pauli operators acting on respectively the $i^{\text{th}}$ and $j^{\text{th}}$ qubit of the system;  examples of previous subsystem codes that have been constructed with this structure are the quantum compass model subsystem code~\cite{Bacon:06a} (including generalizations~\cite{Bacon:06b,Bravyi:10a}) and topological subsystem codes~\cite{Bombin:10a}.  In particular we focus on systems where the measurement operators only couple qubits that are neighboring on a periodic lattice arising from the convex uniform tilings of the plane.  We perform a systematic study of the codes on lattices arising from nine of the eleven such tilings, and present the results of this search.
\listoftables
\tableofcontents
\section{Introduction}

We begin by a brief review of the notion of quantum error correcting codes and in particular the subsystem stabilizer codes~\cite{Poulin:05a}.

In quantum computation we seek to reliably store and manipulate quantum information.  Unfortunately, real quantum systems are open systems that couple to their environment and quickly lose their coherence through the process of decoherence.  Even more troubling, when one wishes to manipulate quantum information one can only do this with a fixed precision.  While considerable progress has been made in finding systems with long coherence times, inevitably current quantum computers will fail before they achieve anything close to the amount of computation needed, for example, to break a public key cryptosystem~\cite{Shor:94a}.  However it turns out that one can generally repair damage to quantum information as long as one knows the form that the damage took.  Furthermore one can build a `trap' --- that is to say, a
\emph{quantum code} --- that tricks nature into giving up the information about what damage has occurred to the quantum system.

The nature of codes is that they separate the space in which our computation lives from the space in which the physical information is
stored; that is to say, although we design our quantum circuits to operate on some Hilbert space of qubits $\mathscr{C}$, each of these qubits
does \emph{not} directly correspond to a physical qubit, but rather there is some isomorphism that relates the entire Hilbert space $\mathscr{C}$
to the Hilbert space of physical qubits, $\mathscr{P}$.  To distinguish between these two Hilbert spaces, we call the Hilbert space of qubits in whose terms the computation is expressed the \emph{computational space} (or \emph{logical space}), and the space of qubits which have physically been built the
\emph{physical space}.  Merely building an isomorphism between these two spaces is not enough to allow us to correct errors.  For one thing, we need to add extra qubits to the computational space that contain a record of the damage that we can read out; thus, we shall say that the full computational space is $\mathscr{C}:=\mathscr{R}\times\mathscr{Q}$, where the qubits that live in $\mathscr{R}$ have the role of keeping a
record of the errors that have been introduced by the environment, and the qubits that live in $\mathscr{Q}$ are the qubits in whose terms our quantum
algorithm is expressed.  

We have to pick a strategy for reading out the information in $\mathscr{R}$ about the errors that have occurred on our system.  One natural choice is to perform a single-qubit Pauli $Z$ operator measurement on each qubit on $\mathscr{R}$.  In order to build the `trap' element into our system, we need to ensure that whenever nature strikes at the physical space $\mathscr{P}$ and produces errors in a form that we intend to correct, this action must be isomorphic to a strike on the computational space that leaves a \emph{measurable} record in $\mathscr{R}$.  For our choice of measuring Pauli $Z$ errors, these are errors that are isomorphic to any operator that \emph{anti-commutes} with a $Z$ operator of at least one of the qubits in $\mathscr{R}$.  Note that although we speak of measuring the qubits in $\mathscr{R}$, the measurement operator of interest in $\mathscr{R}$ is mapped to an operator in the physical space $\mathscr{P}$; this isomorphic operator is referred to as a \emph{stabilizer}, and the full set of operators on $\mathscr{P}$ which are isomorphic to our chosen measurement operators on $\mathscr{R}$ are referred to as the \emph{stabilizers} of the code.

Up to this point, the formalism we have described is known as \emph{stabilizer codes}~\cite{Gottesman:96a,Gottesman:97a,Calderbank:97a,Calderbank:97b} and its essential characteristic is that in determining the syndrome of the physical error, one makes a measurement of all of the qubits in $\mathscr{R}$. What if, however, we relaxed this constraint and only measured some of the qubits in $\mathscr{R}$?  That is to say, what if we split the qubits in $\mathscr{R}$ into two categories: \emph{stabilizer
qubits} whose states we care about and which we measure to obtain an error syndrome, and \emph{gauge
qubits} whose states we do not care about.  (The latter get their name
from the fact that they provide a `gauge' degree of freedom, i.e. a
degree of freedom that is irrelevant to us.)  Then we would have that
$\mathscr{R}=\mathscr{S}\times \mathscr{G}$, where $\mathscr{S}$ is
the space in which the stabilizer qubits live, and $\mathscr{G}$ is
the space in which the so-called gauge qubits live; such a scheme is
known as a \emph{stabilizer subsystem code}~\cite{Poulin:05a}.  In this case, we shall use the term
\emph{stabilizers} to denote the set of operators in $\mathscr{P}$ which are isomorphic to our chosen measurement operators of interest in $\mathscr{S}$.

At first there might not seem to be an advantage to this approach, since it essentially means adding qubits to our code that are
`wasted'; however, in practice subsystem codes have many advantages.  The first advantage is that since we do not care about what happens to the gauge qubits, some quantum errors on the system will neither result in detectable errors nor destroy the information in the logical qubits~\cite{Poulin:05a,Kribs:05a,Kribs:05b,Nielsen:05a,Kribs:06a,Bacon:06a}.  A second advantage is that we no longer need our error-correcting measurements on the physical system to commute with each other, as long as they all commute with the stabilizers and logical qubit operators, since then the fact that they do not commute only affects the gauge qubits, which we do not care about~\cite{Aliferis:07a}.  This sometimes allows one to effectively measure a stabilizer which is a non-trivial $k$-qubit measurement by using a series of two qubit measurements~\cite{Aliferis:07a}.  The individual measurements in this series do not commute (so they cannot be simultaneously measured), however the stabilizer syndrome can nonetheless be reconstructed from these measurements.  A third advantage arises from the fact that subsystem codes often require {\em fewer} measurements to diagnose errors than similar non-subsystem codes, which results in improved performance~\cite{Aliferis:07a,Cross:07a};  counterintuitively, turning stabilizer codes into subsystem stabilizer codes often results in higher thresholds for fault-tolerant quantum computing.  Finally, subsystem codes can often be implemented in a more local manner than non-subsystem codes as exemplified by the quantum compass model code~\cite{Bacon:06a,Aliferis:07a}.

There are now many examples of stabilizer subsystem codes in the literature.  One of the first non-trivial subsystem codes to be described is a code related to the quantum compass model in two-dimensions~\cite{Bacon:01a,Dorier:05a,Bacon:06a}.  In the quantum model one considers a Hamiltonian on a two-dimensional square lattice where nearest horizontal neighbors couple the $x$ component of their spins and nearest vertical neighbors couple the $z$ component of their spin, so that the Hamiltonian is given by
\begin{equation}
H=-\Delta \sum_{i,j} (X_{i,j} X_{i+1,j} +Z_{i,j} Z_{i,j+1}),
\end{equation}
where $P_{i,j}$ represents the Pauli operator $P$ acting on qubit at location $(i,j)$.  This model is interesting for a few reasons.  The first is that the energy levels of this system can be best thought of as elements of a quantum error correcting subsystem code.  The second reason is that the model provides some amount of protection from quantum errors because errors are energetically unfavored\footnote{Unfortunately, in this particular system the protection vanishes as the size of the lattice goes to infinity~\cite{Dorier:05a}, but for small lattice sizes there is some protection from errors due to the energy level structure of the system~\cite{Bacon:01a}.}.  Many other examples of systems which have energy protecting properties are also known: the most famous being Kitaev's toric code in two and four-spatial dimensions~\cite{Kitaev:97c,Kitaev:03a,Dennis:02a}.  The study of such systems is still in its infancy and one central question is whether there exist Hamiltonians with reasonable physical parameters (such as existing in three or fewer spatial dimensions and involving 2-body interactions~\cite{Bravyi:09a,Bravyi:10b}) whose physics enact quantum error correction on the system when the system is in contact with a thermal reservoir; such systems are called \emph{self-correcting} quantum computers~\cite{Bacon:06a,Bombin:09a}.  In this paper we will talk about quantum subsystem codes from the perspective of active error correction where error syndromes are identified through carefully engineered measurements, but it shall be understood that this formalism can equivalently be seen from the perspective of passive error correction where errors are guarded against by carefully engineered interactions.  That is, measurement operators in the active error correction picture are equivalent to interactions in the passive error correction picture.

Because we ultimately want to build a system implementing our measurements, physical considerations typically constrain our measurements to be \emph{local}, which means that they can be expressed in the physical space as a tensor product of single-qubit Pauli operators --- i.e, for each measurement operator $o$ we have that
$o := \bigotimes_i P_i$ where $P_i$ is the pauli operator $P$ acting on the $i^{\text{th}}$ qubit.  An important question then is which sets of local measurements give rise to useful quantum error correcting subsystem codes.

Approaches to answering this question typically involve applying theoretical analysis with varying degrees of cleverness.  In this paper we present an alternative approach.  In section~\ref{sec:algorithm}, we present an algorithm which for every set of local measurement operators computes a quantum subsystem code that arises from the algebra of these operators\footnote{The code that we find is almost never unique, since among other transformations one can multiple every gauge and logical qubit operator by an element from the stablizers and end up with an equivalent code.}.  Along the way we develop a formalism that allows us to prove not only that this algorithm is correct, but also that the code that it computes is \emph{optimal} in the sense that there exists no other code arising from the same set of measurements for which the distance of any of the logical qubits has been increased.  This property makes this algorithm useful for analyzing the properties of codes arising from measurements that are too overwhelming to analyze by hand.  

We shall also show that an important property of this algorithm is that it terminates (relatively) quickly when the distance of the code is small, which allows it to be used not only to solve for individual codes, but also to search through entire classes of sets of measurements to see if any have high-distance qubits.  Motivated by previous results demonstrating the utility of codes implemented using systems on a lattice, we undertake a systematic investigation of codes where the measurement operators are restricted to the 2-body interactions arising from the edges of periodic lattices derived from the 11 regular tilings.  In section~\ref{sec:lattice} we discuss our approach for applying the algorithm to perform a systematic search for codes that can be implemented on these tilings, and we the present numerical results that we obtained. In section~\ref{sec:conclusion} we present our conclusions.
\subsection{Notation}

In this paper we adopt the following conventions for notation:

\begin{itemize}
\item \emph{sets} are denoted by a variable with a tilde, e.g. $\tilde A$;
\item \emph{sequences} are denoted by a variable with an arrow, e.g. $\vec A$;
\item \emph{operators} and \emph{integers} are denoted by using lower-case letters, e.g. $o$ and $i$;
\item \emph{collections} of \emph{operators} and \emph{pairs of operators} are denoted by using upper-case letters with either a tilde or an arrow above them, e.g. $\tilde O$ and $\vec O$;
\item \emph{collections} of \emph{integers} are denoted by using lower-case letters with either a tilde or an arrow above them, e.g. $\tilde k$ and $\vec k$;  and
\item \emph{collections} of \emph{other kinds} of objects are typically denoted by capital letters in a fancy script.
\end{itemize}
\section{Theory} \label{sec:algorithm}
\subsection{Construction of the subsystem code}

\begin{remark}
This subsection describes by way of a constructive proof how to compute given a set of measurement operators the quantum code that can be implemented by these operators.  For a listing of pseudo-code that implements the algorithm described in this proof, see Table \ref{table:algo1} near the end of this subsection.
\end{remark}
Although conceptually a subsystem code is an isomorphism $T$ such that  $\mathscr{P}\approx^T \mathscr{S}\times\mathscr{G}\times\mathscr{Q}$ --- that is, an isomorphism between the \emph{physical} space of qubits and the \emph{computational} space of qubits in whose terms our computation is actually expressed --- we do not need to actually construct this isomorphism in order to be able to use the code.  Since all of our work will be done on the physical system anyway, it suffices to know the operators in the physical space $\mathscr{P}$ that are isomorphic to the qubit measurement operators of interest in the computational space $\mathscr{S}\times\mathscr{G}\times\mathscr{Q}$, and it is exactly the operators on $\mathscr{P}$ that the algorithm we present shall compute\footnote{If one really wanted to, one could explicitly construct the isomorphism $\mathscr{T}$ from these operators by computing the unitary operator which simultaneously diagonalizes a the maximal subset of commuting measurements from this set of operators on $\mathscr{P}$, but in practice this is not particularly useful.}.

When one wants to define a qubit in terms of its measurement operators, it suffices to define two operators that anti-commute with each other but which commute with all of the others measurement operators that have been defined, since this gives us the $X$ and $Z$ measurements on the qubit which are sufficient to generate the full $Pauli$ group (minus phases).  Since working with such pairs of operators shall be a common theme in this algorithm, we shall introduce the following definition in order to simplify the language used to describe them.

\begin{definition} A pair of operators is a \emph{conjugal pair in relation to the set} $\set X$ when each of the operators in the pair commutes with every operator in $\set X$ except for its \emph{conjugal partner} --- that is, the other operator in the conjugal pair --- should its conjugal partner be a member of $\set X$.
\label{conjugal-pair-definition}
\end{definition}

Note that we have explicitly not required that the operators in the conjugal pair be members of $\set X$ in order to be a conjugal pair in relation to it.  However, should both operators be members of $\set X$, then neither operator can belong to a different conjugal pair with respect to $\set X$, since in that case there would be an operator in $\set X$ (namely, its original conjugal partner) with which it anti-commutes that was not its conjugal partner in the new pair, leading to a contradiction.

For convenience, we introduce the following additional definitions:

\begin{definition}

\begin{enumerate}
\item $\pauligroup$ is the group of Pauli operators --- that is, the group of tensor products of the (unnormalized) Pauli matrices --- acting on the physical space $\mathscr{P}$, \emph{modulo phases};
\item $\powerset(\set{S})$ is the power set of $\set S$, i.e. the set of all subsets of $\set S$; and
\item $\centralizer_\mathfrak{G}(\set S)$ is the centralizer of $\set S$, that is the subgroup of elements in $\mathfrak{G}$ which commute with $\set S$;
\item the function $\genfun:\powerset(\pauligroup)\to\powerset(\pauligroup)$ is defined such that $\genfun(\set S)$ is the set of all possible products of operators in $\set S$ --- that is, it is the set \emph{generated} by $\set S$.
\end{enumerate}

\end{definition}

We now introduce the main theorem of this subsection.

\begin{theorem} \label{theorem-SG} Suppose we are given a sequence of Pauli operators, $\lst O$.  Then there exist sets of Pauli operators $\set S\subseteq\pauligroup$, $\set G\subseteq\pauligroup$, and $\set L\subseteq\pauligroup$ such that
\begin{enumerate}
\item each of the operators in $\set S \cup \set G \cup \set L$ is independent from the rest --- i.e., no operator in this (unioned) set can be written as a product of other operators in the set;
\item each operator in $\set L \cup \set G$ is a member of a conjugal pair in relation to $\set S \cup \set G \cup \set L$;
\item $\genfun(\set S \cup \set G)=\genfun\paren{\{\lst O_i\}}$;\footnote{Here we use the notation $\{\vec{O}_i\}$ to refer to the set of elements in the sequence $\vec{O}$.}
\item and $\genfun(\set S \cup \set G \cup \set L)=\centralizer_\pauligroup(\set S )$
\end{enumerate}
\end{theorem}

\begin{remark}
This theorem follows, at least implicitly, from prior work on stabilizer codes~\cite{Gottesman:97a}, the definitions of stabilizer subsystem codes given by Poulin~\cite{Poulin:05a}, and the constructive approach to finding such codes as exemplified in~\cite{Bacon:06a}.  Because we wish to be constructive, however, we will present a full proof of this theorem and show how it gives rise to an algorithm for finding sets of Pauli operators which satisfy Theorem~\ref{theorem-SG}.  To be explicit, we note that $\set S$ will be a set of stabilizers (or equivalently, generators for the stabilizer group), $\set G$ will be a set of gauge qubit operators, and $\set L$ will be a set of logical qubit operators (i.e., those on which the computation is performed).

The main work in the proof of this theorem will be performed by proving several related propositions.  First we shall show how the set $\set G$ and a sequence $\lst S$ are constructed from the sequence of operators $\lst O$.  Since we want our stabilizers to form an independent set of operators, we shall then show that through a Gaussian elimination procedure it is possible to extract a list of independent operators from a sequence $\lst S$ resulting in a set $\set S$.  Finally, we shall show how using this same Gaussian elimination procedure we can transform a subset of the operators of $\set S\cup\set G$ into a form that makes it trivial to compute the logical qubit operators $\set L$.
\end{remark}
\begin{proposition} \label{proposition-SG} Suppose that we are given a sequence of Pauli operators $\lst O\subseteq \pauligroup$.  Then there exists a sequence of Pauli operators $\lst S\subseteq\pauligroup$ and a set of Pauli operators $\set G\subseteq\pauligroup$ such that
\begin{enumerate}
\item all of the operators in $\lst S$ commute with each other and also all of the operators in $\lst G$; \label{stabs-commute-with-G}
\item each operator in $\set G$ is a member of a \emph{conjugal pair} (Definition \ref{conjugal-pair-definition}) in relation to $\{\lst S_i\} \cup \set G $ \label{conjugal-pairs-commute-with-SAG}; and
\item $\genfun\paren{\{\lst S_i\}\cup \set G}=\genfun\paren{\{\lst O_i\}}$ \label{SAG-spans-all}.
\end{enumerate}
\end{proposition}

\begin{proof}
Proof by induction.  For the base case, note that if $\lst O$ is empty then $\lst S:=\emptyset$ and $\set G:=\emptyset$ trivially satisfy all properties.

Now assume that the proposition holds for a sequence of length $n-1$, and consider a sequence of operators $\lst O$ of length $n$.  By the inductive hypothesis, we know that there is a sequence $\lst S'$ and a set $\set G'$ satisfying the properties above for the subsequence of $\lst O$ consisting of the first $n-1$ operators.  Let $o:=\lst O_n\cdot \prod_{g\in \set G, \{\lst O_n,g\}=0} \text{conj}_{\set G}(g)$ --- that is, the product of $\lst O_n$ with the conjugal partner of every operator in $\set G$ with which $\lst O_n$ anti-commutes.  This definition guarantees that $o$ commutes with every operator in $\set G$;  furthermore, we can obtain $\lst O_n$ back from $o$ since every operator in $\set G$ squares to the identity and thus $\lst O_n=o\cdot \prod_{g\in \set G, \{\lst O_n,g\}=0} \text{conj}_{\set G}(o)$; therefore we conclude that $\genfun\paren{\{\lst S'_i\} \cup \set G' \cup \{o\}}=\genfun\paren{\{\lst O_i\}}$.

If $o$ commutes with every operator in $\lst S'$, then set
$$\lst S_i :=
\begin{cases}
\lst S'_i & i \le n-1 \\
o & i = n
\end{cases}
$$
and $\set G := \set G'$, and we are done.  Otherwise, let $s$ be some operator in $\lst S'$ that anti-commutes with $o$, $\set G:=\set G'\cup \{s,o\}$
\footnote{Observe that neither $o$ nor $s$ can be present in $\set G'$ since they commute with every operator in $\set G'$, so the new set $\set G:=\set G'\cup \{s,o\}$ gives us a strictly larger set.  This fact is irrelevant far as the proof is concerned, but it has the important consequence that a computer code implementing the algorithm described by this proof can append $s$ and $o$ to a list of gauge operators and assume that this list continues to form a set (i.e., a sequence without duplicates) without having to explicitly check for this.}, $\lst S_i'' := f(\lst S'_i)$, and $\lst S$ be the subsequence of $\lst S''$ with the identity operators removed, where
$$
f(s') :=
\begin{cases}
s'\cdot s & \{s',o\}=0\\
s' & \text{otherwise}
\end{cases}.
$$
Observe that by this definition, all of the operators in $\lst S$ commute with every operator in $\set G$, so property \ref{stabs-commute-with-G} is satisfied.  Since the only difference between $\set G'$ and $\set G$ is the addition of $s$ and $o$, which form a conjugal pair with respect to $\{\lst S_i\} \cup \set G$, we conclude that property \ref{conjugal-pairs-commute-with-SAG} is satisfied.
Lastly, since $s\in \set G$, we can form any operator in $\lst S'$ with products of operators in $\lst S$ and $\set G$, so therefore $\genfun\paren{\{\lst S_i\} \cup \set G}=\genfun\paren{\{\lst S'_i\} \cup G' \cup \{s,o\}}=\genfun\paren{\{\lst O_i\}}$, and so the final property is satisfied.

We conclude by noting that since all of the operators in $\lst S$ and $\set G$ were formed from products of operators in $\lst O$, which are Pauli operators (i.e., members of the group $\pauligroup$), they are Pauli operators themselves.
\end{proof}
\begin{remark}
A consequence of not requiring independence of the operators in $\lst O$ is that the operators $\lst S$ given by Proposition \ref{proposition-SG} are not necessarily independent.  Happily, since all of these operators can be expressed as tensor products of Pauli operators, we can construct a set of independent operators by performing an analog of Gaussian elimination.
\end{remark}

\begin{proposition}
\label{make-independent-using-elimination}
Suppose that we have been given a sequence of Pauli operators which commute with each other, $\lst R$.  Then there exists
\begin{enumerate}
\item a sequence $\lst S$ of $n$ independent operators such that $\genfun\paren{\{\lst S_i\}}=\genfun\paren{\{\lst R_i\}}$,
\item a sequence of $n$ integers without duplicates in the inclusive range $1\dots n$,
\item and a map $p:\{1\dots n\} \to \{0,1\}$ such that $\lst S_i$ is the only operator in $\lst S$ that anti-commutes with $P_{k_i}^{[p(i)]}$, where $P_k^{[0]}:=X_k$ and $P_k^{[1]}:=Z_k$.
\end{enumerate}
\end{proposition}

\begin{proof}
Proof by induction.  For the base case, we observe that if $\lst R$ is empty, then the trivial sequences $\lst S:=\emptyset$ and $\lst k :=\emptyset$ and the trivial function $p:\emptyset\to\emptyset$ satisfy the requirements.

Now suppose that we know the proposition holds for sequences of length $N-1$, and we are given a sequence $\lst S$ of length $N$.  By our inductive hypothesis, we can apply the proposition to the first $N-1$ operators in $\lst R$ obtain sequences $\lst S'$ and $\lst k'$ of length $n-1$\footnote{Note that $n\ne N$ in general, since some of the first $N-1$ operators might not have been independent.}, and a map $p':\{1\dots n-1\}\to \{0,1\}$ which all satisfy the respective properties of the theorem.  Let $$s:=\lst R_N\cdot \prod_{i=1\dots n-1, \,\,\left\{\lst R_N,P_{k'_i}^{[p(i)]}\right\}=0} \lst S'_i.$$  We know that $s$ commutes with every operator in $\lst S'$ because both $s$ and every operator in $\lst S'$ are equal to products of operators in $\lst R$, which all commute with each other.  Furthermore, since $s$ is a product of $\lst R_N$ and a factor of $\lst S'_i$ for every $i$ such that $\lst R_N$ and $P_{k'_i}^{[p'(i)]}$ anti-commute, and we know that $\lst S_i'$ is the only operator in $\lst S'$ that anti-commutes with $P_{k'_i}^{[p'(i)]}$ for $i=1\dots n-1$, it is therefore the case that $s$ commutes with every member of the set $\{P_{k'_i}^{[p'(i)]}\}_{i=1\dots n-1}$.  Finally, since $s$ is a product of $\lst R_N$ and operators in $\lst S'$, we can obtain $\lst R_N$ entirely from products of operators in $\{\lst S'_i\} \cup \{s\}$, and so $\genfun\paren{\{\lst S'_i\} \cup \{s\}}=\genfun\paren{\{\lst R_i\}}$.

If $s$ is the identity operator, then let $\lst S:=\lst S'$ and $p:=p$ and we are done.  Otherwise, we shall now show that there must exist integers $j\in\{1,\dots,N\}\backslash\{\lst k'_i\}$ and $l\in\{0,1\}$ such that $s$ anti-commutes with $P_{j}^{[l]}$, by demonstrating that if this were not the case then $s$ would have to anti-commute with some element in $\lst S'$, leading to a contradiction.

Assume that $s$ commutes with every operator in the set $\left\{P_j^{[l]}:\quad j\in\{1,\dots,N\}\backslash\{\lst k'_i\}, \quad l\in\{0,1\}\right\}.$  Recalling that $s$ is a member of the Pauli group and thus a tensor product of single-particle Pauli spin matrices, and also that $s$ commutes with every member of the set $\{P_{\lst k'_i}^{[p'(i)]}\}_{i=1\dots n-1}$, we see therefore that $s$ must be a product of elements from this set --- that is, there is some subset $\emptyset \ne \set F \subseteq \{P_{\lst k'_i}^{[p'(i)]}\}_{i=1\dots n-1}$ such that $s=\prod_{o\in \set F} o$.  However, from our inductive hypothesis we know that for every operator $f\in\set F$ there is an operator $s'\in\lst S'$ that anti-commutes with $f$ but commutes with the operators in $\set F\backslash\{f\}$.  Since $s$ is therefore a product of a single operator that anti-commutes with $s'$ and more operators that commute with $s'$, we conclude that $s$ and $s'$ anti-commute, which contradicts our earlier conclusion that $s$ commutes with every operator in $\lst S'$.

Now that we have shown that there exist integers $j\in\{1,\dots,N\}\backslash\{\lst k'_i\}$ and $l\in\{0,1\}$ such that $s$ anti-commutes with $P_{j}^{[l]}$, in terms of these integers we define
$$
\begin{aligned}
\lst S_i &:= 
\begin{cases}
\begin{cases}
\lst S'_i \cdot s & \{\lst S_i',P_j^{[l]}\}=0 \\
\lst S'_i & \text{otherwise}
\end{cases} & 1\le i\le n-1 \\
S' & i=n
\end{cases}, \\
\lst k_i &:=
\begin{cases}
\lst k'_i & 1 \le i \le n-1 \\
j & i=n
\end{cases},\quad \text{and} \\
p(i) &:=
\begin{cases}
p'(i) & 1 \le i \le n-1\\
l & i=n
\end{cases},
\end{aligned}
$$ and we are done.
\end{proof}
\begin{remark}
Proposition \ref{make-independent-using-elimination} is good for more than computing an independent set of generators from a commuting list of operators;  it is also the key ingredient in computing the logical qubit operators.
\end{remark}

\begin{proposition}
\label{construction-of-logicals}
Suppose that we have been given the objects described in 1-3 of Proposition \ref{make-independent-using-elimination}.  Let $\set S := \{\vec S_i\}_i.$  Then there exists a set of operators $\set L$ such that
\begin{enumerate}
\item \label{L-are-independent} the operators in $\set S\cup\set L$ are independent;
\item \label{L-are-conjugal-pairs} every operator in $\set L$ is a member of a conjugal pair with respect to $\set S\cup\set L$;
\item \label{L-completes-the-generators} $\genfun\paren{\set S\cup\set L}=\centralizer_\pauligroup(\set S)$ --- that is, the set generated by $\set S\cup\set L$ is equal to the set of Pauli operators that commute with $\set S$.
\end{enumerate}
\end{proposition}

\begin{proof}
Recalling that $n$ is the number of elements in $\lst S$ (and $\set S$), let $\lst l$ be some sequential ordering of $\{1 \dots N\}\backslash\{\vec k_i\}_i$, and then let $\set L:=\{\lst A_i\}_i\cup\{\lst B_i\}_i$ where
$$
\begin{aligned}
\lst A_i &:= P_{\lst l_i}^{[1]}\cdot \prod_{j=1\dots n\atop \{P_{l_i}^{[1]},\lst S_j\}=0} P_{\lst k_j}^{[s(j)]},\\
\lst B_i &:= P_{\lst l_i}^{[0]}\cdot \prod_{j=1\dots n\atop \{P_{l_i}^{[0]},\lst S_j\}=0} P_{\lst k_j}^{[s(j)]}.\\
\end{aligned}
$$

To see that property \ref{L-are-independent} is satisfied, observe the following.  First, the operators in $\set L$ are independent from the operators in $\set S$ since none of them is the identity operator and they all commute with every operator in $\{P_{\lst k_i}^{[s(i)]}\}_{i=1 \dots n}$.  Second, they are independent from each other since for every $i=1 \dots |\lst l|$ we have that $\vec A_i$ is the only operator that anti-commutes with $P_{\lst l_i}^{[0]}$ and $\lst B_i$ is the only operator that anti-commutes with $P_{\lst l_i}^{[1]}$.  Thus we conclude that all of the operators in $\set S\cup\set L$ are independent.

Next, to see that property \ref{L-are-conjugal-pairs} holds, observe that for every choice of operators $\lst A_i$ and $\lst S_j$ we have (by intentional construction) that $\lst S_j$ either anti-commutes with two of the operators in the product forming $\lst A_i$ or none at all, and so $[\lst S_i,\lst A_j]=0$ for all $i=1\dots n$ and $j=1\dots |\lst l|$;  by the same reasoning we see also that $[\lst S_i,\lst B_j]=0$ for all $i=1\dots n$ and $j=1\dots |\lst l|$.  Furthermore, each operator $\lst A_i$ commutes with every operator in $\set L$ except for its conjugal partner $\lst B_i$, since the only factor in $\lst A_i$ that could anti-commute with a factor contained within another operator in $\set L$ is $P_{l_i}^{[1]}$, and $\lst B_i$ is the only operator in $\set L$ that contains a factor $P_{l_i}^{[0]}$ that anti-commutes with $\lst X_{l_i}$;  reversing this argument, we also see that $\lst B_i$ commutes with every operator in $\set L$ except for $\lst A_i$.  Thus, every operator in $\set L$ is a member of a conjugal pair with respect to $\set L\cup\set S$.

Finally, to see that property \ref{L-completes-the-generators} holds, observe that since the operators in $\set S$ commute they can therefore be simultaneously diagonalized, which means that there is an automorphism on $\pauligroup$ that takes $\lst S_i\mapsto P_i^{[1]}$ for every $i=1 \dots n$.  The only operators that commute with every such $P_i^{[1]}$ are those which do not contain any factor of $P_i^{[0]}$ for $i=1 \dots n$, and so $\centralizer_\pauligroup\paren{\{P_i^{[0]}\}_{i=1\dots n}} = \genfun\paren{\{P_i^{[1]}\}_{i=1 \dots n}\cup \{P_i^{[l]}\}_{i=n+1 \dots N, \,\, l=0,1}}$, which has $2N-n$ generators.  Since the automorphism preserves the number of generators in the centralizer, we thus conclude that $\centralizer_\pauligroup(\set S)$ has exactly $2N-n$ generators.  Since $\set S\cup\set L$ contains independent operators which commute with every member of $\set S$, and furthermore $|\set S\cup\set L|=2N-n$, we thus conclude that $\genfun\paren{\set S\cup\set L}=\centralizer_\pauligroup(\set S)$.
\end{proof}
With these building blocks in place, we now prove the main theorem:

\begin{proof}[Proof of Theorem \ref{theorem-SG}]
By Proposition \ref{proposition-SG}, we know that there exists a list of operators $\lst S$ and a set of independent operators $\set G$ satisfying the properties that are listed there.  By Proposition \ref{make-independent-using-elimination}, we know that there is an independent set of operators $\set S$ that generate the same subgroup as $\lst S$.  

Now let $\set F$ be a maximal subset of commuting operators in $\set G$ --- i.e., for each conjugal pair in $\set G$ take one of the two operators --- and then let $\set O := \set F \cup \set S$.  Since all of the operators in $\set O$ commute, we apply Proposition \ref{make-independent-using-elimination} again to conclude the existence of the objects listed there, and then we immediately apply Proposition \ref{construction-of-logicals} to show that a set $\set M$ exists with the properties listed there.  We are not done yet, however, since there might be operators in $\set G$ with which operators in $\set M$ anti-commute, so we let
$$\set L := \{m\cdot\prod_{f\in \set F\atop \{M,\text{conj}_{\set G}(f)\}=0} f :\quad m \in\set M\}$$
where $\text{conj}_{\set G}(F)$ is the conjugal partner of $F$ in the set $\set G$.  This guarantees that the operators in $\set L$ commute with every operator in $\set S\cup\set G$, and so we are done.
\end{proof}
\begin{remark}
A pseudo-code representation of the algorithm described by Theorem \ref{theorem-SG} is given in Table \ref{table:algo1}.
\end{remark}

\begin{table}
\begin{codebox}
\Procname{$\proc{Compute-Subsystem-Code}(\lst O)$}
\li $\lst S \gets []$
\li $\lst G\gets []$
\li \For $o \gets \lst O$ \label{li:csg-main-loop-start}
\li \Do
\li      \For $(g_X, g_Z) \gets \lst G$ 
\li      \Do
\li        \kw{if}  $\func{anti}(o,g_X)$ \kw{then} $o \gets o \cdot g_Z$
\li        \kw{if} $\func{anti}(o,g_Z)$ \kw{then} $o \gets o \cdot g_X$
          \End 
\li      \kw{if} \text{$o$ is identity} \kw{then} \Goto \ref{li:csg-main-loop-start}
\li      \For $s \gets \lst S$
\li      \Do
\li        \kw{if} $\func{anti}(o,s)$ \kw{then} \Goto \ref{li:make-into-gauge}
          \End
\li      \Goto \ref{li:csg-main-loop-start}
\li      $\lst G \gets \lst G \cup [(o,s)]$ \label{li:make-into-gauge}
\li      $i \gets 1$ 
\li      \For $s' \gets \lst S$ \label{li:kill-redundant-stabs}
\li      \Do
\li            \kw{if} $s'=s$ \kw{then} \Goto \ref{li:kill-redundant-stabs}
\li            \If $\func{anti}(s',o)$
\li            \Then
\li                  $\lst S[i] \gets s' \cdot s$
\li            \Else
\li                  $\lst S[i] \gets s$
                 \End
\li            $i \gets i+1$
             \End
\li        delete $\lst S[i\dots|\lst S|]$
      \End
\li $\lst I \gets []$
\li $\lst P \gets []$
\li \kw{call} $\proc{Gaussian-Elimination}(\lst S,1,\lst I,\lst P)$ \emph{(Table \ref{table:gaussian-elimination})}
\li $\lst T \gets \lst S \cup [g_X | (g_X,g_Z) \in \lst G]$
\li \kw{call} $\proc{Gaussian-Elimination}(\lst T,|\lst S|+1,\lst I,\lst P)$ \emph{(Table \ref{table:gaussian-elimination})}
\li $\lst L \gets []$
\li \For $i\gets 1$ \To $\text{number of physical qubits}$ \label{li:compute-logicals-loop}
\li \Do
\li     \kw{if} $i\in\lst I$ \kw{then} \Goto \ref{li:compute-logicals-loop}
\li     $l_X \gets X_i$
\li     $l_Z \gets Z_i$
\li     \For $(j,p,t)\gets (\lst I,\lst P,\lst T)$
\li     \Do
\li        \If $p=0$
\li        \Then
\li            \kw{if} $\func{anti}(t,X_j)$ \kw{then} $l_X \gets l_X \cdot Z_j$
\li            \kw{if} $\func{anti}(t,Z_j)$ \kw{then} $l_Z \gets l_Z \cdot Z_j$
\li        \Else
\li            \kw{if} $\func{anti}(t,X_j)$ \kw{then} $l_X \gets l_X \cdot X_j$
\li            \kw{if} $\func{anti}(t,Z_j)$ \kw{then} $l_Z \gets l_Z \cdot X_j$
             \End
          \End 
\li     \For $(g_X,g_Z)\in\lst G$
\li     \Do
\li         \kw{if} $\func{anti}(l_X,g_Z)$ \kw{then} $l_X \gets l_X \cdot g_X$
\li         \kw{if} $\func{anti}(l_Z,g_Z)$ \kw{then} $l_Z \gets l_Z \cdot g_X$
          \End
      \End
\li \Return $(\lst S,\lst G,\lst L)$
\end{codebox}
\caption[Algorithm \proc{Compute-Subsystem-Code}]{Algorithm which computes the subsystem code generated by a given list of measurement operators $\lst O$.  The subroutine \textsc{Gaussian-Elimination} is listed in Table \ref{table:gaussian-elimination}.} \label{table:algo1}
\end{table}
\begin{table}
\begin{codebox}
\Procname{$\proc{Gaussian-Elimination}(\lst S,i,\lst I,\lst P)$}
\li \While $i < |\lst S|$ \label{li:while-loop}
\li \Do
\li   $s \gets \lst S[i]$
\li   \For $j \gets 0$ \kw{to} $i-1$
\li   \Do
\li     $(n,z)\gets (\lst I[j],\lst P[j])$
\li     \If $z = 0$
\li     \Then
\li       \If $\func{anti}(s,X_n)$
\li       \Then $s \gets s \cdot \lst S[j]$
          \End
\li     \Else
\li       \If $\func{anti}(s,Z_n)$
\li       \Then $s \gets s \cdot \lst S[j]$
          \End
        \End
      \End
\li   \If $s$ is identity
\li   \Then
\li     delete $\lst S[i]$
\li     \Goto \ref{li:while-loop}
      \End
\li   \For $n \gets 0$ \kw{to} number of physical qubits \label{li:next-physical-qubit}
\li   \Do
\li     \kw{if} $n\in\lst I$ \kw{then} \Goto \ref{li:next-physical-qubit}
\li     \If $\func{anti}(s,X_n)$
\li     \Then
\li         $z\gets 0$
\li         \Goto \ref{li:found-the-pauli}
        \End
\li     \If $\func{anti}(s,Z_n)$
\li     \Then
\li         $z\gets 1$
\li         \Goto \ref{li:found-the-pauli}
        \End
      \End
\li   \If $z = 0$ \label{li:found-the-pauli}
\li   \Then
\li     \For $j \gets 0$ \kw{to} $i-1$
\li     \Do
\li       \If $\func{anti}(\lst S[j],X_n)$
\li       \Then $\lst S[j] \gets \lst S[j] \cdot s$
          \End
        \End
\li   \Else
\li     \For $j \gets 0$ \kw{to} $i-1$
\li     \Do
\li       \If $\func{anti}(\lst S[j],Z_n)$
\li       \Then $\lst S[j] \gets \lst S[j] \cdot s$
          \End
        \End
      \End
\li   append $n$ to $\lst I$
\li   append $z$ to $\lst P$
\li   $\lst S[i] \gets s$
\li   $i \gets i + 1$
    \End
\end{codebox}
\caption[Algorithm \proc{Gaussian-Elimination}]{Subroutine which performs a procedure analgous to Gaussian-elimination on $\vec S$ to distill a set of independent operators from a possible dependent set of operators.  This subroutine is called by the main subsystem code algorithm in Table \ref{table:algo1}.}
\label{table:gaussian-elimination}
\end{table}
\subsection{Optimization of the logical qubits}

\begin{remark}
A pseudo-code representation of the algoritm that will be described in this section is presented in Table \ref{table:algorithm-optimization}.
\end{remark}

In general there are multiple sets of operators that satisfy the properties of \ref{theorem-SG}, as is illustrated by the following Lemma:

\begin{lemma}
\label{combining-pairs}
Given conjugal pairs $Q:=(a,b)$ and $R:=(c,d)$ in relation to some set $\set X$ such that either $a\ne c$ or $b\ne d$, we have that
\begin{enumerate}
\item the pairs $Q':=(a\cdot c,b)$ and $R':=(c,d\cdot b)$ are conjugal pairs with respect to $\set X \backslash \{Q,R\} \cup \{Q',R'\}$; and
\item $\genfun\paren{\{a,b,c,d\}}=\genfun\paren{\{a\cdot c,b,c,d\cdot b\}}$.
\end{enumerate}
\end{lemma}

\begin{proof}
\begin{enumerate}
\item Since $[a,c]=[a,d]=[b,c]=[b,d]=\{a,b\}=\{c,d\}=0$, we see therefore that $[a\cdot c,c]=[a\cdot c,d\cdot b]=[b,c]=[b,d\cdot b]=\{a\cdot c,b\}=\{c,d\cdot b\}=0$.  Furthermore, since $a$, $b$, $c$ and $d$ commute with every operator in $\set X\backslash \{Q,R\}$, so do $a\cdot c$ and $d\cdot b$.
\item Since $b$ and $c$ are Pauli operators and thus square to the identity, we have that $a\cdot c\cdot c=a$ and $d\cdot b\cdot b=d$, and so $\genfun\paren{\{a,b,c,d\}}=\genfun\paren{\{a\cdot c,b,c,d\cdot b\}}$.
\end{enumerate}
\end{proof}
As a result of this lemma, we see that we can take pairs of arbitrary conjugal pairs from sets $\set G$ and $\set L$ of Theorem \ref{theorem-SG} and replace them with different pairs per the recipe in Lemma \ref{combining-pairs} such that the properties of the theorem still hold.  So given that these sets are not unique, the natural question is:  What is the best choice of $\set G$ and $\set L$?  To answer this, we observe that another criteria we would like for our code to satisfy is that it be as robust to errors as possible;  in particular, we seek to maximize the difficulty of \emph{undetectable errors}, which is defined as follows:

\begin{definition}
Given a set $\set S\subseteq\pauligroup$ and operators $l\in\pauligroup$ and $e\in\centralizer_\pauligroup(\set S)$ which anti-commute (i.e., $\{l,e\}=0$), we say that $e$ is an \emph{undetectable error with respect to} $\set S$ \emph{acting on} $l$.
\end{definition}
We assume that the `difficulty' of an interaction between our physical system and its environment is related to the number of physical qubits in our system that are participating in the interaction.  Thus, the natural metric for measuring the relative difficulty of an error is given by its weight, which recall is defined as follows:

\begin{definition}
Given an operator $p\in\pauligroup$---which recalls means that $p$ is the tensor product of single-qubit Pauli unnormalized spin matrices---the \emph{weight} of $p$ is the number of single-qubit operators in the product which are non-trivial (i.e., not the identity).  So for example, the weight of $I\otimes I\otimes I$ is 0, the weight of $I\otimes Z\otimes I\otimes X$ is 2, and the weight of $Z\otimes X\otimes Y$ is 3.
\end{definition}
For convenience, we introduce the following additional notation:

\begin{definition}
$\quad$

\begin{itemize}
\item the function $\Pi:\powerset(\pauligroup)\to\pauligroup$ is defined such that $\Pi(\set X):=\prod_{x\in \set X} x$ --- that is, it is the product of the operators in $\set X$.
\item assuming we have a set of independent operators, $\set Q$, the function $\set G_{\set Q}:\genfun(\set Q)\to\powerset(\set Q)$ is defined (uniquely) such that for every $q\in\set Q$ we have that $q=\prod_{o\in\set G_{\set Q}(q)} o$;
\item the function $w:\pauligroup\to \mathscr{N}$ is defined such that $w(o)$ gives the weight of $o$;
\item the function $e_{\set S}:\centralizer_\pauligroup\paren{\set S}\to \paren{\powerset\circ\centralizer_\pauligroup}\paren{\set S}$ is defined such that $\set e_{\set S}(l)$ is the set of minimizers of $w$ over the set $\left\{o: o\in \centralizer_\pauligroup\paren{\set S}, \{o,l\}=0\right\}$ --- that is, it gives the undetectable errors with respect to $\set S$ acting on $l$ that are of minimum weight;
\item the function $\om_{\set S}:\centralizer_\pauligroup\paren{\set S}\to\mathscr{N}$ is defined such that $\om_{\set S}:=w(o)$ for an arbitrarily chosen $o\in\circ \set e_{\set S}$ --- note that function is well-defined since all operators in the set $\set e_{\set S}$ have the same weight;
\item the function $m_{\set S}:\centralizer_\pauligroup\paren{\set S}\times \centralizer_\pauligroup\paren{\set S} \to \mathscr{N}$ is defined such that $m_{\set S}(l,l'):=\min \{\om_{\set S}(l),\om_{\set S}(l')\}$ --- that is, it gives the smaller of the weights of the smallest weight errors acting on respectively $\set L$ and $\set L'$;
\item the function $\lst M_{\set S}$ is defined such that $\lst M_{\set S}\paren{\lst P}$ is the sequence of $|\lst P|$ integers such that $\lst M_{\set S}\paren{\lst P}_i := m_{\set S}\paren{\lst P_i}$;
\item the functions $p_1$ and $p_2$ are defined such that, given $(a,b):=x$, we have that $p_1(x):=a$ and $p_2(x):=b$.
\item the function $\set U:\powerset(\pauligroup\times\pauligroup)\to\powerset(\pauligroup)$ is defined (for convenience) such that $\set U\paren{\set X}:=\bigcup_{x\in\set X} \{p_1(x),p_2(x)\}$ --- that is, it `unpacks' a set of pairs of operators into a set of operators;  in an abuse of notation, we shall also allow $\set U$ to apply to sequences, so that $\set U(\lst P) := \set U\paren{\{\lst P_i\}_i}$, and to individual pairs, so that if $X$ is a single pair then $\set U(X) := \set U(\{X\})$;
\item finally, a \emph{choice of qubits stabilized by $\set S$}, $\lst P$, is a sequence of pairs of operators from the Pauli group such that 
\begin{enumerate}
\item no operator in $\set U(\lst P)$ appears in more than one pair in $\lst P$;
\item $\set U(\lst P)\subseteq \centralizer_\pauligroup(\set S)$
\item every pair in $\lst P$ is a conjugal pair with respect to $\set S \cup \set U(\lst P)$;
\item $(\om_{\set S}\circ p_1)(\lst P_i)=m_{\set S}(\lst P_i)$; and
\item $\lst M_{\set S}(\lst P)$ is an ordered sequence.
\end{enumerate}
\end{itemize}

\end{definition}
Given the notation above, we now precisely define what we mean by the ``best choice'' of logical qubits.

\begin{definition}
An \emph{optimal choice of qubits stabilized by $\set S$} is any choice of qubits, $\lst P$ stabilized by $\set S$, such that given any other choice of qubits, $\lst P'$, that is also stabilized by $\set S$ and which satisfies $(\genfun\circ \set U)(\lst P)=(\genfun\circ \set U)(\lst P')$, we have that $\lst M(\lst P)_i \ge \lst M(\lst P')_i$ for all $1\le i \le |\lst P|=|\lst P'|$.\footnote{Note that since $\lst P$ and $\lst P'$ are sequences of conjugal pairs without duplicates they are therefore independent, and so if $(\genfun\circ \set U)(\lst P)=(\genfun\circ \set U)(\lst P')$ then we know automatically that $|\lst P|=|\lst P'|$.}
\end{definition}
We now present an algorithm for computing the optimal choice of logical qubits from a set of input qubits.  The key insight upon which the algorithm is built is that undetectable errors acting on the space of logical qubits can never be eliminated entirely, so there will always be \emph{some} operator on which they act.  Thus, the goal of the optimization procedure is not to eliminate errors, but rather to \emph{contain} them, so that they act on as few operators as possible.

The optimization algorithm works by starting with an empty (and therefore automatically optimal) choice of qubits and a set of `unoptimized' qubits, and making progress by gradually moving qubits from the unoptimized set into the choice in such a way that preserves the optimality of the choice.  The trick is that we want to delay as long as possible moving a qubit into the choice, until we have had every chance to improve it.  Thus, we additionally keep track of a subset of pairs in the choice whose second members have yet to be used to contain an error, and then use them as much as possible to fix errors.  That is, at every step in the algorithm, we scan for the minimal weight undetectable error acting on the set of operators consisting of both the second member of the pairs in this subset and all of the operators in the unoptimized set of qubits.  If the minimal weight error acts on an operator in the first category, then we remove the pair from the subset and use this operator to fix this error wherever it occurs in both the second members of pairs in the subset and the unoptimized qubits.  Otherwise, we pull out a qubit from the unoptimized set on which the error acts, use the first member in the pair to fix the error in the qubits remaining in the unoptimized subset, add the pair to the subset of qubits whose second members have yet to be used to contain an error, and then add it to the end of the choice.  At this point our choice turns out to still be optimal because if there had been a way to make the qubit we just added any better by recombining it with other qubits in the choice then we would have already done so by now.

This procedure is presented formally by means of the following inductive definition.

\begin{definition}
Let the function $\optimizer$ be a map from a tuple of the form tuple $(\set S,\set L)$ to a sequence of tuples each of the form $(\set Q,\lst P,\lst s)$, where
\begin{itemize}
\item $\set S$ is a set of commuting Pauli operators;
\item $\set L$ is a set of Pauli operators that are conjugal in relation to $\set L\cup\set S$;
\item $\set Q$ is a set of pairs of Pauli operators;
\item $\lst P$ is a sequence of pairs of Pauli operators; and
\item $\lst s$ is a sequence of integers from $\{0,1\}$ with the same length as $\lst P$.
\end{itemize}

The sequence is defined inductively.  For convenience, we let the first index of this sequence be zero, and define $\optimizer(\set S,\set L)_0=(\set L',\lst\emptyset,\lst\emptyset)$, where $\set L'$ is the set of pairs such that $\set U(\set L')$ and no operator appears in more than one pair in $\set L'$, and $\lst\emptyset$ is the empty sequence.  Now assume that $\optimizer(\set S,\set L)_i$ is defined and that $\optimizer(\set S,\set L)_i=(\set Q,\lst P,\lst s)$.  If $\set Q$ is the empty set, then $\optimizer(\set S,\set L)_i$ is the last element of the sequence so that $|\optimizer(\set S,\set L)|=i+1$.  Otherwise, $\optimizer(\set S,\set L)_{i+1}:=(\set Q',\lst P',\lst s'),$ where $\set Q'$, $\lst P'$ and $\lst s'$ are defined as follows.

Let $\set R:=\{p_2(\lst P_i):i\in\{1\dots |\lst P|\},\lst s'=1\}$ and $\set X:=\set U(\set Q)\cup\set R$.  Note that since $\set Q\ne\emptyset$ that therefore $\set X\ne\emptyset$.  Let $h$\footnote{An observant reader may have noticed that we do not specify exactly how one goes about computing $h$.  This was an intentional omission since the details are quite technical and fortunately they are irrelevent for proving that this algorithm works correctly as long as we can assume that $h$ \emph{can} be computed.  Thus, the discussion of how to compute $h$ will be deferred until Section \ref{subsubsection:running time analysis} when we analyze bounds on the running time of the algorithm.} be the minimal weight error with respect to $\set S$ acting on any operator in $\set X$.  There are two cases: either $h$ acts on some operator in $\set R$, or it doesn't and so must act on some operator in $\set U(\set Q)$.  The definition of $\set Q'$, $\lst P'$ and $\lst a'$ depends on which of these cases holds.

\begin{center}
\textbf{Case 1:} $h$ acts on some operator in $\set R$
\end{center}

Let $k$ be the smallest index such that $h$ acts on $p_2(\lst P_k)$, and $o:=p_2(\lst P_k)$.  Define
$$
f(x) :=
\begin{cases}
x & [h,x] = 0 \\
x\cdot o & \{h,x\} = 0, \\
\end{cases}
$$
and
$$
g(a,b) :=
\begin{cases}
I & [h,a] = 0 \\
b & \{h,b\} = 0, \\
\end{cases}
$$
Let $s':=p_1(\lst P_k)\cdot \alpha \cdot \beta$ where
$$\alpha := \Pi\paren{\{g(p_2(\lst P_i),p_1(\lst P_i)\}:i\in\{k+1\dots |\lst P|\},\lst a_i=0}$$
and
$$\beta := \Pi\paren{\{g(a,b)\cdot g(b,a):(a,b)\in\set Q\}}.$$

Then we define
$$
\begin{aligned}
\set Q'&:= \{(f(a),f(b)):(a,b)\in\set Q\},\\
\lst P'_i&:=
\begin{cases}
\lst P_i & i < k \\
(a',o) & i = k \\
\lst P_i & i > k, \lst a_i=0 \\
\lst P_i & i > k, [h,p_2(\lst P_i)]=0 \\
(p_1(\lst P_i),p_2(\lst P_i)\cdot o) & \text{otherwise} \\
\end{cases}\\
\lst s'_i&:=
\begin{cases}
\lst s_i & i\ne k\\
0 & i =k
\end{cases}
\end{aligned}
$$

\begin{center}
\textbf{Case 2:} $h$ does not act on some operator in $\set R$
\end{center}

Let $q\in\set Q$ be a pair such that $h$ acts on one of its members, and without loss of generality assume that $h$ acts on the first member since otherwise we can swap the members of the pair.  Let $o:=p_1(q)$.
Define
$$
f(x) :=
\begin{cases}
x & [h,x] = 0 \\
x\cdot o & \{h,x\} = 0, \\
\end{cases}
$$
and
$$
g(a,b) :=
\begin{cases}
I & [h,a] = 0 \\
b & \{h,b\} = 0, \\
\end{cases}
$$
Let $$b'':=p_2(q)\cdot\Pi\paren{\{g(a,b)\cdot g(b,a):(a,b)\in\set Q\backslash\{q\}\}},$$ and
$$b':=
\begin{cases}
b'' & [h,b'']=0 \\
b''\cdot o & \{h,b''\}=0 \\
\end{cases}
$$

Then we define
$$
\begin{aligned}
\set Q'&:= \{(f(a),f(b)):(a,b)\in\set Q\backslash\{q\}\},\\
\lst P'_i&:=
\begin{cases}
\lst P_i & i \le |\lst P| \\
(o,b') & i = |\lst P|+1
\end{cases} \\
\lst s'_i&:=
\begin{cases}
\lst s_i & i \le |\lst s|\\
1 & i = |\lst s|+1
\end{cases}
\end{aligned}
$$
\end{definition}

Table \ref{table:algorithm-optimization} contains a listing of pseudo-code that uses the above algorithm to compute the optimal choice of qubits.  For the sake of completeness, it includes additional steps that pertain to the details of how the minimal weight operator is computed, which will be discussed in more detail in Section \ref{subsubsection:running time analysis}.

\begin{table}
\begin{codebox}
\Procname{$\proc{Optimize-Logical-Qubits}(\lst S,\lst G,\lst L)$}
\li $N \gets |\lst L|$
\li $k \gets 0$
\li $\lst m \gets [\const{true}]*|\lst L|$
\li \kw{nested function} $\proc{Query-Function}(o)$
\li \Do
\li     \For $i \gets 0$ \To $k$, $(l_X,l_Z) \gets \lst L[i]$
\li     \Do
\li         \If $\lst m[i]$ and $\func{anti}(o,l_Z)$
\li         \Then
\li           \Return $(\const{true},(1,i))$
            \End
        \End
\li     \For $i \gets k+1$ \To $N$, $(l_X,l_Z) \gets \lst L[i]$
\li     \Do
\li         \If $\func{anti}(o,l_X)$ or $\func{anti}(o,l_Z)$
\li         \Then
\li           \Return $(\const{true},(2,i))$
            \End
        \End
\li     \Return $(\const{false},\const{undefined})$
    \End
\li $\lst O \gets \func{copy}(\lst S)$
\li \For $(g_X, g_Z) \gets \lst G$
\li \Do
\li     append $g_X$ and $g_Z$ to $\lst O$
    \End
\li \For $(l_X, l_Z) \gets \lst L$
\li \Do
\li     append $l_X$ and $l_Z$ to $\lst O$
    \End
\li $\lst P \gets \proc{Compute-Pseudogenerators}(\lst O)$ \emph{(Table \ref{table:compute-pseudo-generators})}
\li \While $k < N$
\li \Do
\li     $(e,(c,l))\gets$ $\proc{Find-Weight-Minimizer}$ 
        \Indentmore
        \Indentmore
\zi         $(\proc{Query-Function},\lst P)$ \emph{(Table \ref{table:find-weight-minimizer})}
        \End
        \End
\li     $(q_X,q_Z)\gets \lst L[l]$
\li     \If $c = 1$
\li     \Then
\li         $\lst m[l] \gets \const{false}$
\li         \For $i\gets l+1$ \To $k$,  $(l_X,l_Z) \gets \lst L[i]$
\li         \Do
\li             \If $m[i]$ and $\{e,l_Z\}$
\li             \Then
\li                 $q_X \gets q_X \cdot l_X$
\li                 $\lst L[i] \gets (l_X,l_Z\cdot q_Z)$
                \End
            \End
\li         \kw{call} $\proc{Fix-Logical-Qubits}$ 
            \Indentmore
\zi         $(\lst L,k,e,q_Z,q_X)$ \emph{(Table \ref{table:fix-logical-qubits})}
            \End
\li         $L[l] \gets (q_X,q_Z)$
\li     \Else
\li         \If $\func{commute}(e,q_X)$
\li         \Then swap $q_X$ and $q_Z$
            \End
\li         $(q_X,q_Z) \gets \lst L[l]$
\li         $\lst L[l] \gets \lst L[k]$
\li         $k \gets k+1$
\li         \kw{call} $\proc{Fix-Logical-Qubits}$
            \Indentmore
\zi         $(\lst L,k,e,q_X,q_Z)$ \emph{(Table \ref{table:fix-logical-qubits})}
            \End
\li         $L[k] \gets (q_X,q_Z)$
        \End
    \End
\end{codebox}
\caption[Algorithm \proc{Optimize-Logical-Qubits}]{Algorithm which optimizes the logical qubits for a given subsystem code.} \label{table:algorithm-optimization}
\end{table}
\begin{table}
\begin{codebox}
\Procname{$\proc{Fix-Logical-Qubits}(\lst L,k,e,q_A,q_B)$}
\li \For $i\gets k$ \To $|\lst L|$,  $(l_X,l_Z) \gets \lst L[i]$
\li \Do
\li     \If $\func{anti}(e,l_Z)$ and $\func{anti}(e,l_X)$
\li     \Then
\li         $q_A \gets q_A \cdot l_X \cdot l_Z$
\li         $\lst L[i] \gets (l_X\cdot q_B,l_Z\cdot q_B)$
\li     \ElseIf $\func{anti}(e,l_Z)$
\li     \Then
\li         $q_A \gets q_A \cdot l_X$
\li         $\lst L[i] \gets (l_X,l_Z\cdot q_B)$
\li     \ElseIf $\func{anti}(e,l_X)$
\li     \Then
\li         $q_A \gets q_A \cdot l_Z$
\li         $\lst L[i] \gets (l_X\cdot q_B,l_Z)$
        \End
        \End
\li \If $\proc{anti}(e,q_A)$
\li \Then
        $q_A \gets q_A \cdot q_B$
    \End
\li \Return $(q_A,q_B)$
    \End
\end{codebox}
\caption[Algorithm \proc{Fix-Logical-Qubits}]{Algorithm which `fixes' a subset of the logical qubits so that they are robust to a given error.} \label{table:fix-logical-qubits}
\end{table}
In addition to proving that the above algorithm successfully constructs an optimal choice of qubits, we shall also provide a bound on its running time.  In order to do this, we first need to precisely define what we mean by the running time for the purposes of this section.

\begin{definition}
We say that a computation can be performed \emph{in time $x$} if the computation requires taking $x$ products of Pauli operators.
\end{definition}

Of course, the number of products of Pauli operators is not the only metric that could serve as the gauge for the running time, but it suffices for our purposes.
We now present the main result of this section.

\begin{theorem}
\label{theorem:optimization procedure}
Suppose we are given
\begin{itemize}
\item a set of commuting Pauli operators, $\set S$, acting on $N$ physical qubits;
\item a set of pairs, $\set L\subseteq\centralizer_\pauligroup(\set S)$, conjugal with respect to $\set U(\set Q)\cup\set S$;
\item and a set of Pauli operators $\set C$ such that $\genfun(\set C)=\centralizer_\pauligroup(\set S)$;
\end{itemize}
then $\optimizer(\set S,\set L)$ is a sequence of finite length, and if $(\set Q,\lst P,\lst s)$ is the last element in the sequence then $\lst P$ is an optimal choice of qubits such that $\genfun(\lst P)=\genfun(\set L)$, and furthermore it can be computed in a time that is in the set $O\paren{|\set C|^2+|\lst L|^2d3^d\choose{N}{d}}$\footnote{A function $f$ is said to be in the set $O(g)$ if $f$ is asymptotically bounded by some fixed constant times $g$;  formally $f\in O(g)$ if and only if there exists constants $c$ and $x_0$ such that $f(x)<c g(x)$ for all $x>x_0$.}, where $d:=\lst M(\lst P)_{|\lst P|}$\footnote{Recall that $\lst M(\lst P)_{|\lst P|}$ is the distance of the best qubit in the (optimized) code.}.
\end{theorem}

The proof of this Theorem is rather technical and shall be split into several subsections.  First we shall prove the existence of a condition that suffices to prove that a choice of logical qubits is optimal.   Second we shall prove that the algorithm above constructs a choice satisfying this condition.  Third we shall prove that the running time of the algorithm has the claimed bound.  Finally we shall tie these results together to prove the Theorem above.

\subsubsection{Optimality condition}

\label{optimal-generators}

How do we know that a choice of qubits is optimal?  Intuitively, it should be sufficient to prove that a choice of qubits is optimal if we can show that there is no way that we can recombine qubits in the choice to form one or more qubits that are more robust than their component factors --- that is, there is no way that any qubit can be ``improved'' by its involvement in such a product.  This condition is stated formally in the following definition of \emph{unimprovable sets}:

\begin{definition}
An \emph{unimprovable set with respect to $\set S$} is a set of Pauli operators, $\set O$, such that for any subset, $\set X\subseteq \set O$, we have that $(\om_{\set S}\circ\prod)(\set X) = \min_{x\in\set X}\om_{\set S}(x)$.  We say that an unimprovable set $\set O$ \emph{extends to $\lst Q$} if for all subsets $\set X \subseteq \set O\cup\set Q$ such that $\set X\cap \set O \ne \emptyset$ we have that $(\om_{\set S}\circ\prod)(\set X) \le \min_{x\in\set X\cap \set O}\om_{\set S}(x)$.
\end{definition}
The following Theorem is the main result of this subsection that proves that this condition is indeed sufficient to show that a choice of logical qubits is optimal.

\begin{theorem}
\label{theorem:optimality-condition}
If $\lst P$ is a choice of $N$ logical qubits stabilized by $\set S$ such that $\{p_1(\lst P_i)\}_i$ is an unimprovable set with respect to $\set S$ that extends to $\set U(\lst P)$, then $\lst P$ is an optimal choice of qubits.
\end{theorem}

\begin{remark}
The intuition behind the proof of this Theorem is that because the set of first members of pairs is unimprovable and extends to the set of all members of pairs, we know that no qubit can be ``improved'' by recombining it with one or more other qubits.  Thus, the only way one could construct a better choice would by forming $n+k$ independent qubits from products of $n$ qubits (where $k>0$), which intuitively should be impossible.  Thus, we conclude that it is not possible for there to be a choice of qubits generated by the same qubits in this choice that is ``better'' than this choice.

To assist us in proving this Theorem, we shall first prove a number of useful Lemmas and Propositions.  We start with a simple Lemma that proves that taking a product of operators results in an operator that is no ``worse'' (with respect to its robustness to errors) than the worst operator in the product.
\end{remark}
\begin{lemma}
\label{combinations-can't-make-things-worse}
For any set of operators $\set O$, we have that $(\om_{\set S}\circ \prod)(\set O) \ge \text{min}_{o\in \set O}\,\om_{\set S}\paren{o}$.
\end{lemma}

\begin{proof}[Proof of Lemma]
Any undetectable error with respect to $\set S$ acting on $(\om_{\set S}\circ \prod)(\set O)$ must also act at least one of the operators in $\set O$ since otherwise it cannot anti-commute with the product. \end{proof}
\begin{remark}
In general, taking products of operators might result in an operator that is better than the worst operator in the product because errors will cancel each other out --- i.e., if two operators in the product anti-commute with an error then their product commutes with the error.  Thus, it is useful to state a condition under which we can be certain that this will not happen, so that the product is exactly as bad as the worst operator, which we do in the following Lemma.
\end{remark}

\begin{lemma}
\label{lesser-operator-wins}
Suppose we are given two operators $a,b\in\pauligroup$ such that $\om_{\set S}(a)<\om_{\set S}(b)$;  then $\om_{\set S}(a\cdot b) = \om_{\set S}(a)$.
\end{lemma}

\begin{proof}[Proof of Lemma]
Since $\om_{\set S}(a)<\om_{\set S}(b)$, there must be an undetectable error with respect to $\set S$ that acts on $a$ but not on $b$;  thus, it must anti-commute with and hence act on the product $a\cdot b$, so that $\om_{\set S}(a\cdot b)\le \om_{\set S}(a)$.  Since $\om_{\set S}(a\cdot b)\ge \om_{\set S}(a)$ by Lemma \ref{lesser-operator-wins}, we conclude that $\om_{\set S}(a\cdot b) = \om_{\set S}(a)$.
\end{proof}
\begin{remark}
Intuitively we should expect that it is not possible to take $n$ qubits and recombine them to form $n+k$ independent qubits where $k>0$.  To state this intuition in other terms, suppose we are given a set of conjugal pairs $\set C$ that are generated from some other set of conjugal pairs $\set D$.  We know that every pair in $\set C$ must have a member that includes a factor that is a first member of a pair in $\set D$ (since otherwise the members of the pair cannot anti-commute), so let $\set A$ be the set of first members of pairs in $\set D$.  Our intuition then tells us that $|\set C|\le |\set A|=|\set D|$.  The following Proposition states this fact formally:
\end{remark}
\begin{proposition}
\label{bound-on-recombinations}
Suppose we are given
\begin{enumerate}
\item sets of independent Pauli operators $\set Q$ and $\set S$;
\item a non-empty set of conjugal pairs, $\set C$, with respect to $\set U(\set C) \cup \set S$, such that $U(\set C)\subseteq \genfun(\set Q)$; and
\item a set $\set A$ of independent Pauli operators with the property that for any conjugal pair $X:=(a,b)$ such that $\{a,b\}\in\genfun(\set C)$, we must have that $\set G_{\set Q}(X) \cap \set A \ne \emptyset$.
\end{enumerate}
Then $|\set C|\le|\set A|$.
\end{proposition}
\begin{remark}
The basic idea behind the proof of this Proposition is that an analogue of Gaussian elimination can be used on the conjugal pairs to eliminate the presence of members of $\set A$ from them;  when we are done with this process, we can see that unless $|\set C|\le |\set A|$ we would have eliminated \emph{all} members of $\set A$ from some of the qubits, which contradicts the assumptions of this Proposition.

The formal proof is somewhat technical and so we first introduce several Lemmas.  First we prove a small helper Lemma that shows that it is possible to take a conjugal pair in which a given generator appears and rearrange it so that the generator only appears in the first member of the pair.
\end{remark}
\begin{lemma}
\label{single-pair-rearrangement}
Let $A:=(a,b)$ with $\{a,b\}\subseteq\genfun(\set Q)$ be a conjugal pair with respect to some set $\set S$, and $o$ be some Pauli operator such that $o\in\set G_{\set Q}(A)$.  Then there exists a pair $B:=(c,d)$ such that
\begin{enumerate}
\item $\{c,d\}\subseteq\genfun(\set Q)$;
\item $o\in \set G_{\set Q}(c)$;
\item $o\notin \set G_{\set Q}(d)$;
\item $(c,d)$ is a conjugal pair with respect to $\paren{\set S\backslash \{a,b\}}\cup\{c,d\}$; and
\item $(\genfun\circ\set U)(B) = (\genfun\circ\set U)(A)$.
\end{enumerate}
\end{lemma}

\begin{proof}
Let
$$
(c,d):=
\begin{cases}
\paren{a,b} & o\in \set G_{\set Q}(a), o\notin \set G_{\set Q}(b) \\
\paren{b,a} &  o\notin \set G_{\set Q}(a), o\in \set G_{\set Q}(b) \\
\paren{a,b\cdot a} & o\in \set G_{\set Q}(a), o\in \set G_{\set Q}(b) \\
\end{cases}
$$
Note that in any of the above cases, properties 1-3 are satisfied by construction, property 4 is satisfied because $c$ and $d$ are products of $a$ and $b$ which commute with every element in $\set S\backslash \{a,b\}$ and $\{c,d\}=0$, and finally property 5 is satisfied because $\{a,b\}\subseteq \genfun\paren{\{c,d\}}$ and $\{c,d\}\subseteq \genfun\paren{\{a,b\}}$.
\end{proof}
\begin{remark}
This next Lemma contains the heart of this Proposition by introducing an analogue to a directed Gaussian elimination procedure.  Specifically, it shows that if we have a generator $a\in\set A$ that appears in one or more conjugal pairs, then we can take products of the conjugal pairs to eliminate it from appearing anywhere except in the first member of a single pair.
\end{remark}

\begin{lemma}
\label{directed-gaussian-elimination-of-logicals}
In the context of Proposition \ref{bound-on-recombinations}, suppose we are given an element $a\in \set A$ with the property that there exists a pair $Y''\in\set C$ such that $a\in\set G(Y'')$. Then there exists a conjugal pair $Y$ and set of conjugal pairs $\set D$, all with respect to $\set U\paren{\{Y\}\cup \set D} \cup \set S$, such that
\begin{enumerate}
\item $|\set D| = |\set C|-1$
\item $(\genfun\circ\set U)(\{Y\}\cup \set D)=(\genfun\circ\set U)(\set C)$;
\item $a\in (\set G_{\set Q}\circ p_1)(Y)$ but $a\notin (\set G_{\set Q}\circ p_2)(Y)$;
\item $a\notin \bigcup_{D\in \set D} \set G_{\set Q}(D)$; and
\item for every conjugal pair $O\in\set D$, we have that $\set G_{\set Q}(O) \cap \set A\backslash \{a\} \ne \emptyset$.
\end{enumerate}
\end{lemma}

\begin{proof}
Proof by induction on the size of $\set C$.  If $\set C=\{Y''\}$, then apply Lemma \ref{single-pair-rearrangement} letting $o:=a$, $A:=Y''$, and $Y:=B$, and we see that we have a pair $Y$ which is conjugal with respect to $\{Y\}\cup\set S$ and also such that $(\genfun\circ\set U)(\{Y\})=(\genfun\circ\set U)(\{Y''\})$.  Let $\set D:=\emptyset$, and we see that the remaining properties hold trivially, so we are done.

Now let us assume that this lemma has been proven for the case where $|\set C|=n-1$, and we are given a set $\set C$ with $n$ elements.  Take any $X''\in\set C\backslash\{Y''\}$, and apply the lemma to $\set C\backslash \{X''\}$, $\set A$, $a$ and $Y''$ to obtain the objects $Y'$ and $\set D'$ described in this Lemma without the primes.  If $a\notin\set G(X'')$, then by the assumptions of the Lemma we know that $\set G(X'') \cap \set A\backslash \{a\} \ne \emptyset$, so let $Y:=Y'$ and $\set D:=\set D'\cup\{X''\}$, and we are done.

Otherwise, apply Lemma \ref{single-pair-rearrangement}, setting $A:=X''$, $o:=a$, and $X':=B$, and let $X:=\paren{p_1(X')\cdot p_1(Y'),p_2(X')}$ and $Y:=\paren{p_1(Y'),p_2(Y')\cdot p_2(X')}$.  Note $X'$ and $Y'$ are conjugal pairs with respect to $\set U\paren{\{X',Y'\}\cup\set D}\cup\set S$ and $\{X',Y'\}\cap \paren{\set D\cup\set S}=\emptyset$, and so by Lemma \ref{combining-pairs} we conclude that $X$ and $Y$ are conjugal pairs with respect to $\set U\paren{\{X,Y\}\cup\set D}\cup\set S$, and also that $(\genfun\circ\set U)(\{X,Y\})=(\genfun\circ\set U)(\{X',Y'\})$;  since $X'$ was obtained from applying Lemma \ref{single-pair-rearrangement} to $X''$ and $a$, we furthermore conclude that $(\genfun\circ\set U)(\{X,Y\})=(\genfun\circ\set U)(\{X'',Y'\})$.  Since $X'$ was obtained as a result of Lemma \ref{single-pair-rearrangement}, we know that $a\in (\set G \circ p_1)(X')$ but $a\notin (\set G_{\set Q} \circ p_2)(X')$, and we also know from the earlier recursive application of this Lemma that $a\in (\set G \circ p_1)(Y')$ but $a\notin (\set G_{\set Q} \circ p_2)(Y')$.  Thus, we observe that by construction, $a\in (\set G_{\set Q} \circ p_1)(Y)$, and $a\notin \paren{(\set G_{\set Q} \circ p_2)(Y) \cup \set G_{\set Q}(X)}$.

Let $\set D:=\{X\}\cup\set D'$, and observe that $|\set D|=|\set D'|+1=|\set C\backslash \{X''\}|-1+1=|\set C|-1$, and also that $(\genfun\circ\set U)(\{Y\}\cup\set D)=(\genfun\circ\set U)(\{X,Y\}\cup\set D')=(\genfun\circ\set U)(\{X''\}\cup(\{Y'\}\cup\set D'))=(\genfun\circ\set U)(\{X''\}\cup(\set C'\backslash\{X''\}))=(\genfun\circ\set U)(\set C)$.  Furthermore, by the earlier recursive application of this Lemma we know that $a\notin\set G_{\set Q}(O)$ for every $O\in \set D'$, so since we have also established that $a\notin \set G_{\set Q}(X)$, we conclude that $a\notin\set G_{\set Q}(O)$ for every $O\in \set D$;  since also know that every such $O$ must also satisfy $\set G_{\set Q}(O)\cap \set A \ne \emptyset$, we conclude that every such $O$ satisfies $\set G_{\set Q}(O) \cap A\backslash\{a\}\ne\emptyset$.
\end{proof}
\begin{remark}
This next Lemma provides the small but important result that we can always find a generator $a$ that appears somewhere in the conjugal pairs;  this has the consequence that we can now perform \emph{undirected} Gaussian elimination (in contrast to the \emph{directed} Gaussian elimination procedure described in the previous Lemma) by picking an arbitrary generator to eliminate rather than specifying a particular generator up-front.
\end{remark}

\begin{lemma}
\label{undirected-gaussian-elimination-of-logicals}
In the context of Proposition \ref{bound-on-recombinations}, there exists a Pauli operator $a$ satisfying the assumption of Lemma \ref{directed-gaussian-elimination-of-logicals}.
\end{lemma}

\begin{proof}
Take any pair $Y'\in\set C$.  By the assumptions of Proposition \ref{bound-on-recombinations}, we know that $\set G(Y')\cap \set A \ne \emptyset$, which implies that there exists an element $a\in A$ such that either $a\in (\set G_{\set Q}\circ p_1)(Y')$ or $a\in (\set G_{\set Q}\circ p_2)(Y')$.  The existence of $Y$ and $\set D$ then follow immediately from the application of Lemma \ref{directed-gaussian-elimination-of-logicals}.
\end{proof}
\begin{remark}
This final Lemma (inside the proof of Proposition \ref{bound-on-recombinations}) shows using Gaussian elimination that there must be a number of generators from $\set A$ present in the pairs in $\set C$ that is equal to the size of $\set C$, since otherwise we could recombine the pairs in $\set C$ to obtain a pair that includes no generator from $\set A$, contradicting the assumptions of Proposition \ref{bound-on-recombinations}.
\end{remark}
\begin{lemma}
\label{elimination-to-create-subset}
In the context of Proposition \ref{bound-on-recombinations}, there exists a set of conjugal pairs, $\set X$, with respect to $\set X\cup\set S$, and a subset of operators, $\set O\subseteq \set A$, such that
\begin{enumerate}
\item $|\set X|=|\set O|=|\set C|$;
\item $(\genfun\circ\set U)(\set X)=(\genfun\circ\set U)(\set C)$; and
\item for every $o\in\set O$, there is a conjugal pair $Y\in\set X$ such that $o\in\set G_{\set Q}(Y)$;
\end{enumerate}
\end{lemma}

\begin{proof}
Proof by induction.  If $\set C$ is empty, then the empty sets trivially satisfy this Lemma.

Now suppose that we have proven this Lemma for $|\set C|=N-1$, and assume we have been given sets $\set C$ and $\set A$ such that $|\set C|=N$.  Applying Lemma \ref{undirected-gaussian-elimination-of-logicals} to $\set C$ and $\set A$ we obtain the conjugal pair $Y$,  the set of conjugal pairs $\set D$, and the element $a$ described in the conclusions of that Lemma.  Apply this Lemma recursively to the respective sets $\set D$ and $\set A\backslash\{a\}$, we obtain the sets $\set X'$ and $\set O'$ described (without the primes) in this Lemma; let $\set X := \set X'\cup\{Y\}$ and $\set O:=\set O'\cup\{a\}$.  Note that $\set X$ is a set of conjugal pairs with respect to $\set X\cup\set S$ since $\set X'$ is a set of conjugal pairs with respect to $\set X'\cup\set S$, and we know that the operators in $Y$ commute with every operator in every pair in $\set X'$ since they commute with every operator in $(\genfun\circ\set U)(\set D)=(\genfun\circ\set U)(\set X')$.

First, observe that $|\set O|=|\set O'|+1$ since $a\notin \set O'$.  Furthermore, $a\notin \bigcup_{x\in \set U(\set X')} \set G(x)$ since $a\notin \bigcup_{x\in \set U(\set D)} \set G_{\set Q}(x)$ by Lemma \ref{undirected-gaussian-elimination-of-logicals} and $(\genfun\circ\set U)(\set D)=(\genfun\circ\set U)(\set X')$ by recursive application of this Corollary.  Thus, $|\set X|=|\set X'|+1$ since $Y\notin X'$ as $a\in\set G(Y)$, and $|\set X|=|\set O|=|\set D|+1=|\set C|-1+1=|\set C|$.

Second, observe that since $(\genfun\circ\set U)(\set X')=(\genfun\circ\set U)(\set D)=(\genfun\circ\set U)(\set C)$ by recursive application of this Lemma and $(\genfun\circ\set U)(\{Y\}\cup\set D)=(\genfun\circ\set U)(\set C)$ by Lemma \ref{undirected-gaussian-elimination-of-logicals}, we conclude that $(\genfun\circ\set U)(\set X) = (\genfun\circ\set U)(\{Y\}\cup\set X') = (\genfun\circ\set U)(\{Y\}\cup\set D) = (\genfun\circ\set U)(\set C)$.

Finally, observe that for every $o\in\set O$ we either have that $o=a$, in which case $Y\in\set X$ and $o\in\set G_{\set Q}(Y)$, or $o\in \set A\backslash\{a\}$, in which case by recursive application of this Lemma we know that there is an operator $Z\in \set X'\subseteq \set X$ such that $o\in\set G_{\set Q}(Z)$.
\end{proof}
\begin{remark}
With the preceding Lemmas having performed the heavy lifting, the proof of Proposition \ref{bound-on-recombinations} is quite simple.
\end{remark}

\begin{proof}[Proof of Proposition \ref{bound-on-recombinations}]
Proof by contradiction.  By Lemma \ref{elimination-to-create-subset}, there would have to exist a subset $\set O\subseteq\set A$ such that $|\set C|=|\set O|>|\set A|$, which is impossible.
\end{proof}
\begin{remark}
With the preceding Lemmas and Propositions, we now have all of the tools that we need to prove Theorem \ref{theorem:optimality-condition}.  Again, the idea behind this proof is that because the first members of pairs in the choice are contained in an unimprovable set, one cannot take products of the qubits in the choice in order to improve them;  thus, the only way one could construct a better choice would by forming $n+k$ independent qubits from products of $n$ qubits (where $k>0$), which is disallowed by the result of Proposition \ref{bound-on-recombinations}.  Hence, there can be no better choice.
\end{remark}

\begin{proof}[Proof of Theorem \ref{theorem:optimality-condition}]
Proof by contradiction.  Let $\lst P'$ be some choice of qubits stabilized by $\set S$ such that $(\genfun\circ\set U)(\lst P)=(\genfun\circ\set U)(\lst P')$ (which automatically implies that $|\lst P|=|\lst P'|$) and there exists some integer $k$ such that $M_{\set S}(\lst P')_k > M_{\set S}(\lst P)_k$;  in particular, let $k$ be the smallest such integer, and let $\set C:=\{\lst P'_i : i \ge k\}$.  Let $l$ be the smallest integer such that $\lst M_{\set S}(\lst P)_l\ge \lst M_{\set S}(\lst P')_k$ or $|\lst P|+1$ if there is no such integer, and let $\set A := \{p_1(\lst P_i) : i \ge l\}$; note that since $\lst M_{\set S}(\lst P')_k > \lst M_{\set S}(\lst P)_k$ we must have $l>k$, and hence $|\set C| > |\set A|$.

Take any conjugal pair $O:=(a,b)$ such that $\{a,b\}\in\genfun(\set C)$.  Since $a$ and $b$ anti-commute, it must be the case that $\{p_1(\lst P_i)\}_i\cap \set G_{\set U(\lst P)}(O)\ne\emptyset$, because if every operator in $\set G_{\set U(\lst P)}(O)$ were the second member of a pair in $\lst P$ then $a$ and $b$ would commute.  Let $c$ be a choice of $a$ or $b$ such that $\{p_1(\lst P_i)\}_i\cap \set G_{\set U(\lst P)}(c)\ne\emptyset$.  By Lemma \ref{combinations-can't-make-things-worse} we know that $\lst M_{\set S}(\lst P')_k\le\om_{\set S}(c)$ since $c\in\genfun(\set C)$.  By the assumption of this Theorem that $\{p_1(\lst P_i)\}_i$ is an unimprovable set that extends to $\set U(\lst P)$, we know that $\lst M_{\set S}(\lst P')_k \le \om_{\set S}(c)\le\min \{\om_{\set S}(x):x\in\{p_1(\lst P_i)\}_i\cap \set G_{\set U(\lst P)}(c)\}$.  From these bounds we conclude that $\{p_1(\lst P_i)\}_i\cap \set G_{\set U(\lst P)}(c)\subseteq \set A$, and since $\{p_1(\lst P_i)\}_i\cap \set G_{\set U(\lst P)}(c)\ne\emptyset$ we see therefore that $\set G_{\set U(\lst P)}(c)\cap\set A\ne\emptyset$ and so $\set G_{\set U(\lst P)}(O)\cap\set A\ne\emptyset$.

We have now demonstrated that for every pair $O:=(a,b)$ such that $\{a,b\}\in\set \genfun(\set C)$, we must have $\set G(O)\cap\set A \ne\emptyset$.  Observe that this means that sets $\set C$ and $\set A$ match the descriptions in Proposition \ref{bound-on-recombinations} (letting set $\set Q:=\{\set P_i\}_i$), and thus we see that it is impossible for $|\set C|>|\set A|$, and so we have a contradiction.  We thus conclude that no such choice $\lst P'$ can exist.
\end{proof}
\subsubsection{Correctness of the algorithm}

We now prove that this algorithm is correct --- that is, that it terminates and outputs an optimal choice of logical qubits.  We do so by proving the following theorem, which is the main result of this section.

\begin{theorem}
\label{theorem:algorithm is correct}
Given a set of commuting Pauli operators $\set S$ and a set of pairs $\set L\subseteq\centralizer_\pauligroup(\set S)$ conjugal in relation to $\set U(\set L)\cup\set S$, the sequence $\optimizer(\set S,\set L)$ is finite and if $(\set Q,\lst P,\lst s)$ is the last element then $\set Q=\emptyset$ and $\lst P$ is an optimal choice of logical qubits stabilized by $\set S$ such that $\genfun(\lst P)=\genfun(\set L)$.
\end{theorem}
\begin{remark}
Before proving this Theorem, we shall first prove several related Lemmas and Propositions.

Our ultimate goal is to expand the unimprovable set so that it includes at least the first member of every conjugal pair in the set of logical qubit operators, since this means that we have satisfied the optimality condition.  Thus, we want to be able to add operators to this set while preserving the property of being an unimprovable set.

The following Lemma shows that if we have an operator $o$ in a set $\set X$ to which some unimprovable set extends, then if the smallest weight undetectable error acting on $o$ acts on no other operator in $\set X$ then we may move $o$ to the unimprovable set to obtain a new unimprovable set that extends to $\set X\slash \{o\}$.  The intuition here is that because said error acts only on $o$, it cannot be canceled by multiplying $o$ by other operators, and so it is an ``unimprovable'' operator that can be included in our unimprovable set.
\end{remark}

\begin{lemma}
\label{lemma:move-it-over}
If $\set O$ is an unimprovable set with respect to $\set S$ that extends to $\set X:=\{o\}\cup\set X'$, and there exists an undetectable error, $h$, of weight $\om_{\set S}(o)$ that acts on $o$ but not on any operator in $\set X'$, then $\set O':=\set O\cup\{o\}$ is an unimprovable set with respect to $\set S$ that extends to $\set X'$.
\end{lemma}

\begin{proof}
Take any subset $\set R\subseteq \set O'\cup\set X'$ such that $\set R\cap\set O'\ne\emptyset$.  We need to show that $\om_{\set S}\paren{r} \le \min_{a\in\set R\cap\set O'}\om_{\set S}(a)$ where $r:=\Pi(\set R)$.

First consider the case where $o\notin \set R$;  in this case we have that $\set R\subseteq \set O\cup\set X$ such that $\set R\cap\set O=\set R\cap\set O'\ne\emptyset$, and so since $\set O$ extends to $\set X$ we conclude that $\om_{\set S}\paren{r} \le \min_{a\in\set R\cap\set O}\om_{\set S}(a)= \min_{a\in\set R\cap\set O'}\om_{\set S}(a)$.

Now consider the case where $\set R\cap\set O'=\{o\}$.  In this case, by the assumptions of this Lemma, we know $h$ acts on $o$ but not on any operator in $\set X'$ which implies that $h$ acts on $r$ and so $\om_{\set S}(r) \le w(h)=\om_{\set S}(o)=\min_{a\in\set R\cap\set O'}\om_{\set S}(a)$.

Finally we consider the remaining case where $\{o\}\subset\set R\cap\set O'$. Let $\set Z := \set R\backslash\{o\}\ne\emptyset$.  Since $\set O$ extends to $\set X$ and $\set R\cap\set O=\set Z\cap\set O\ne\emptyset$, we know that $\om_{\set S}(x)\le \min_{a\in\set Z\cap\set O}\om_{\set S}(y)=:d$.  If $d \le \om_{\set S}(o)$, then $\om_{\set S}(r)\le d = \min_{a\in\set R\cap\set O'}\om_{\set S}(a)$, and we are done.  Otherwise, since $d> \om_{\set S}(o)$ we know that $h$ acts on $r$ since there can be no other operator in $\set Z$ that anti-commutes with $h$, and so $\om_{\set S}(r) \le w(h)\le\om_{\set S}(o)=\min_{a\in\set R\cap\set O'}\om_{\set S}(a)$
\end{proof}
\begin{remark}
In Case 1 of the algorithm we take an element that is a member of an unimprovable set and replace it with the product of this element times some elements in the set to which the unimprovable set extends.  We want to show that this preserves the unimprovability of the set, and this is done in the following Lemma.
\end{remark}

\begin{lemma}
\label{lemma:replacing element with product preserves unimprovability}
Suppose we are given an unimprovable set $\set O$ with respect to $\set S$ that extends to $\set X$.  Let $o$ be any element in $\set O$, and $\set Z\subseteq \set O\cup\set X$ such that $o\in \set Z$.  Let $o':=\Pi(\set Z)$ and $\set O' := \paren{\set O\backslash\{o\}}\cup\{o'\}$.  If $\om_{\set S}(o')=\om_{\set S}(o)$, then $\set O'$ is also an unimprovable set with respect to $\set S$ that extends to $\set X$.
\end{lemma}

\begin{proof}
Take any subset of elements $\set R\subseteq\set O'\cup\set X$ such that $\set R\cap\set O'\ne\emptyset$, and let $x := \Pi(\set R)$.  We need to show that $\om_{\set S}(x) \le \min_{a\in\set R\cap\set O}\om_{\set S}(a)$.  If $o'\notin\set R$, then this follows immediately from the fact that $\set R\subseteq\set O\cup \set X$ and $\set R\cap\set O\ne\emptyset$ and $\set O$ is an unimprovable set that extends to $\set X$, so assume that that $o'\in\set R$.  Since $o'=\Pi(\set Z)$ and $\set Z\subseteq\set O\cup\set X$, we conclude that the set $\set T\subseteq\set O\cup\set X$ which is the symmetric difference of $\set Z$ and $\set R$ satisfies the property that $x=\Pi(\set T)$.  Note that $o\in\set T$ since $o\in \set Z$ and $o\notin \set O'$ and so $o\notin \set R$.  Thus, $\set T \cap\set O\ne\emptyset$, and so $\om_{\set S}(x)\le \min_{a\in\set T\cap\set O}\om_{\set S}(a)$ since $\set O$ is an unimprovable set that extends to $\set X$.  Thus, for us to show that $\om_{\set S}(x) \le \min_{a\in\set R\cap\set O'}\om_{\set S}(a)$, it suffices for us to show that $\min_{a\in\set R\cap\set O'}\om_{\set S}(a)=\min_{a\in\set T\cap\set O}\om_{\set S}(a)$.

First note that since $\om_{\set S}(o')=\om_{\set S}(o)$, there is no element $z\in\set Z\cap\set O$ such that $\om_{\set S}(z)<\om_{\set S}(o)$.  Thus, any operator $t\in\set T\cap\set O$ such that $\om_{\set S}(t)<\om_{\set S}(o)$ must also appear in $\set R\cap\set O'$, and vice versa;  put another way, any operator that is less robust to errors than $o$ must be present in both $\set T\cap\set O$ and $\set R\cap\set O'$ together or neither.  Thus, if at least one such operator exists, then we conclude that $\min_{a\in\set R\cap\set O'}\om_{\set S}(a)=\min_{a\in\set T\cap\set O}\om_{\set S}(a)$ since in this case any minimizer of $\om_{\set S}$ must be shared between the two sets.  If no such operator exists, then since $o\in\set T\cap\set O$ and $o'\in\set R\cap\set O'$ and there is no other operator present in either set with a smaller minimum weight undetectable error, we conclude that $o$ is the minimizer of $\om_{\set S}$ over $\set T\cap\set O$ and $o'$ is the minimizer over $\set R\cap\set O'$ and since $\om_{\set S}(o')=\om_{\set S}(o)$ we have that $\min_{a\in\set R\cap\set O'}\om_{\set S}(a)=\om_{\set S}(o)=\min_{a\in\set T\cap\set O}\om_{\set S}(a)$.

Thus we have shown that $\min_{a\in\set R\cap\set O'}\om_{\set S}(a)=\min_{a\in\set T\cap\set O}\om_{\set S}(a)$, and since our choice of $\set R$ was arbitrary we conclude that $\set O'$ is an unimprovable set that extends to $\set X$.
\end{proof}
\begin{remark}
In both cases of the algorithm we replace a set to which an unimprovable set extends with a new set that a product of elements in the old set.  We want to show that the new set is also an extension of the unimprovable set, and this is proved by the following Lemma.
\end{remark}

\begin{lemma}
\label{lemma:recombining extension elements preserves extension}
If $\set O$ is an unimprovable set with respect to $\set S$ that extends to $\set X$, and $\set X'$ is a set such that $\set X'\subseteq\genfun(\set X)$, then $\set O$ also extends to $\set X'$.
\end{lemma}

\begin{proof}
Take any subset $\set R' \subseteq \set O\cup\set X'$ such that $\set A := \set R'\cap \set O \ne \emptyset$.  We need to show that $\om_{\set S}(x)\le\min_{y\in\set A}\om_{\set S}(y)$, where $x := \Pi(\set R')$.  Note that since $\set X'\subseteq\genfun(\set X)$, there exists a set $\set B\subseteq \set X$ such that $x = \Pi\paren{\set A\cup\set B}$, and so since $\set O$ extends to $\set X$ we conclude that $\om_{\set S}(x)\le\min_{y\in\set A}\om_{\set S}(y)$.  Since our choice of $\set R'$ was arbitrary, we conclude that that $\set O$ extends to $\set X'$.
\end{proof}
\begin{remark}
Most of the heavy lifting in this section is performed in the following Proposition, which uses induction to prove a number of properties about the output of the algorithm at every step.
\end{remark}

\begin{proposition}
\label{proposition:properties of the algorithm}
Given a set of Pauli operators $\set S$ and a set of pairs $\set L\subseteq\centralizer_\pauligroup(\set S)$ conjugal in relation to $\set U(\set L)\cup\set S$, for every $(\set Q,\lst P,\lst s)\in\optimizer(\set S,\set L)$ we have that
\begin{enumerate}
\item $\genfun\paren{\set U(\set Q)\cup\set U(\lst P)}=\genfun(\set L)$;
\item $\set U(\set Q)\cup\set U(\lst P)\subseteq\centralizer_\pauligroup(\set S)$;
\item $2(|\set Q|+|\lst P|)=|\set L|$;
\item $|\set U(\set Q)\cup\set U(\lst P)|=|\lst L|$;
\item $\set U(\set Q)\cap\set U(\lst P)=\emptyset$, and no operator appears in more than one pair in either $\set Q$ or $\lst P$;
\item $\set O$ is an unimprovable set of operators that extends to $\set X$;
\item $\max_{o\in\set O}\om_{\set S}(o)\le\min_{x\in\set X}(x)$;
\item $(\om_{\set S}\circ p_1)(q)=m_{\set S}(q)$ for all $q\in\lst P$;
\item $\lst M(\lst P)$ is ordered;
\item $\lst P$ is a choice of qubits stabilized by $\set S$;
\item $\lst P$ is an optimal choice of qubits;
\end{enumerate}
where $\set X := \unpack(\set Q)\cup\{p_2(\lst P_i):1 \le i \le |\lst P|, \lst s_i=1\}$ if $\lst P$ is non-empty and $\set X := \unpack(\set Q)$ otherwise, and $\set O:=\unpack(\lst P)\backslash\set X$.
\end{proposition}

\begin{proof}
Proof by induction.  It is easy to see that these properties hold for $\optimizer(\set S,\set L)_0=(\set L',\lst\emptyset,\lst\emptyset)$, so now assume that they hold for $(\set Q,\lst P,\lst s):=\optimizer(\set S,\set L)_i$, and let $(\set Q',\lst P',\lst s'):=\optimizer(\set S,\set L)_{i+1}$.  For convenience, define $\set X:=\unpack(\set Q)\cup\{p_2(\lst P_i):1 \le i \le |\lst P|, \lst s_i=1\}$ (or $\set X:=\unpack(\set Q)$ if $\lst P$ is empty), $\set X':=\unpack(\set Q')\cup\{p_2(\lst P_i'):1 \le i \le |\lst P'|, \lst s'_i=1\}$, $\set O:=\unpack(\lst P)\backslash\set X$ and $\set O':=\unpack(\lst P')\backslash\set X'$.

We now prove each of the conclusion above;  note that in each conclusion we may assume that the conclusions prior to it have already been established, so we do so implicitly.

Also, when we say that we are assuming we are in ``Case 1'' or ``Case 2'', we mean that we are assuming that $(\set Q',\lst P',\lst s')$ followed from respectively Case 1 or Case 2 in the definition of $\optimizer$.

\begin{enumerate}
\item

Examination of the definition reveals that $\set Q'$ and $\lst P'$ are constructed entirely from products of elements in $\set Q$ and $\lst P$ so that $\genfun\paren{\set U(\set Q')\cup\set U(\lst P')}\subseteq\genfun\paren{\set U(\set Q)\cup\set U(\lst P)}$.

If $\optimizer(\set S,\set L)_{i+1}$ was defined using Case 1 then let $(b,a):=\lst P_k$ and $(b',a'):=\lst P'_k$ where $k$ is the integer described in Case 1;  otherwise let $(a,b):=q$ be the pair selected from $\set Q$ in the definition and $(a',b')$ be the last element of $\lst P'$.  Note that in either case, $a=a'$.

In both cases, observe that for every operator $o\in\set U(\set Q)\cup\set U(\lst P)\backslash\{b\}$ we have that either $o$ or $o\cdot a$ is contained in $\set U(\set Q')\cup\set U(\lst P')$, and since $a$ is also contained in this set we see immediately that any operator in $\set U(\set Q)\cup\set U(\lst P)\backslash\{b\}$ can be obtained from products of elements in $\set U(\set Q')\cup\set U(\lst P')$ (i.e., from an operator in $\set U(\set Q')\cup\set U(\lst P')$ times possibly $a$).  Thus we conclude that $\genfun\paren{\set U(\set Q)\cup\set U(\lst P)\backslash\{b\}}\subseteq\genfun\paren{\set U(\set Q')\cup\set U(\lst P')}$.  Since $b'$ is the product of $b$ with elements in $\genfun\paren{\set U(\set Q)\cup\set U(\lst P)\backslash\{b\}}$, and $\genfun\paren{\set U(\set Q)\cup\set U(\lst P)\backslash\{b\}}\subseteq\genfun\paren{\set U(\set Q')\cup\set U(\lst P')}$, we conclude that $b\in\genfun\paren{\set U(\set Q')\cup\set(\lst P')}$ and so $\genfun\paren{\set U(\set Q)\cup\set U(\lst P)}\subseteq\genfun\paren{\set U(\set Q')\cup\set U(\lst P')}$.

Thus we have proven that $\genfun\paren{\set U(\set Q')\cup\set U(\lst P')}=\genfun\paren{\set U(\set Q)\cup\set U(\lst P)}=\genfun(\set L)$, and so we are done.
\item

This follows from the fact that $\set L\subseteq\centralizer_\pauligroup(\set S)\Rightarrow\genfun(\set L)\subseteq\centralizer_\pauligroup(\set S)$ and $\set U(\set Q)\cup\set U(\lst P)\subseteq\genfun\paren{\set U(\set Q)\cup\set U(\lst P)}=\genfun(\set L)$ (which we just proved).
\item

By construction, either $|\set Q'|=|\set Q|$ and $|\lst P'|=|\lst P$ (in Case 1) or $|\set Q'|=|\set Q|-1$ and $|\lst P'|=|\lst P|+1$ (in Case 2).  In either case we have that $2(|\set Q'|+|\lst P'|)=2(|\set Q|+|\lst P|)=|\set L|$.
\item

Since the elements in $\set L$ are members of conjugal pairs, they are therefore independent, and so we see that we need at least $|\set L|$ operators to generate $\genfun(\set L)$.  Thus we need $|\set L|\le|\set U(\set Q')\cap\set U(\lst P')|\le 2(|\set Q'|+|\lst P'|)=|\lst L|$,  where the second inequality comes from the fact that a pair can unpack to at most two operators, and the last equality comes from the previous conclusion.  We thus conclude that $|\set U(\set Q')\cap\set U(\lst P')|=|\lst L|$.
\item

By combining the previous two conclusions we see that $|\set U(\set Q')\cup\set U(\lst P')|=2(|\set Q'|+|\lst P'|)$;  if this conclusion were false (i.e., an operator were repeated somewhere) then we would have that $|\set U(\set Q')\cup\set U(\lst P')|<2(|\set Q'|+|\lst P'|)$, which is contradicts our earlier results.
\item

First assume that we are in Case 1.  Let $k$ be the integer described in this case, $(a,b):=\lst P_k$, and $(a',b'):=\lst P'_k$.  Note that $a\in\set O$, and by construction $a':=\Pi(\set A')$ where $\set A'\subseteq \set O\cup\set X$ and $\set A'\cap\set O=\{a\}$, and so since by the inductive hypothesis we know that $\set O$ is a unimprovable set that extends to $\set X$ we know that $\om_{\set S}(a')\le\om_{\set S}(a)$.  Since by the inductive hypothesis we also know that $\max_{o\in\set O}\om_{\set S}(o)\le\min_{x\in\set X}(x)$, by Lemma \ref{combinations-can't-make-things-worse} we conclude that $\om_{\set S}(a')=\om_{\set S}(a)$.  Lemma \ref{lemma:replacing element with product preserves unimprovability} thus applies to our situation and allows us to conclude that $\set O'':=\paren{\set O\backslash\{a\}}\cup\{a'\}$ is an unimprovable set that extends to $\set X$.  Furthermore, since by construction $\set X'\subseteq\genfun(\set X)$ and $b\in\set X$, Lemma \ref{lemma:recombining extension elements preserves extension} allows us to conclude that $\set O''$ extends to $\{b\}\cup\set X'$.  By construction, there is an error of minimal weight that acts on $b$ but not on any other operator in $\set X'$, which means that by Lemma \ref{lemma:move-it-over} we conclude that $\set O''\cup\{b\}\equiv\set O'$ extends to $\set X'$.

Now assume that we are in Case 2.  Note that the only change from $\set O$ to $\set O'$ is the addition of a single element $o$ that has an error that acts only on it but not on any other operator in $\set X'$.  Note that since $\{o\}\cup\set X'\subseteq\genfun(\set X)$ we conclude from Lemma \ref{lemma:recombining extension elements preserves extension} that $\set O$ extends to $\{o\}\cup\set X'$, and from Lemma \ref{lemma:move-it-over} we conclude that $\set O'$ extends to $\set X'$.
\item

First observe that since $\set X'\subseteq\genfun(\set X)$, we conclude from Lemma \ref{combinations-can't-make-things-worse} that $\min_{x\in\set X}\om_{\set S}(x)\le\min_{x'\in\set X'}\om_{\set S}(x')$.

The difference between $\set O$ and $\set O'$ is the addition of a minimizer of $\om_{\set S}$ over $\set X$, $o$, and possible also the replacement of a single element.  Since $\max_{a\in\set O}\om_{\set S}(a)\le\min_{a\in\set X}\om_{\set S}(a)\le\min_{a\in\set X'}\om_{\set S}(a)$, we conclude that since $\om_{\set S}(o)=\min_{a\in\set X}\om_{\set S}(a)$ that therefore $\max_{a\in\set O\cup\{o\}}\om_{\set S}(a)\le\min_{a\in\set X'}\om_{\set S}(a)$.  If $\set O'=\set O\cup\{o\}$ then we are done.  Otherwise, we are in Case 1 which means that we have also replaced an element in $\set O$;  however, the operator we have replaced it with is the product of an operator from $\set O$ and operators from $\set X$, and since $\set O$ is an unimprovable set that extends to $\set X$ we conclude that the replacement can be no better than the operator it is replacing.  Thus, $\max_{a\in\set O'}\om_{\set S}(a)\le\min_{a\in\set X'}\om_{\set S}(a)$.
\item

Since $\max_{a\in\set O'}\om_{\set S}(a)\le\min_{a\in\set X'}\om_{\set S}(a)$, we immediately conclude that $(\om_{\set S}\circ p_1)(\lst P'_i)=m_{\set S}(\lst P'_i)$ when $\lst a'_i=1$.  By the inductive hypothesis, we know that $(\om_{\set S}\circ p_1)(\lst P'_i)=m_{\set S}(\lst P'_i)$ where $\lst P'_i=\lst P_i$;  furthermore, in both cases the pairs at the locations where $\lst a_i=0$ are unchanged from $\lst P$ to $\lst P'$, and in each case this turns out to leave just a single location that we still need to examine.

In Case 1, this location is the index $k$ described in that case, where $\lst a_k=1$ and $\lst a'_k=0$.  Since $p_1(\lst P'_k)$ is the product of a single element of $\set O$ and elements from $\set X$, we conclude from the fact that $\set O$ is an unimprovable set that extends to $\set X$ that $(\om_{\set S}\circ p_1)(\lst P'_k)\le(\om_{\set S}\circ p_1)(\lst P_k)$.  By the inductive hypothesis we know that $\max_{a\in\set O}\om_{\set S}(a)\le\min_{a\in\set X}\om_{\set S}(a)$.  Because $p_2(\lst P'_k)$ is a product of elements from $\set X$ we conclude from Lemma \ref{combinations-can't-make-things-worse} that $\min_{a\in\set X}\om_{\set S}(a)\le (\om_{\set S}\circ p_2)(\lst P'_k)$.  Since $p_1(\lst P)\in\set O$, we conclude that $(\om_{\set S}\circ p_1)(\lst P_k)\le\min_{a\in\set X}\om_{\set S}(a)$.  Combining all of these inequalities we reach the conclude that $(\om_{\set S}\circ p_1)(\lst P'_k)\le(\om_{\set S}\circ p_2)(\lst P'_k)$ and hence $(\om_{\set S}\circ p_1)(\lst P'_k)=m_{\set S}(\lst P'_k)$.

In Case 2, this location is the end of the sequence $\lst P'$, but since the addition to the sequences is a pair of operators from $\set X$ such that the first member is a minimizer of $\om_{\set S}$ over $\set X$ we conclude that $(\om_{\set S}\circ p_1)(\lst P'_{|\lst P'|})=m_{\set S}(\lst P'_{|\lst P'|})$.
\item

By the inductive hypothesis we have that $\lst M_{\set S}(\lst P)_i=(\om_{\set S}\circ \om_{\set S})(\lst P_i)$, and we have just shown that $\lst M_{\set S}(\lst P')_i=(\om_{\set S}\circ \om_{\set S})(\lst P'_i)$.  By the inductive hypothesis we know that $\lst M_{\set S}(\lst P)$ is ordered, and so to prove that $\lst M_{\set S}(\lst P')$ is ordered we need only check the places in the sequence where $p_1(\lst P_i)\ne p_1(\lst P'_i)$.  In both cases there is exactly one location where the first member of a pair is modified from $\lst P$ to $\lst P'$.

In Case 1, this is the index $k$ defined in that case, at which the first member was replaced with a product of that first member with elements in $\set X$.  Since this member is in $\set O$, and since $\max_{a\in\set O}\om_{\set S}(a)\le\min_{a\in\set X}\om_{\set S}(a)$ (by the inductive hypothesis), we conclude from the fact that $\set O$ is an unimprovable set that extends to $\set X$ that $(\om_{\set S}\circ p_1)(\lst P'_k)=(\om_{\set S}\circ p_1)(\lst P_k)$, and so we conclude that $\lst M_{\set S}(\lst P')$ is ordered.

In Case 2, this is the end of the sequence $\lst P'$ where a pair was appended to $\lst P$.  Since the pair contains elements from $\set X$, and $\max_{a\in\set O}\om_{\set S}(a)\le\min_{a\in\set X}\om_{\set S}(a)$, we conclude that $\lst M_{\set S}(\lst P')_{|\lst P'|}\ge\lst M_{\set S}(\lst P)_i$ for $i<|\lst P'|$, and so we conclude that $\lst M_{\set S}(\lst P')$ is ordered.
\item

The fact that $\lst P$ is a choice of logical qubits stabilized by $\set S$ follows from directly from the previous conclusions.
\item

From the definition of an unimprovable set it is easy to see that since $\set O'$ is an unimprovable set that extends to $\set X'$, it also extends to $\set X'\cup\set O'=\set U(\set Q')\cup\set U(\lst P')$.  Since $\{p_1(q):q\in\lst P'\}\subseteq \set O'$ and $\set U(\lst P')\subseteq \set X'\cup\set O'$, it is also easy to see from the definition that $\{p_1(q):q\in\lst P'\}$ is an unimprovable set that extends to $\set U(\lst P')$ --- that is, taking subsets does not affect the property of unimprovability.  Thus, we conclude from Theorem \ref{theorem:optimality-condition} that $\lst P'$ is therefore an optimal choice of qubits.
\end{enumerate}
\end{proof}
\begin{remark}
Now that the heavy lifting has been done by the preceding Proposition, the proof of Theorem \ref{theorem:algorithm is correct} is relatively simple.
\end{remark}

\begin{proof}[Proof of Theorem \ref{theorem:algorithm is correct}]
At every step in the algorithm, we either change an entry in $\lst s$ from 1 to 0 or remove an element from $\set Q$.  Since $\lst s$ is of finite length, as long as $\set Q$ is non-empty there will be a step at which another element is removed from it.  Thus, there is an index $k$ such that if $\lst O(\set S,\set L)_k=(\set Q,\lst P,\lst s)$ then $\set Q=\emptyset$, and by definition this is the last element of the sequence.  By Proposition \ref{proposition:properties of the algorithm} we know that $\lst P$ is optimal and also that $\genfun(\lst P)=\genfun(\lst L)$ (since $\set Q$ is empty), and so we are done.
\end{proof}
\subsubsection{Running time of the algorithm}

\label{subsubsection:running time analysis}

In this section we analyze the running time of the optimization algorithm;  the result is presented in the following Theorem.

\begin{theorem}
\label{theorem:bound on running time}
Suppose we are given
\begin{itemize}
\item a set of commuting Pauli operators, $\set S$, acting on $N$ physical qubits;
\item a set of pairs, $\set L\subseteq\centralizer_\pauligroup(\set S)$, conjugal with respect to $\set U(\set Q)\cup\set S$;
\item and a set of Pauli operators $\set C$ such that $\genfun(\set C)=\centralizer_\pauligroup(\set S)$;
\end{itemize}
then the time needed to compute the sequence $\optimizer(\set S,\set L)$ is in the set $O\paren{|\set C|^2+|\set L|(|\set L|+d)3^d\choose{N}{d}},$ where $d:=\lst M(\lst P)_{|\lst P|}$ and $(\set Q,\lst P,\lst S)$ is the last element in the sequence (i.e., $\lst P$ is the desired optimal choice of qubits).
\end{theorem}
\begin{remark}
Before proving this Theorem, we shall first prove a number of related Lemmas and Propositions.

The most complicated part of analyzing the running time of the optimization algorithm is analyzing the time needed to find the minimum weight undetectable error.  In fact, the procedure for doing this was not even described explicitly in the algorithm, so we shall now explain how we do it.

The algorithm we employ is based on the Brouwer-Zimmermann search algorithm, which searches for the minimum weight binary string satisfying some property given a set of binary string generators endowed with a multiplication operation defined to be the exclusive-or operation.  The Brouwer-Zimmermann algorithm works by using a Gaussian elimination analogue to express the generators in reduced row echelon form;  it then performs its search by examining all products of $r$ generators for increasing $r$.  When all of the products of $r$ generators have been enumerated, one knows that the set of strings that has yet to be enumerated has weight $r+1$ or greater, since the row-echelon form means that every string has a column for which it is the only string with a 1 in that column, and so a product of $k$ strings must have a weight of at least $k$.  Thus, as the search proceeds, there is a growing lower-bound on the weight of the binary string, and the search halts when a string has been found that matches this bound.

This algorithm cannot be immediately applied to the current problem because we are not working with binary strings, and in particular Pauli operators have a more complicated multiplication operation that binary strings.  Fortunately, in~\cite{White:2006fj} White and Grassl showed that the Brouwer-Zimmermann enumeration can be generalized.

The key difference between binary strings and Pauli operators is that binary strings only have two possible values in a given column, whereas Pauli operators have four.  Thus, whereas we only need one element to generate all of the possible values in a given column for a binary string, we need \emph{two} elements to generate all of the possible values in a given column for a Pauli operator.  Thus, rather than working with generators, we instead work with a generalization that White and Grassl call \emph{pseudo-generators}, which we shall define here as follows.
\end{remark}

\begin{definition}
A \emph{pseudo-generator} is a set of either 1 or 2 Pauli operators.  In an abuse of notation, we extend the functions $\set U$ and $\genfun$ to be respectively $\set G\mapsto \bigcup_{\set g\in\set G} \set g$ and $\set G\mapsto (\genfun\circ\set U)(\set G)$ when applied to a set of pseudo-generators.
\end{definition}
\begin{remark}
(The preceding definition does not follow that of White and Grassl exactly; it has been specialized to our situation for the sake of simplicity.)

Unlike `normal' generators --- i.e., Pauli operators --- a product of generators is not a Pauli operator but rather a set of Pauli operators, which we define as follows.
\end{remark}

\begin{definition}
\label{definition:pseudo-product}
Suppose we are given a set of $r$ pseudo-generators $\set G$.  Let $\set X:=\{\genfun(\set g)\backslash\{I\}: \lst g\in\set G\}$, $\set Y$ be the $r-$ary Cartesian product of the $r$ sets contained in $\set X$, and $\set Z$ be the set consisting of the normal quantum operator product of the $r$ operators in each $r-$tuple in $\set Y$.  Then $\set Z$ is defined to be the \emph{pseudo-product} of the pseudo-generators in $\set G$.  For convenience, we define a function $\pseudoproduct$ such that $\pseudoproduct(\set G)$ is the pseudo-product of the pseudo-generators in $\set G$.
\end{definition}

\begin{remark}
The following Lemma places a bound on the size of the pseudo-product.
\end{remark}

\begin{lemma}
\label{lemma:bound-on-pseudo-product}
The pseudo-product of $r$ pseudo-generators contains at most $3^r$ operators.
\end{lemma}

\begin{proof}
Every set in $\set X$ described in Definition \ref{definition:pseudo-product} has a cardinality of either 1 or 3, and the size of $\set Y$ is equal to the product of the sizes of all the sets in $\set X$;  since $|\set X|=r$, we therefore conclude that the cardinality of $\set Y$ and hence the number of operators in the pseudo-product is at most $3^r$.
\end{proof}

\begin{corollary}[to Lemma \ref{lemma:bound-on-pseudo-product}]
\label{corolary:bound-on-pseudo-product}
Given a set of $r$ pseudo-generators, $\set G$, then the set $\set O:=\{f(o):o\in\genfun(\set G)\}$ can be computed in time $O((T+r) 3^r)$, where the time needed to compute $f$ is in $O(T)$.
\end{corollary}

\begin{proof}
From Lemma \ref{lemma:bound-on-pseudo-product} we know that there are at most $3^r$ operators in the pseudo-product, so $|\set O|\le 3^r$.  Furthermore, for every element in the set we first need to compute the corresponding operator in the pseudo-product, which requires $r$ time since it is the product of $r$ operators, and then we need to compute $f$, which by assumption requires a time in $O(T)$.
\end{proof}
\begin{remark}
In order to be able to place a lower bound on binary strings that have yet to be examined in the Brouwer-Zimmermann enumeration, we need the generators of the binary strings over which we are searching to have the property that each generator has a column such that it is the only generator with a 1 in that column, so that products of $r$ generators must have at least weight $r$.  Because we want to similarly place a bound on unexamined products of pseudo-generators, we generalize this property with the following definition.
\end{remark}

\begin{definition}
\label{definition:disjoint-pseudo-generators}
A set of pseudo-generators $\set G$ is said to be \emph{disjoint} if for every $\set g\in\set G$ there exists some physical qubit $k$ such that either $X_k$ or $Z_k$ (or both) anti-commutes with every operator in $\genfun(\set g)\backslash\{I\}$, but both $X_k$ and $Z_k$ commute with every operator in $\bigcup_{\set g'\in\set G\backslash\{\set g\}}\genfun(\set g')$.
\end{definition}

\begin{remark}
With the following Lemma, we show that the property of disjointness is exactly what we need to obtain the bounds that we want.
\end{remark}

\begin{lemma}
\label{lemma:disjoint-pseudo-generators-bound}
All of the operators in the pseudo-product of any $r$ (distinct) pseudo-generators chosen from a disjoint set of pseudo-generators have weight of at least $r$.
\end{lemma}

\begin{proof}
Every operator in the pseudo-product is the product of $r$ factors, each of which is associated with some distinct physical qubit $k$ such that it anti-commutes with either $X_k$ or $Z_k$ (or both) but every other factor commutes with both $X_k$ and $Z_k$;  thus, the product must anti-commute with at least $r$ single-qubit operators acting on distinct physical qubits, and so it must have a weight of at least $r$.
\end{proof}
\begin{remark}
Now that we have the concept of a disjoint set of pseudo-generators and a result showing that an operator in a pseudo-product of $r$ of them must have a weight of at least $r$, we present in the following Lemma an algorithm for searching through the space spanned by the pseudo-generators for an operator satisfying a given property.
\end{remark}

\begin{lemma}
\label{lemma:minimal-weight-search}
Given a set of pseudo-generators $\set G$ acting on $N$ physical qubits and a test function $f:\pauligroup\to\{0,1\}$ such that $f^{-1}(1)\cap\genfun(\set G)\ne\emptyset$, then a solution $o$ such that $f(o)=1$ and $\om_{\set S}(o)=\min_{o'\in \genfun(\set G), f(o')=1} w(o')$ can be computed in time $O\paren{(T+d)3^d\choose{|\set G|}{r}}$ where $d:=\min\paren{w(o),|\set G|}$ and $T$ is the time needed to compute $f$.
\end{lemma}

\begin{remark}
This Lemma follows directly from the results in~\cite{White:2006fj}, though the proof is included here both for completeness and also to show specifically how the results specialize to our case.  A pseudo-code representation of this algorithm can be found in Table \ref{table:find-weight-minimizer}.
\end{remark}

\begin{proof}
Define $\set C_r$ to be the set of all operators such that if $o\in\set C_r$ then there is some subset of exactly $r$ pseudo-generators from $\set G$ such that $o$ is contained in their pseudo-product.  Note that $\cup_r \set C_r = \genfun(\set G)$, so for every operator $o$ in the search space there is an integer $r$ such that $o\in\set C_r$.  Corollary \ref{corolary:bound-on-pseudo-product} shows that we can evaluate $f$ on every element of the pseudo-product of $r$ pseudo-generators in time $O((T+r)3^r)$, so since there are $\choose{|\set G|}{r}$ ways to choose $r$ pseudo-generators from $\set G$ we conclude that we can search $\set C_r$ for a solution to $f(o)=1$ in time $O\paren{(T+r) 3^r \choose{|\set G|}{r}}$.

From Lemma \ref{lemma:disjoint-pseudo-generators-bound} we conclude that $w(o)\ge r$ for every $o\in\set C_r$.  By extension this means that $w(o)\ge r$ for every $o\in\bigcup_{r'=r}^{|\set G|}\set C_{r'}$, and therefore that if $o\notin\bigcup_{r'=0}^{r-1}\set C_{r'}$ and $o\in\genfun(\set G)$ then $w(o)\ge r$.  Thus, if there exists an $r$ such that $r=\min_{o\in\bigcup_{r'=0}^{r-1}\set C_{r'},f(o)=1}w(o)$ then we know that $r=\min_{o\in\genfun(\set O),f(o)=1}w(o)$ --- that is, $r$ is \emph{exactly} the weight of the minimum weight solution to $f(o)=1$, since any operator in the search space that \emph{isn't} contained in $\bigcup_{r'=0}^{r-1}\set C_{r'}$ must have a weight of at least $r$.  Put another way, after having enumerated all of the elements in $\bigcup_{r'=0}^{r-1}\set C_{r'}$ we can check to see whether the smallest solution to $f$ we have seen so far (if any) has weight less than or equal to $r$, and if so we are done since we have found the minimal weight solution.

Now consider the procedure of searching through each $\set C_r$ starting with $r=0$.  We know that we will eventually find at least one solution to $f(o)=1$, since in this Proposition we have assumed that such an operator exists in the search space (by the assumption that $f^{-1}(1)\cap\genfun(\set G)\ne\emptyset$).  Furthermore, employing this procedure we will find the minimal solution $o$ no \emph{later} than after we have searched through $\set C_r$ for $r=0\dots w(o)$, since at that point all of the unexamined operators have a weight greater than $w(o)$.  Thus, we conclude that we shall find the minimal weight solution after having searched at most all of the elements in $\bigcup_{r=0}^{\min(w(o),|\set G|)} C_r$, which we can do in time $$O\paren{\sum_{r=0}^{\min(w(o),|\set G|)}(T+r) 3^r \choose{|\set G|}{r}}\subseteq O\paren{(T+d) 3^d \choose{N}{d}},$$ where $d:=\min(w(o),|\set G|)$.
\end{proof}
\begin{table}
\begin{codebox}
\Procname{$\proc{Find-Weight-Minimizer}(f,\lst G)$}
\li $r \gets 1$
\li $m \gets \infty$
\li \While $m > r$ and $r \le |\lst G|$
\li \Do
\li     \For each $\lst H\subseteq \lst G$ such that $|\lst H|=r$,
\li     and each $o$ in the pseudo-product of $\lst H$
\li     \Do
\li         \If $\func{weight}(o) < m$
\li         \Then
\li             $(q,u) \gets f(o)$
\li             \If $q$ is $\textsc{true}$
\li             \Then
\li                 $m \gets \func{weight}(o)$
\li                 $\alpha \gets (o,u)$
\li                 \If $m = r$
\li                 \Then
\li                     \Goto \ref{li:found-the-minimum}
                    \End
                \End
            \End
        \End
\li     $r \gets r + 1$
    \End
\li \Return $\alpha$ \label{li:found-the-minimum}
\end{codebox}
\caption[Algorithm \proc{Find-Weight-Minimizer}]{Algorithm which finds the minimal weight operator in a given generating set that satisfies a given predicate.  For the sake of convenience, we also allow the query function to return auxialiary information that is returned to the caller along with the minimal weight operator.} \label{table:find-weight-minimizer}
\end{table}
\begin{remark}
The proceeding Lemma is rather general, so we shall show how it specializes to our case.  First, however, we use the following three Lemmas to prove that our search space is generated by exactly $N$ pseudo-generators.
\end{remark}

\begin{lemma}
\label{lemma:lower bound on number of disjoint pseudo-generators}
Given a set of disjoint pseudo-generators, $\set G$, the largest subset $\set X\subseteq\set U(\set G)$ such that $\set X$ commutes has size $|\set X|\le|\set G|$.
\end{lemma}

\begin{proof}
If $\set X$ contained more that $\set G$ operators then by the pigeon hole principle there would have to be at least two operators from the same pseudo-generator, and thus which did not commute
\end{proof}

\begin{lemma}
\label{lemma:upper bound on number of disjoint pseudo-generators}
A set of disjoint pseudo-generators $\set G$ acting on $N$ qubits satisfies $|\set G|\le N$.
\end{lemma}

\begin{proof}
This follows directly from the definition and the pigeon hole principle.
\end{proof}

\begin{lemma}
\label{lemma:exact bound on number of disjoint pseudo-generators}
For any set of commuting operators $\set S$ acting on $N$ physical qubits, if $\set G$ is a set of disjoint pseudo-generators satisfying $\genfun(\set G)=\centralizer_\pauligroup(\set S)$ then $|\set G|= N$.
\end{lemma}

\begin{proof}
Since $\centralizer_\pauligroup(\set S)$ contains a subset of $N$ independent commuting operators, so must $\genfun(\set G)$ and therefore $\set U(\set G)$;  thus, Lemma \ref{lemma:lower bound on number of disjoint pseudo-generators} implies that $|\set G|\ge N$, and combining this bound with that given by Lemma \ref{lemma:upper bound on number of disjoint pseudo-generators} we see that $|\set G|=N$.
\end{proof}
\begin{remark}
We now prove a Lemma which shows how the search specializes to the case of our qubit optimization algorithm.
\end{remark}

\begin{lemma}
\label{lemma:search for minimal weight undetectable error}
Given
\begin{itemize}
\item a set of Pauli operators $\set S$ acting on $N$ physical qubits,
\item a set of disjoint pseudo-generator $\set G$ such that $\genfun(\set G)=\centralizer(\set S)$, and
\item a non-empty set of Pauli operators $\set Q$ such that $\set Q\cap\set S=\emptyset$ and every operator in $\set Q$ is a member of a conjugal pair in relation to $\set Q \cup \set S$,
\end{itemize}
then a minimal weight undetectable error acting on any operator in $\set Q$ can be found in time $O\paren{(|\set Q|+d)3^d\choose{N}{d}}$ where $d:=\min\paren{w(o),N}$.
\end{lemma}

\begin{proof}
First observe that by Lemma \ref{lemma:exact bound on number of disjoint pseudo-generators} we know that $|\set G|=N$.

Define the function $f:\genfun(\set G)\to\{0,1\}$ by
$$f(o):=
\begin{cases}
1 & \exists\,\, q\in\set Q \,\,\text{such that}\,\, \{o,q\}=0\\
0 & \text{otherwise},
\end{cases}
$$
Note that solutions to $f$ are undetectable errors acting on $\set Q$, and also that this function can be computed in time $O(|\set Q|)$ by checking the commutator for each element in $\set Q$.  Furthermore note that for every operator in $\set Q$ there is another operator in $\set Q$ which anti-commutes with it, and also that $\set Q\subset \centralizer_\pauligroup(\set S)$.  Thus, since $\set Q$ is non-empty, there is at least one operator $o\in\genfun(\set G)$ such that $f(o)=1$.  Thus, by Lemma \ref{lemma:minimal-weight-search}, we know that we can compute a minimal weight solution to $f$ in time $O\paren{(|\set Q|+d)3^d\choose{N}{d}}$ where $d:=\min\paren{w(o),N}$.
\end{proof}
\begin{remark}
In order to make use of the preceding result, we need to have a set of disjoint pseudo-generators whose pseudo-product covers our search space.  However, we usually start instead with a set of ordinary Pauli operators that generate this space.  Thus, we shall now show that the former can be computed from the latter --- i.e., that given a set of Pauli operators, we can compute a set of disjoint pseudo-generators that spans the same space.  First we present a Lemma that provides a criteria sufficient to show that a set of pseudo-generators is distinct.
\end{remark}

\begin{lemma}
\label{lemma:disjointness-equvalence}
Given a set of pseudo-generators, $\set G$, if there exists a map $p:\set U(\set G)\to \paren{\{X_k\}_k \cup \{Z_k\}_k}$ such that
\begin{enumerate}
\item for every $o\in\set U(\set G)$, $o$ is the unique operator in $\genfun(\set G)$ that anti-commutes with $p(o)$ and
\item for every $\set g\in\set G$, the operators in $\set g$ are both mapped by $p$ to single-qubit operators acting on the same physical qubit $k$, and they are the only such operators in $\genfun(\set G)$ that are mapped by $p$ to operators acting on $k$,
\end{enumerate}
then $\set G$ is disjoint.
\end{lemma}

\begin{proof}
For every $\set g\in\set G$, we conclude from property 2 of $p$ that there is some physical qubit $k$ such that every operator in $\set g$ is mapped by $p$ to either $X_k$ or $Z_k$, and hence by property 1 this means that every operator in $\genfun(g)$ anti-commutes with either $X_k$ or $Z_k$.  Since by property 1 we know that the choice of $X_k$ or $Z_k$ is different for each operator in $\genfun(\set g)$, we conclude that if there is more than one operator in $\set g$ then the product anti-commutes with \emph{both} $X_k$ or $Z_k$.  Finally, by property 2 we know that every operator in $\set g'$ for $\set g'\in\set \set G\backslash\{g\}$ commutes with $X_k$ and $Z_k$.
\end{proof}

\begin{remark}
We now show that any set of operators that we are using to generate a search space can be expressed equivalently as a set of disjoint pseudo-generators.
\end{remark}

\begin{lemma}
\label{lemma:computing-disjoint-pseudo-generators}
Given any a set of Pauli operators, $\set O$, there exists a set $\set G$ of psuedo-generators such that
\begin{enumerate}
\item $|\set G|\le|\set O|$,
\item $\genfun(\set O)=\genfun(\set G)$,
\item the map $p$ described in Lemma \ref{lemma:disjointness-equvalence} exists for $\set G$,
\end{enumerate}
and $\set G$ can be computed in time $O(|\set O|^2)$.
\end{lemma}

\begin{remark}
The structure of this proof bears some similarities to Proposition \ref{make-independent-using-elimination}.  In constrast with Proposition \ref{make-independent-using-elimination}, however, in the setting of this Lemma we are working with operators that in general will not commute.
\end{remark}

\begin{proof}
Proof by induction.  For the base case, we observe that if $\set O$ is empty, then the trivial set $\set G:=\emptyset$ and the trivial function $p$ whose domain is the empty set satisfy the requirements.

Now assume that this Lemma has been proven for sets of cardinality $n-1$, and suppose we are given a (non-empty) set $\set O$ of cardinality $n$.  Take any operator $o\in\set O$.  By recursive application of this Lemma, we know that we can construct the set $\set G':=\set G$ and the function $p':=p$ described in this Lemma given $\set O:=\set O\backslash\{o\}$ in time $O\paren{(n-1)^2}=O(n^2)$.

Let $$o':=o\cdot \prod_{x\in\set U(\set G'),\atop\{o,p(x)\}=0} x.$$  Note that for every $x\in\set U(\set G')$, it must be that $o'$ commutes with $p'(x)$, since $o'$ is formed from a product that has either two factors that anti-commute with $p'(x)$ (namely, $o$ and $p'(x)$) or no operators that anti-commute with $p'(x)$.  If $o'$ is the identity operator, then let $\set G:=\set G'$ and $p:=p'$ and we are done.  Otherwise, there must be some operator $z\in\paren{\{X_k\}_k \cup \{Z_k\}_k}\slash\{p'(x):x\in\set U(\set G')\}$ that anti-commutes with $o$.  Define the function $f$ by
$$f(x) :=
\begin{cases}
x \cdot o' & \{x,z\}=0, \\
x          & \text{otherwise},
\end{cases}
$$
and let $\set G'':=\left\{\{f(x):x\in\set g\}:\set g\in\set G'\right\}$ and $p'':=p'\circ f^{-1}$.  Note that $o'$ must be independent of the operators in $\set U(\set G')$, because $o'$ is not the identity and the product of $o'$ with any subset of operators $\set A\subset\set U(\set G')$ cannot be the identity since it must anti-commute with $p'(a)$ for every $a\in\set A$.  Thus, $f$ is a bijective map from $\set U(\set G')$ to $\set U(\set G'')$ and hence is invertible, and so we conclude that $p''$ is well-defined.  Since, as previously discussed, $o'$ commutes with $p(y)$ for every $y\in\set U(\set G')$, we conclude that multiplication by $o'$ does not change whether any operator $x\in\set U(\set G')$ commutes or anti-commutes with $p(y)$ for any $y\in\set U(\set G')$, and so we conclude that the properties listed in Lemma \ref{lemma:disjointness-equvalence} that $p'$ has in relation to $\set G'$ (from the inductive hypothesis) are preserved in the transformation by $f$ so that $p''$ also has the same properties in relation to $\set G''$.  Furthermore, since every operator $x\in\set U(\lst G')$ was multiplied by a factor of $o'$ if and only if it anti-commutes with $z$, we conclude that $f(z)$ must commute with $z$, and thus every operator in $\set U(\lst G'')$ must commute with $z$.

There are two cases to consider:  either there is no operator $x\in\set U(\set G'')$ such that $p''(x)$ acts on the same qubit as $z$, or there is exactly one, since if there were more than two then it would violate the properties of $p''$, and if there were exactly two then by construction $o'$ would commute with $z$ leading to a contradiction.  In the first case, let $\set G:=\set G''\cup\left\{\{o'\}\right\}$.  In the second case, let $\set G:=\paren{\set G''\backslash\left\{\{x\}\right\}}\cup\left\{\{x,o'\}\right\}$, where $x$ is the single operator in $\set U(\set G'')$ such that $p(x)$ acts on the same qubit as $z$.   In either case, define
$$
p(x) :=
\begin{cases}
p''(x) & x\in \set U(\set G''), \\
z & x=o',
\end{cases}
$$
observing that it is well-defined since $\set U(\set G)=\set U(\set G'')\cup\{o'\}$.

To prove conclusion 1, we note that $\set G$ has at most one more element than $\set G'$ and $\set O$ always has one more element than $\set O\backslash\{o\}$, so conclusion 1 follows from this fact combined with the inductive hypothesis.

To prove conclusion 2, we note that since $o'$ (and hence $o$) is independent with respect to $\set U(\set G')$, then because of how $\set G$ was constructed and the inductive hypothesis we have that $\genfun(\set G)=\genfun\paren{\set U(\set G'')\cup\{o'\}}=\genfun\paren{\set U(\set G')\cup\{o\}}=\genfun\paren{(\set O\backslash\{o\})\cup\{o\}}=\genfun(\set O)$.

To prove conclusion 3, we need to show that $p$ satisfies the properties listed in Lemma \ref{lemma:disjointness-equvalence}.  To prove the first property, we note that for every $x\in\set U(\set G)$ we have that either $x\in\set U(\set G'')$, in which case we have already shown that it is the unique operator that commutes with $p(x)$ as this is true for the operators in $\set U(\set G'')$ as well as for $o'$ (by construction), or $x=o'$, in which case this is still true since by construction $o'$ is the only operator in $\set U(\set G)$ that anti-commutes with $z$.  To prove the second property, we note that due to the inductive hypothesis we need only consider the single change from $\set G''$ to $\set G$, which consisted of either adding or replacing a pseudo-generator;  in the first case (adding a generator), observe that we showed earlier that no operator $y\in\set G''$ is such that $p''(y)$ acts on the same qubit as $z$, and in the second case (replacing a generator), note that we added $o'$ to the only generator in $\set G''$ containing an operator $y$ such that $p''(y)$ acts on the same qubit as $z$;  in either case, we see that the second property holds for $\set G$.

Finally, we consider the running time.  In addition to the $O(n^2)$ time required to construct $\set G'$, we required an additional $O(n)$ multiplication operations to construct $o'$ and $\set G''$;  hence the total running-time is $O(n^2)$.
\end{proof}

\begin{corollary}
\label{corolary:computing-disjoint-pseudo-generators}
Given any set of Pauli operators, $\set O$, there exists a disjoint set of pseudo-generators $\set G$ such that $|\set G|\le|\set O|$, $\genfun(\set O)=\genfun(\set G)$ and $\set G$ can be computed in time $O(|\set O|^2)$.
\end{corollary}

\begin{proof}
This follows immediately from Lemmas \ref{lemma:disjointness-equvalence} and \ref{lemma:computing-disjoint-pseudo-generators}.
\end{proof}
\begin{table}
\begin{codebox}
\Procname{$\proc{Compute-Pseudogenerators}(\lst O)$}
\li $\lst p \gets []$
\li $i \gets 0$
\li \While $i < |\lst O|$ \label{li:next-operator}
\li \Do
\li     $o \gets \lst O[i]$
\li     \For $j \gets 0$ \To $i-1$
\li     \Do
\li         $(n,z) \gets \lst p[j]$
\li         \If $z = 0$
\li         \Then
\li             \If $\func{anti}(o,X_n)$
\li             \Then $o \gets o \cdot \lst O[j]$
                \End
\li         \Else
\li             \If $\func{anti}(o,Z_n)$
\li             \Then $o \gets o \cdot \lst O[j]$
                \End
            \End
        \End
\li     \If $o$ is identity
\li     \Then
\li         delete $\lst O[i]$
\li         \Goto \ref{li:next-operator}
        \End
\li     \For $n \gets 0$ \kw{to} number of physical qubits
\li     \Do
\li         \If $\func{anti}(o,X_n)$
\li         \Then
\li             $z \gets 0$
\li             \Goto \ref{li:found-anti-commuting-pauli}
\li         \ElseIf $\func{anti}(o,Z_k)$
\li         \Then
\li             $z \gets 1$
\li             \Goto \ref{li:found-anti-commuting-pauli}
            \End
        \End
\li     \If $z = 0$ \label{li:found-anti-commuting-pauli}
\li     \Then
\li         \For $j \gets 0$ \kw{to} $i-1$
\li         \Do
\li             \If $\func{anti}(\lst O[k],X_n)$
\li             \Then $\lst O[j] \gets \lst O[j] \cdot o$
                \End
            \End
\li     \Else
\li         \For $j \gets 0$ \kw{to} $i-1$
\li         \Do
\li             \If $\func{anti}(\lst O[j],Z_n)$
\li             \Then $\lst O[j] \gets \lst O[j] \cdot o$
                \End
            \End
        \End
\li     append $(n,z)$ to $\lst p$
\li     $\lst O[i] \gets o$
\li     $i \gets i + 1$
    \End
\li $\lst G \gets []$
\li \For $n\gets 0$ \To number of physical qubits
\li \Do
\li     $\lst g \gets []$
\li     \For $i \gets 0$ \To $|\lst O|-1$
\li     \Do
\li         $(n',\_)\gets \lst p[i]$
\li         \If $n=n'$
\li         \Then append $\lst O[i]$ to $\lst g$
            \End
        \End
\li     \If $\lst g \ne []$
\li     \Then append $g$ to $\lst G$
        \End
    \End
\li \Return $\lst G$
\end{codebox}
\caption[Algorithm \proc{Compute-Pseudogenerators}]{Algorithm which computes a set of disjoint pseudo-generators that generates the input set of operators.} \label{table:compute-pseudo-generators}
\end{table}
\begin{remark}
We now have the tools that we need to analyze the running time of the algorithm.
\end{remark}

\begin{proof}[Proof of Theorem \ref{theorem:bound on running time}]
First observe that from Lemma \ref{lemma:disjointness-equvalence} we conclude that we can compute a set of disjoint pseudo-generators $\set G$ such that $\genfun(\set G)=\genfun(\set C)=\centralizer_\pauligroup(\set S)$ in time $O(|\set C|^2)$.  We will assume that this set of pseudo-generators is implicitly available to us throughout the algorithm so that we do not need to compute it more than once.

At each step of the algorithm, we first need to find an operator, $o$, that has an undetectable error of minimal weight inside a set which we know from Proposition \ref {proposition:properties of the algorithm} has at most $|\set L|$ elements.  By Lemma \ref{lemma:search for minimal weight undetectable error}, we conclude that this operator can be found in time $O\paren{(|\set L|+d)3^d\choose{N}{d}}$ where $d:=\min\paren{w(o),N}=w(o)$ (since the weight of any operator cannot be greater than $N$).  After this has been found, examination of the algorithm reveals that the computation performed afterward takes a running time in $O(|\set Q|+|\lst P|)=O(|\set L|)$, where the equality comes from Proposition \ref {proposition:properties of the algorithm}.  Thus, the total time needed for each step is in $O\paren{(|\set L|+d)3^d\choose{N}{d}+|\set L|}=O\paren{(|\set L|+d)3^d\choose{N}{d}}$.

Let $(\set Q',\lst P',\lst s')$ be the second to last element of $\optimizer(\set S,\set L)$ and $(\set Q,\lst P,\lst s)$ the last element.  In the final step of the algorithm, we move the last remaining pair in $\set Q'=\{q\}$ over to $\lst P'$, which means that the operator $o$ with the minimal weight is a member of $q$.  From Proposition \ref {proposition:properties of the algorithm}, we know that that $\om_{\set S}(o)\ge m_{\set S}(\lst P'_i)$ for any $i$.  Thus, at each step of the algorithm before this one we know that we spent a time in $O\paren{(|\set L|+d)3^d\choose{N}{d}}$ where $d:=w(o)$ --- i.e., a time no greater than the time spent on the last step.  Since the algorithm requires at most $|\set L|$ steps we conclude that the total running time, including that needed to compute the set of pseudo-generators, is in $O\paren{|\set G|^2+|\set L|(|\set L|+d)3^d\choose{N}{d}}$.  Since $\om_{\set S}(o)=\lst M(\lst P)_{|\lst P|}$, we conclude that $d\equiv\lst M(\lst P)_{|\lst P|}$ (since there can be at most $N$ qubits in the choice), and so we are done.
\end{proof}
\section{Practice} \label{sec:lattice}
\subsection{Methodology} \label{sec:methodology}

In the previous section we presented an algorithm that computes the optimal subsystem code that can be implemented using a given set of measurements.  The procedure for optimizing the code requires an exponential amount of time, but fortunately the power of the exponential is a function of the distance of the best qubit in the code.   Because of this property, this algorithm can be effectively applied to search over a set of choices of measurement operators to see if there is any good choice for implementing a code, since it can (relatively) quickly skip over the bad choices of measurements.

In this section, we shall present an example of applying this algorithm to search for codes on quantum systems with the structure of a graph.  That is, we assume that we have a system of qubits, 2-body Pauli measurement operators and a graph such that there is a bijection between the qubits and vertices and between the edges and measurement operators, and also such that each measurement only acts on the two qubits corresponding to the vertices adjacent to its associated edge.  Specifying a particular graph constrains the number of qubits and the types of measurement operators, but it still allows a great deal of freedom in the choice of the measurement operator at each edge.  In Figure \ref{figure:2labelings} we illustrate an example of a graph with two possible such choices of measurement operator labelings;  note that for the sake of generality we do not impose the constraint that the two operators in the 2-body measurement be identical.

\begin{figure}
\begin{center}
\includegraphics[width=1.5in]{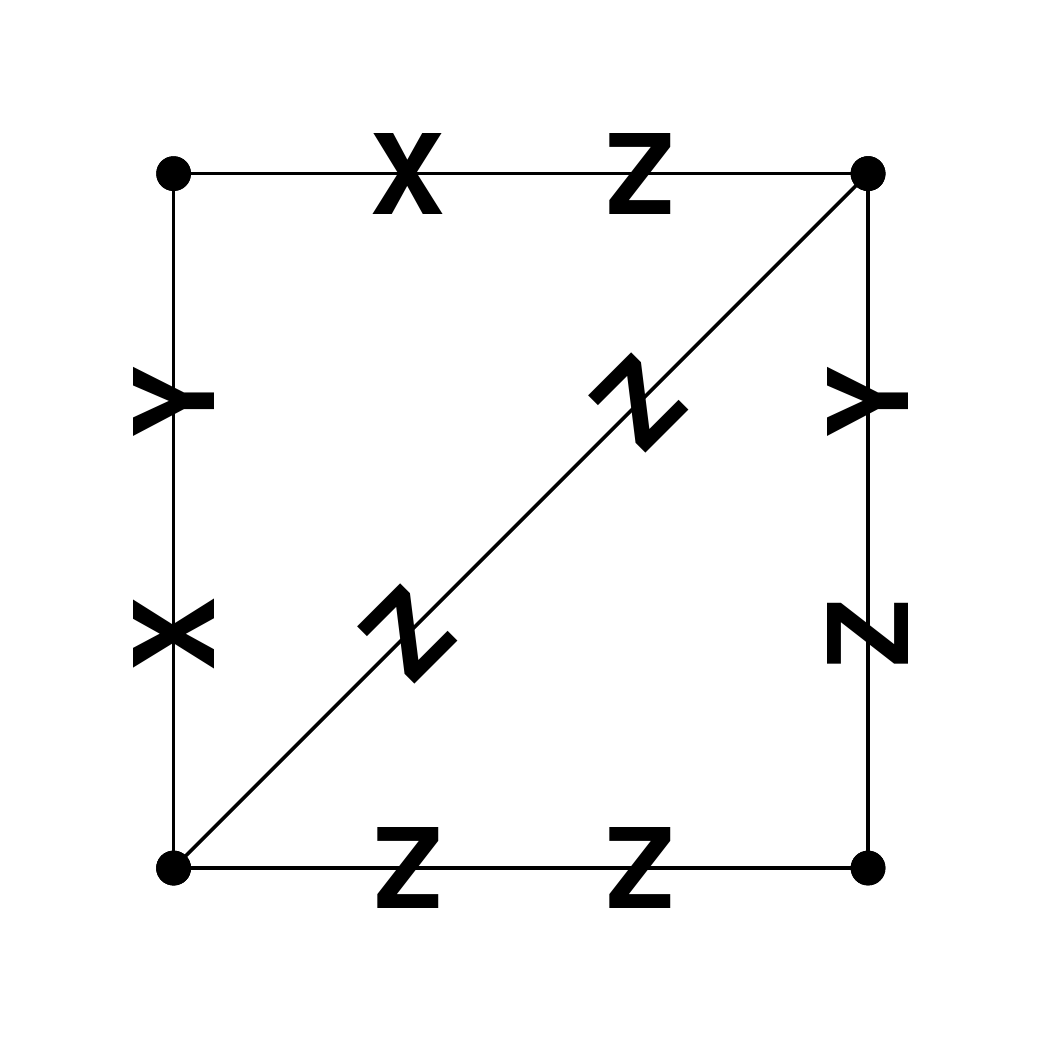}
\includegraphics[width=1.5in]{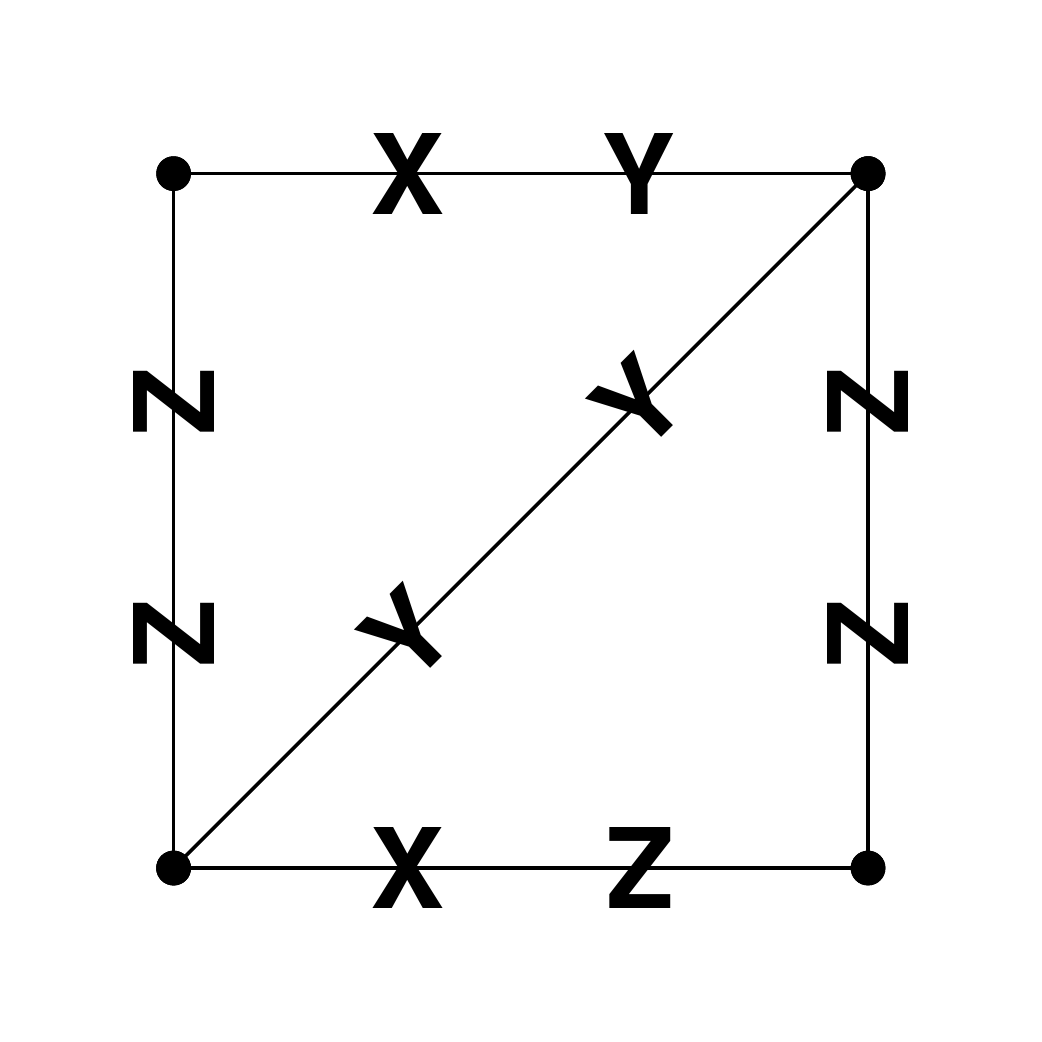}
\end{center}
\caption{
\label{figure:2labelings}
A figure illustrating two possible labelings of a graph, which correspond to two possible choices of 2-body measurements.  In both graphs we see that there are four vertices and five edges, which indicates that our system is constrained to have four qubits and five 2-body measurement operators.  The edges (without the labels) indicate the pairs of qubits on which the 2-body measurement operators are constrained to act within our system.  Within these constraints, we see in this figure two possible choices of measurement operators as specified by the two labelings of the edges of the graph.
}
\end{figure}

For reasons that will become clear, it turns out to be useful to specify choices of measurement operators in terms of ray labelings rather than edge labelings since the former is associated with vertices.  Define a \emph{ray} of a graph to be a pair consisting of a vertex and an edge adjacent to the vertex;  note that every ray can be uniquely associated with an edge, and every edge can be associated uniquely with a ray for each of its two incident vertices.  Thus, we can define a particular choice of measurement operators by labeling each ray in the graph with a single-qubit Pauli operator acting on the qubit of the incident vertex, and then letting the measurement operator associated with each edge be equal to the product of the single-qubit operators in the edge's two rays.

There is a natural symmetry of quantum codes that can be factored out to reduce the search space:  the relevant properties of the code are invariant under single-qubit rotations.  That is, transformations such as swapping the $X$ and $Z$ operators at the location of a single physical qubit in every stabilizer, gauge qubit, and logical qubit operator does not affect the code.  Thus, when labeling the rays of a vertex, exactly which ray is labeled $X$, $Y$, and $Z$ is not important;  what matters is which rays commute and which rays anti-commute.  We see therefore that we need only search over the possible ways to divide the rays into three indistinguishable groups, so that a vertex with $n$ rays only has $1+\frac{3^{n-1}-1}{2}$ relevant labelings that need to be examined.

The specific graphs we shall examine in this section are lattices generated by nine of the eleven convex vertex-uniform (also known as the ``Archimedean'') tilings of the plane --- that is, those tilings with the property that every face is convex and every vertex has the same sequence of faces.  Since these tilings have many translational symmetries, we intentionally narrow our search to the set of labelings that share the translational symmetries of the lattice\footnote{This is not to claim that there are no interesting codes that break these translational symmetries;  however, the investigation of such codes is outside the scope of this particular study.}.  Since the ray labelings must be preserved under these symmetries, we can partition the rays of the graph into equivalence classes such that two rays are equivalent if and only if they are related by a translation symmetry;  thus we see that our narrowed search space is equivalent to the space of possible labelings of each \emph{class} of rays in the lattice examined.  Since there is a symmetry that can be factored out at each vertex (as discussed previously), we note that we can likewise partition the vertices into equivalence classes of vertices related by translation symmetries.  If there are $m$ vertex equivalence classes, and every vertex has $n$ rays, then our search space consists of $\paren{1+\frac{3^{n-1}-1}{2}}^m$ total possible labelings.  In Table \ref{table:combinatorics}, we list the eleven convex vextex-uniform tilings with the number of vertex equivalence classes, the number of rays at each vertex, and the total number of labelings.  (Two of the eleven tilings, ``truncated hexadeltille'' and ``snub hextille'', had such a large number of possible labelings that we decided to exclude them from our search.)

\begin{table}
\begin{tabular}{lrrr} \toprule
Archimedean Tiling & \# Classes & \# Rays & \# Labelings \\ \midrule
quadrille & 1 & 4 & 14\\
truncated quadrille & 4 & 3 & 625\\
snub quadrille & 4 & 5 & 2,825,761\\
isosnub quadrille & 2 & 5 & 1681\\
hextille & 2 & 3 & 25\\
truncated hextille & 6 & 3 & 15,625\\
snub hextille & 6 & 5 & 4,750,104,241\\
deltille & 1 & 6 & 122\\
hexadeltille & 3 & 4 & 2744\\
truncated hexadeltille & 12 & 3 & 244,140,625\\
rhombihexadeltille & 6 & 4 & 7,529,536\\ \bottomrule
\end{tabular}
\caption[Combinatorics of the tilings]{
\label{table:combinatorics}
A table listing the number of vertex equivalence classes, the number of rays at each vertex, and the total number of labelings for each of the 11 convex vertex-uniform tilings.  By our scheme the number of labelings is equal to $\paren{1+\frac{3^{n-1}-1}{2}}^m$, where $m$ is the number of vertex equivalence classes and $n$ is the number of rays at each vertex.}
\end{table}

Note that we could furthermore refine our search to consist of those codes which also share the \emph{rotational} symmetries of the lattice.  We explicit avoid making this refinement because the existence of such codes as the quantum compass model code~\cite{Bacon:06a} indicates that there are good codes on lattices that require breaking the rotational symmetry of the lattice.  However, we can use the rotational symmetries in a different way to reduce the search space as follows.  Partition the labelings into equivalence classes such that two labelings are in the same class if and only if there is a rotational symmetry that relates them, and observe that all of the labelings in each class will give rise quantum codes with identical properties.  Thus, we can reduce our search space to ignore redundant labelings by only examining one labeling in each equivalence class.

Our search algorithm thus works in the following manner.  We start by putting a total ordering on all of the lattice labelings (after having factored out the symmetry at each vertex.)  We enumerate these labelings in order.  For each labeling, we generate new labelings by applying each rotational symmetry to the current labeling.  If any of these new labelings is less than the current labeling under our ordering, then we skip the current labeling because we know that we have already previously examined an equivalent labeling.  Although this algorithm proceeds serially through the search space, it can be parallelized by making use of $n$ walkers, each of which starts at a different labeling (from $0$ to $n-1$) and which proceed by examining the current labeling and then skipping directly to the $n^{\text{th}}$ labeling after the current one.  In Table \ref{table:count-of-labels-scanned} we list the number of non-redundant labelings for each tiling.

\begin{table}
\begin{tabular}{lrrr} \toprule
Archimedean Tiling & \# Non-redundant & \# Total \\ \midrule
quadrille & 10 & 14\\
truncated quadrille & 155 & 625\\
snub quadrille & 706,881 & 2,825,761\\
isosnub quadrille & 743 & 1681\\
hextille & 11 & 25\\
truncated hextille & 2392 & 15,625\\
deltille & 58 & 122\\
hexadeltille & 594 & 2744\\
rhombihexadeltille & 904,741 & 7,529,536\\ \bottomrule
\end{tabular}
\caption[Number of non-redundant labelings in each tiling]{
\label{table:count-of-labels-scanned}
A table listing the number of labelings for each tiling that were not redundant under rotational symmetry transformations.  This number was obtained by placing an ordering on the labelings and counting the number of labelings such that no symmetry transformation obtained a labeling less than the current labeling.  Also listed for the sake of comparison are the total number of labelings from Table \ref{table:combinatorics}.}
\end{table}

In order to preserve the rotational symmetries of the tiling, it is important that the lattice be constructed such that the center of the lattice is at a point of rotational symmetry.  There is not a single unique center point that preserves all of the rotational symmetries of a given tiling, and furthermore for many tilings there are multiple rotational symmetry groups (known as ``wallpaper'' symmetry groups), each of which has a different set of center points.  We thus chose the center of our lattice by picking the largest of the wallpaper groups present in the tiling and choosing the center to give rise to the rotational symmetries in that group\footnote{Note that this approach does not mean that we have eliminated redundant labelings resulting from \emph{all} of the symmetries in the lattice.  For example, we have not eliminated labelings which are equivalent under rotations around a different point, nor which are equivalent under a glide-reflection symmetry.  It is certainly possible to eliminate these labelings, but we choose not to do in this case in the interest of simplicity.  An implication of this is that for many codes we expect to see many labelings giving rise to them that are equivalent under symmetry transformations but not eliminated by our approach.}.  In table \ref{table:symmetries} we list the wallpaper symmetries for each of the 11 convex vertex-uniform tilings along with (where applicable) the particular symmetry group that we chose to utilize.

\begin{table}
\begin{tabular}{llr}
\toprule
Archimedean Tiling & Symmetries & Chosen\\
quadrille & p4m & p4m \\
truncated quadrille & p4m & p4m \\
snub quadrille & p4g, p4, and pg & p4 \\
isosnub quadrille & cmm & cmm \\
hextille & p6m & p6m \\
truncated hextille & p6m and p3m1 & p6m \\
snub hextille & p6 & N/A \\
deltille & p6m and p3m1 & p6m \\
hexadeltille & p6m and p3m1 & p6m \\
truncated hexadeltille & p6m & N/A \\
rhombihexadeltille & p6m & p6m \\
\bottomrule
\end{tabular}
\caption[Symmetry groups of the tilings]{
\label{table:symmetries}
A table listing the wallpaper symmetry groups (using crystallographic notation) for each of the 11 convex vertex-uniform tilings, along with the particular group that we chose for our search.  For two of the tilings no symmetry group was chosen because we decided not to search the tiling.
}
\end{table}

As is usually the case in physical systems, it is important to pay careful attention to the boundary conditions of the lattice.  In order to minimize boundary effects, we decided to put periodic boundary conditions on our lattices;  care had to be taken to impose the periodic boundary conditions in such a way as to preserve the rotational symmetry group.  For example, a boundary that only wraps from left to right and from top to bottom breaks some of the rotational symmetries for hexagonal tilings.  In Figure \ref{figure:boundaries} we illustrate how we placed the centers and the boundaries of the tilings.

\begin{figure*}
\centering
\subfloat[quadrille]{{\includegraphics[width=2in]{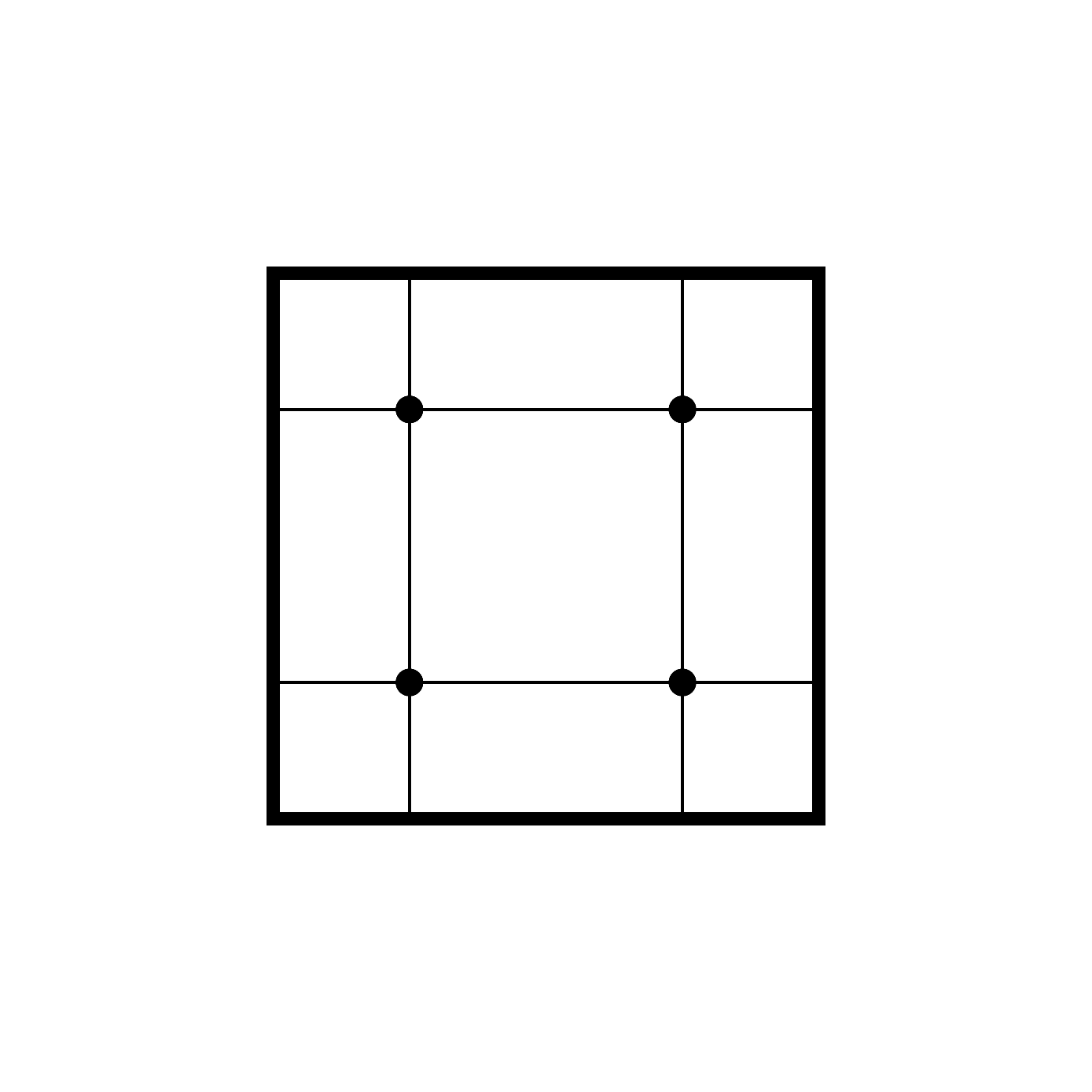}}}
\subfloat[truncated quadrille]{{\includegraphics[width=2in]{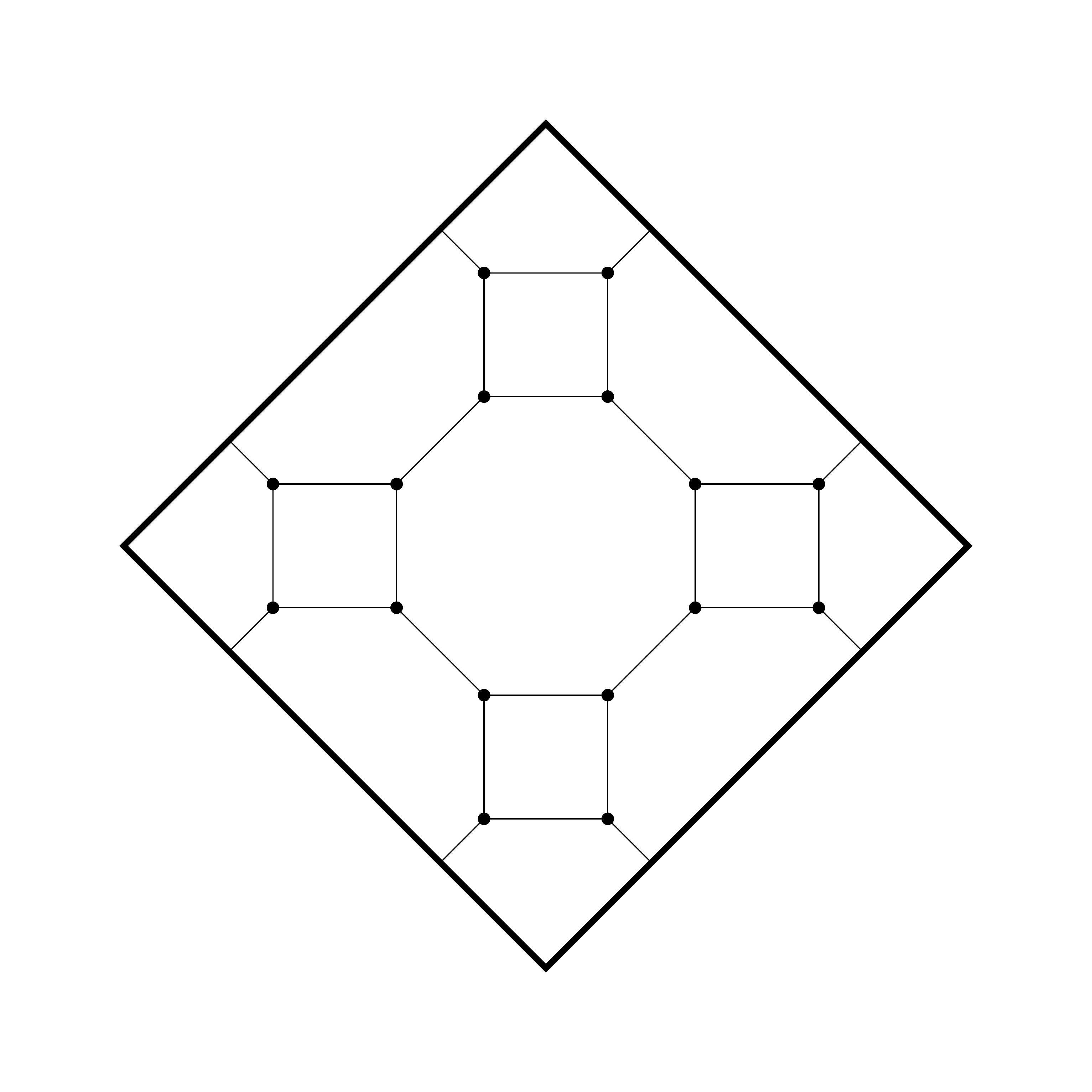}}}
\subfloat[snub quadrille]{{\includegraphics[width=2in]{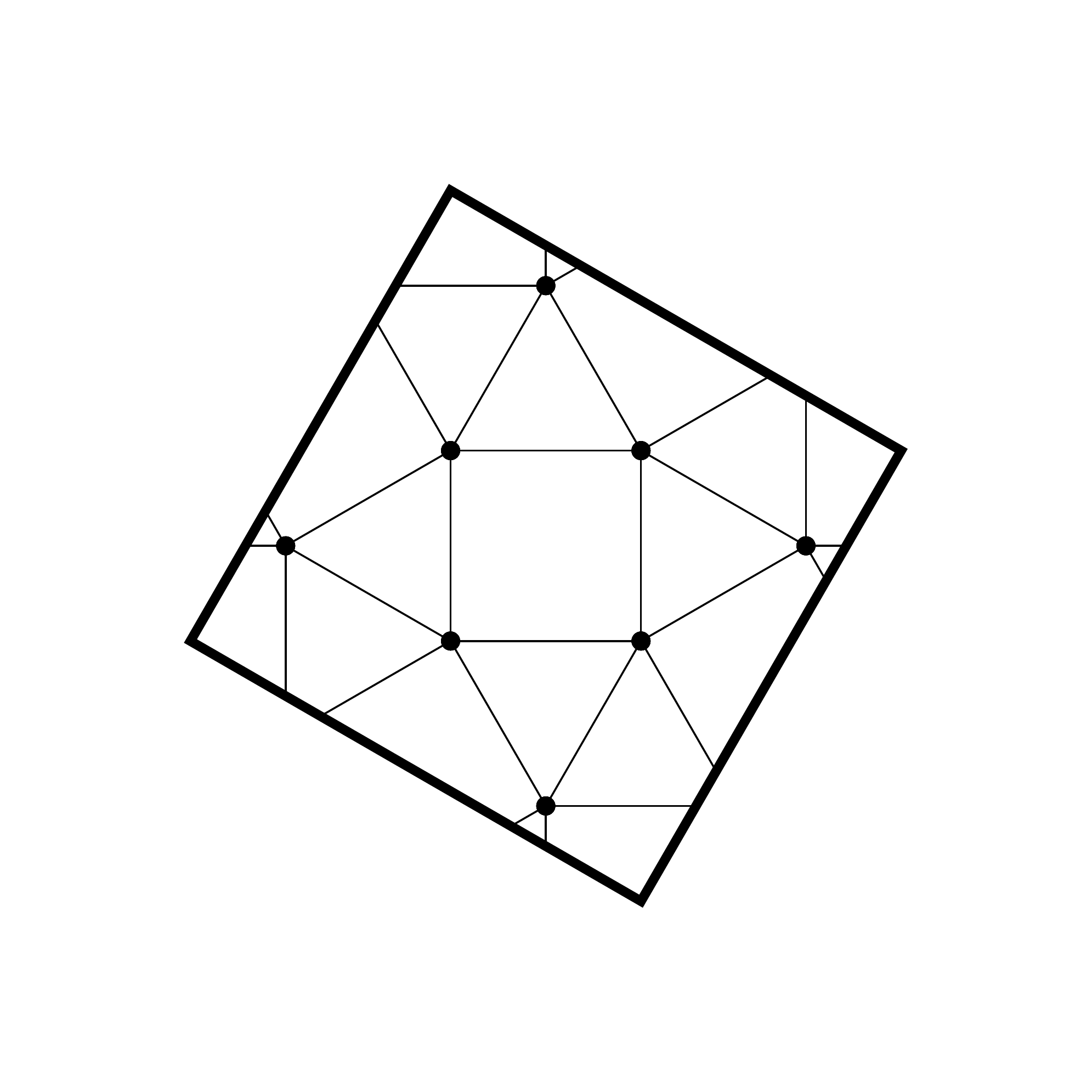}}}\\
\subfloat[isosnub quadrille]{{\includegraphics[width=2in]{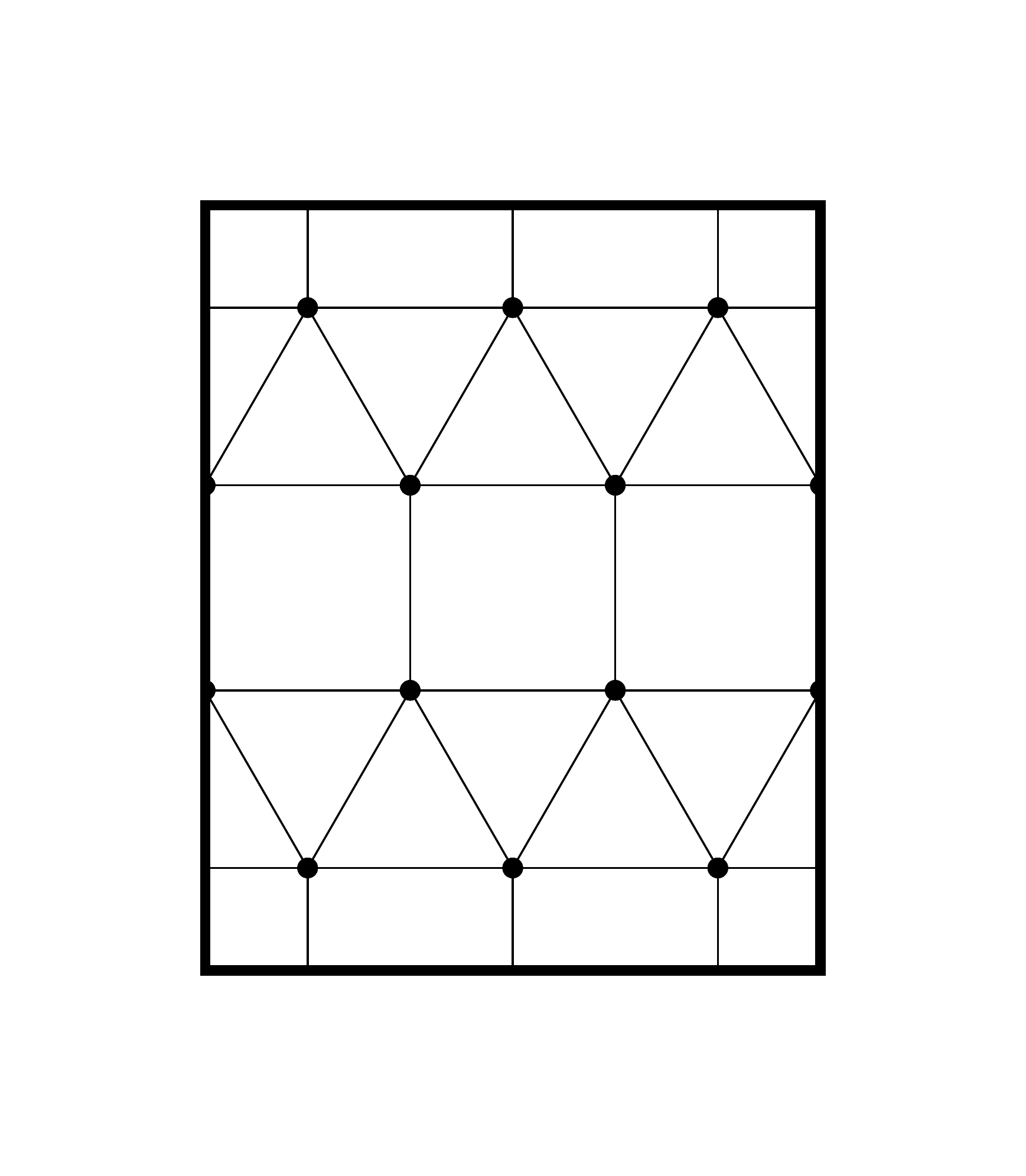}}}
\subfloat[hextille]{{\includegraphics[width=2in]{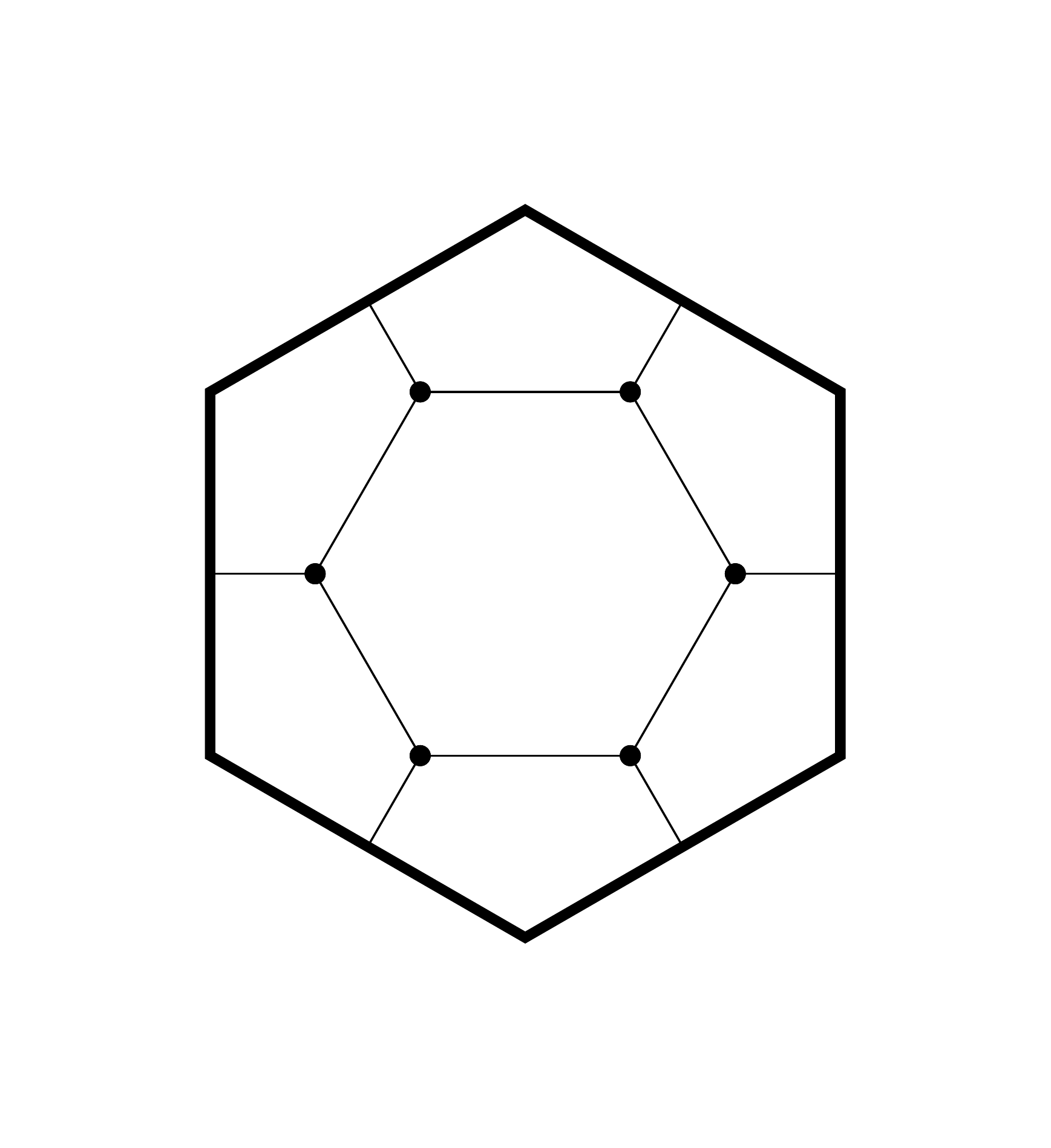}}}
\subfloat[truncated hextille]{{\includegraphics[width=2in]{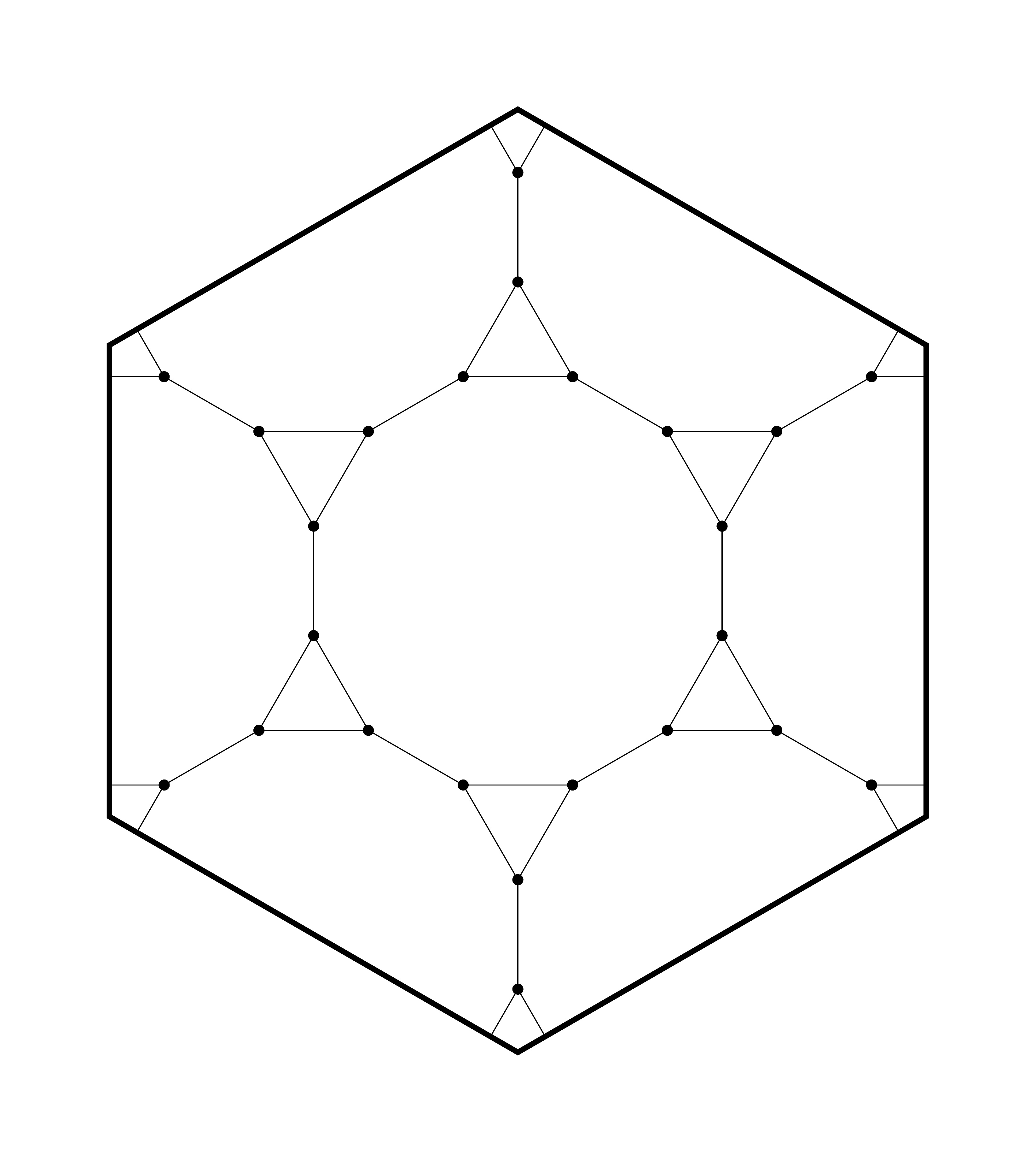}}}\\
\subfloat[deltille]{{\includegraphics[width=2in]{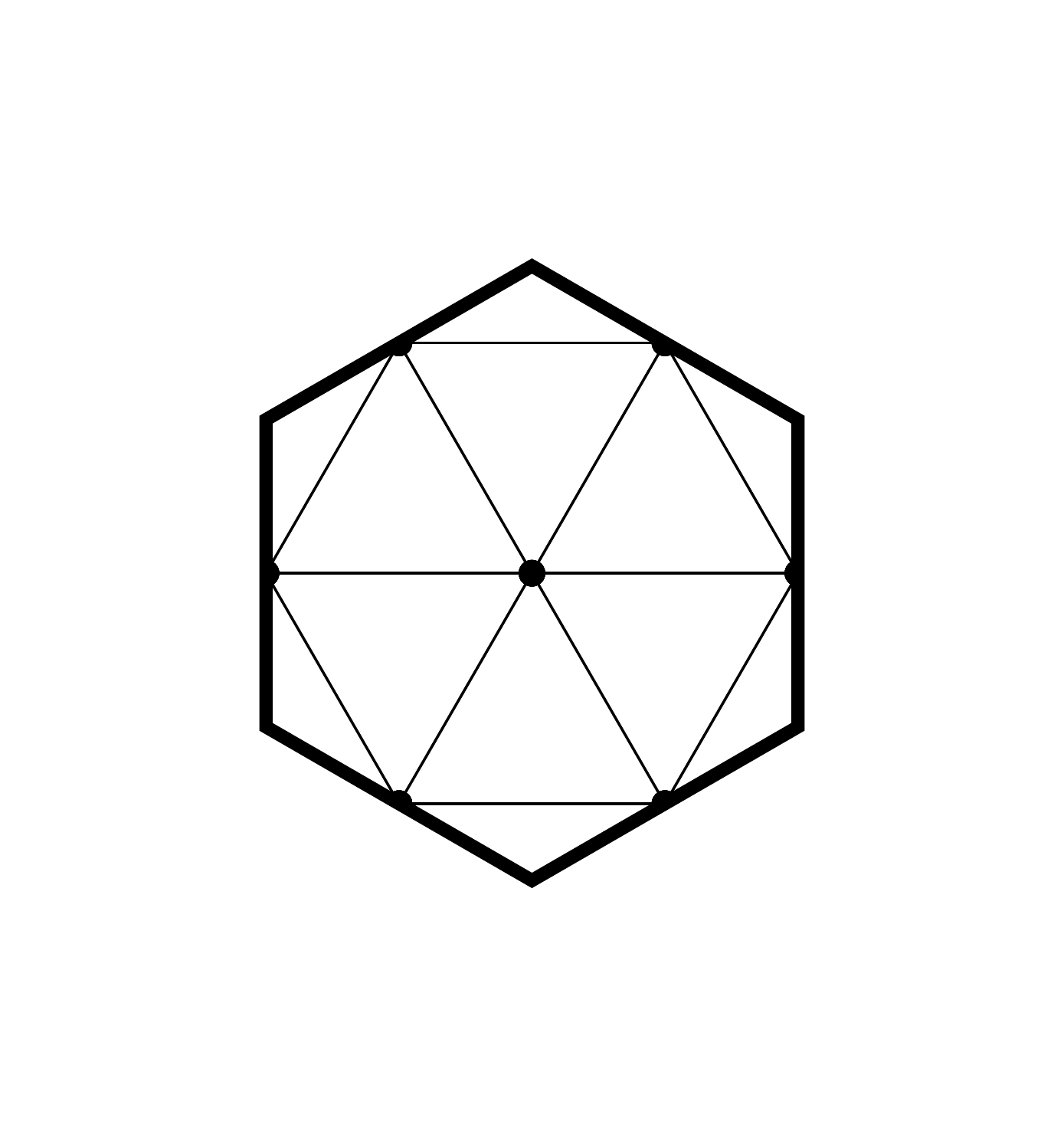}}}
\subfloat[hexadeltille]{{\includegraphics[width=2in]{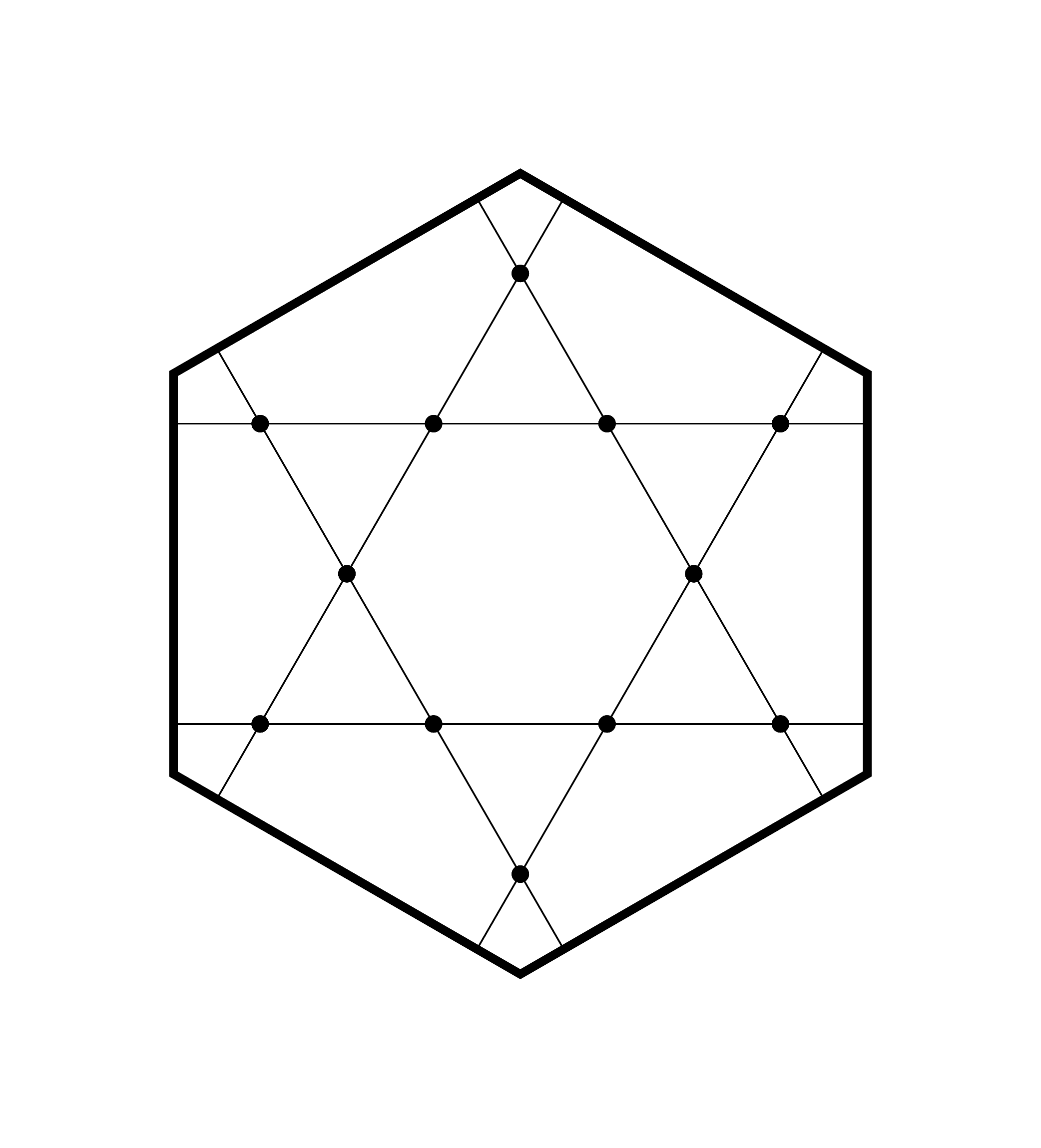}}}
\subfloat[rhombihexadeltille]{{\includegraphics[width=2in]{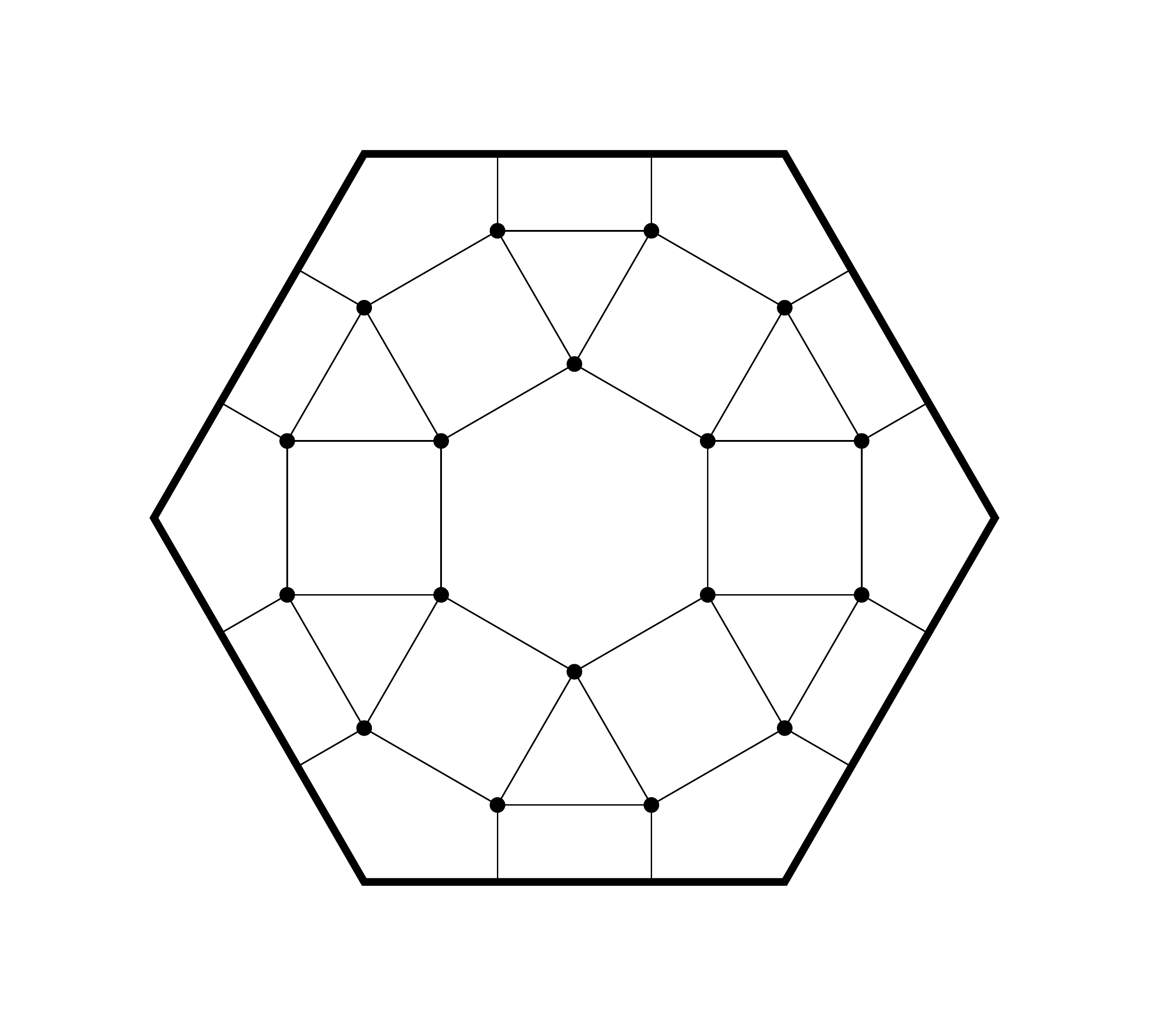}}}
\caption{
\label{figure:boundaries}
A figure illustrating how we placed the centers and boundaries for each of the tilings that we scanned.  The boundaries are periodic, so that edges that pass through one side of the boundary wrap around to the opposite side.  Edges and vertices \emph{on} a boundary are merged with the corresponding edges and vertices on the opposite edge.
Note that under this scheme there can be no vertices on a corner.  For most tilings this will never happen, but it turns out that in the case of the deltille tiling there are vertices on the corners when the radius (smallest distance from the center to the boundary) is three times the radius of the unit cell;  we thus ignore deltille tilings of these sizes.
}
\end{figure*}

Due to limits on our computational resources, we were limited in the size of the lattices that we could search with the algorithm.  We describe the size of the lattices using a quantity we call the `radius', which is an integral quantity equal to the length of the lattice divided by the length of the smallest lattice defined for that tiling;  the unit radius lattices are those illustrated in Figure {figure:boundaries}.  In Table \ref{table:maximum-radius-scanned} we show the maximum radius lattice that was completely scanned (i.e., such that every possible labeling was examined by the algorithm) for each tiling.

\begin{table}
\begin{tabular}{lcr}
Tiling & Maximum Radius & \# Qubits\\
quadrille & 4 & 64\\
truncated quadrille & 6 & 576\\
snub quadrille & 5 & 200\\
isosnub quadrille & 8 & 768\\
hextille & 10 & 600\\
truncated hextille & 5 & 600\\
deltille & 8 & 256\\
hexadeltille & 3 & 108\\
rhombihexadeltille & 3 & 162\\
\end{tabular}
\caption[Lattice sizes scanned for each tiling]{
\label{table:maximum-radius-scanned}
A table in which we show the maximum radius lattice that was completely scanned (i.e., such that every possible labeling was examined by the algorithm) for each tiling.  To give a sense of the size of the lattices involved, we also list the number of physical qubits (corresponding to vertices) for the lattice with the maximum radius.
}
\end{table}

Since each labeling of every lattice results in a quantum code, we had to provide some criteria for our search algorithm to decide whether a code was interesting enough to log.  We set our criteria relatively low:  a code was deemed to be interesting if there was at least one logical qubit with distance three, that is if there was at least one logical qubit such that a single arbitrary error on that qubit can be corrected.  This was done under the reasoning that as long as some of the logical qubits in a code are sufficiently useful to us to make implementing the code worthwhile, then we should not be troubled by the fact that there might be other logical qubits that are not useful because we can always ignore them (or, equivalently, classify them as gauge qubits).
\subsection{Results} \label{sec:results}

In the previous subsection we described the search space to which we applied the algorithm in order to computationally find possible codes that can be implemented using systems with 2-body interactions and a lattice structure following nine of the eleven convex vertex-uniform tilings.  In this subsection we present the results of this search.  The codes that we found are shown in the plots appearing in Figures \ref{figure:results-quadrille} and \ref{figure:results-hextille}.  No plot appears for the deltille tiling because no codes were found for that tiling.  It is worth emphasizing that these codes indicated in these figures are \emph{all} of the (useful) codes that exist for the scanned lattices of that tiling given our constraints, since we scanned every possible labeling that was not redundant under a rotational symmetry transformation about the center.

\begin{figure*}
\centering
\subfloat[]{{\includegraphics[width=2.5in]{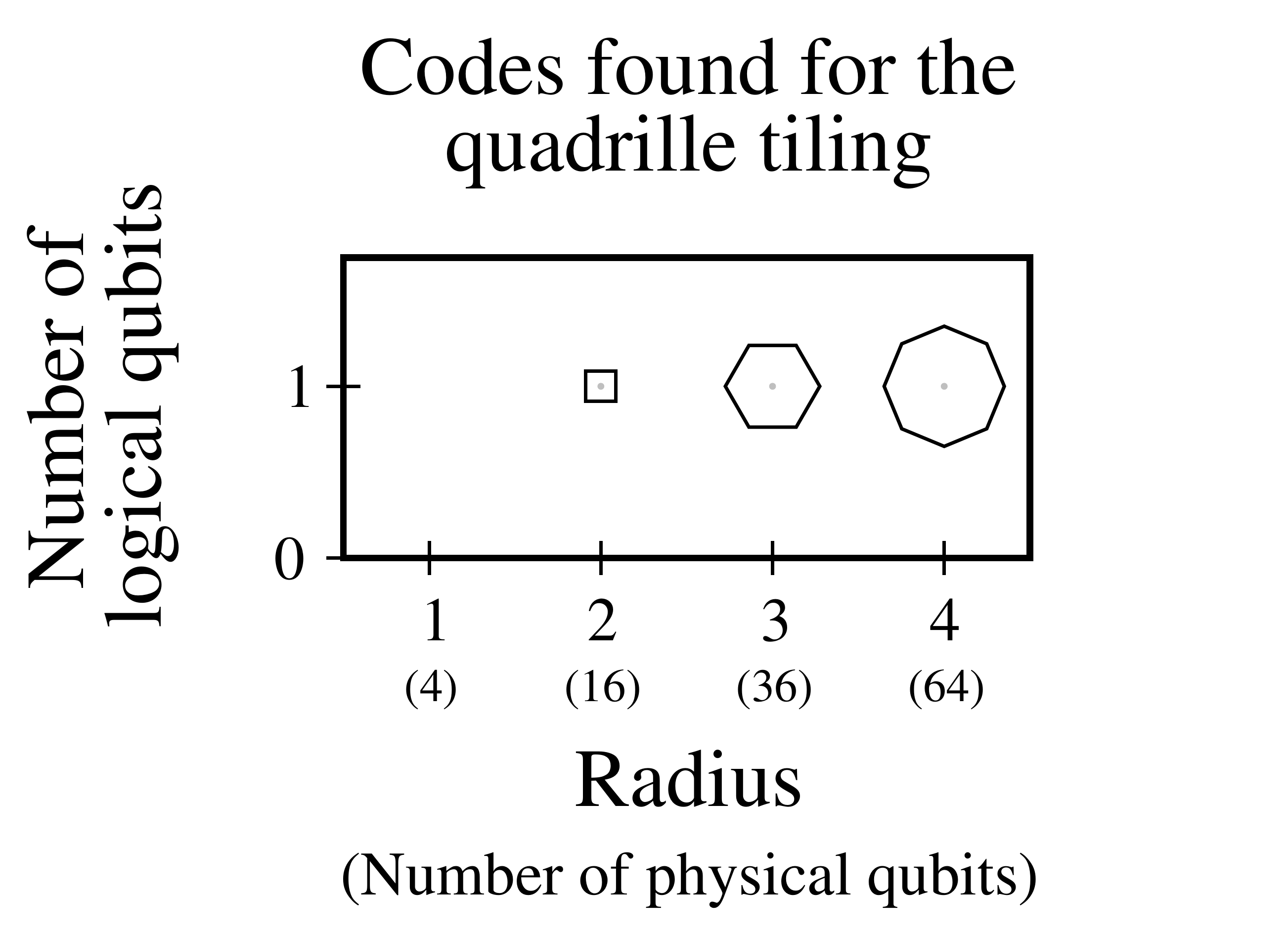}}}
\subfloat[]{{\includegraphics[width=3.85in]{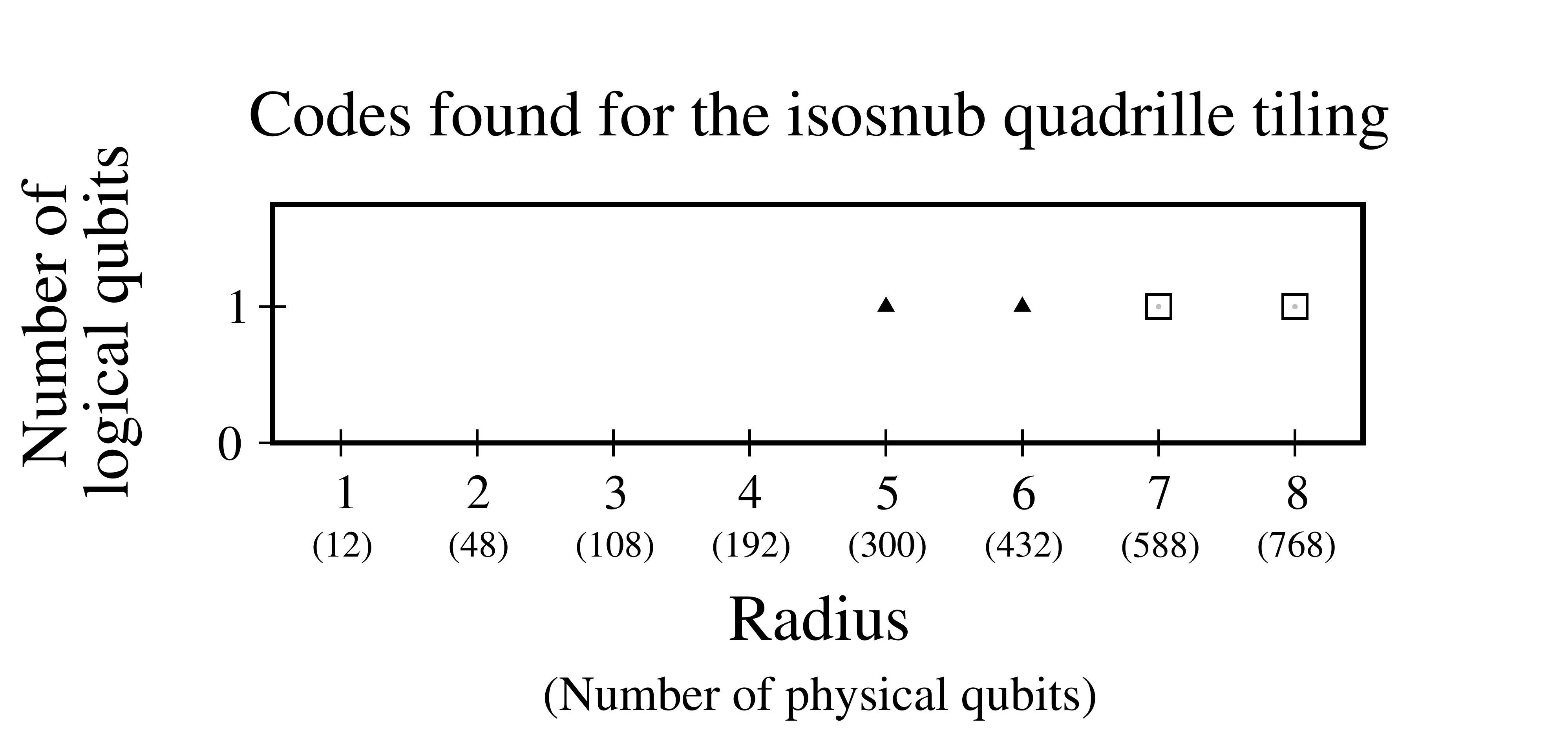}}}\\
\vspace{0.5in}
\subfloat[]{{\includegraphics[width=3.17in]{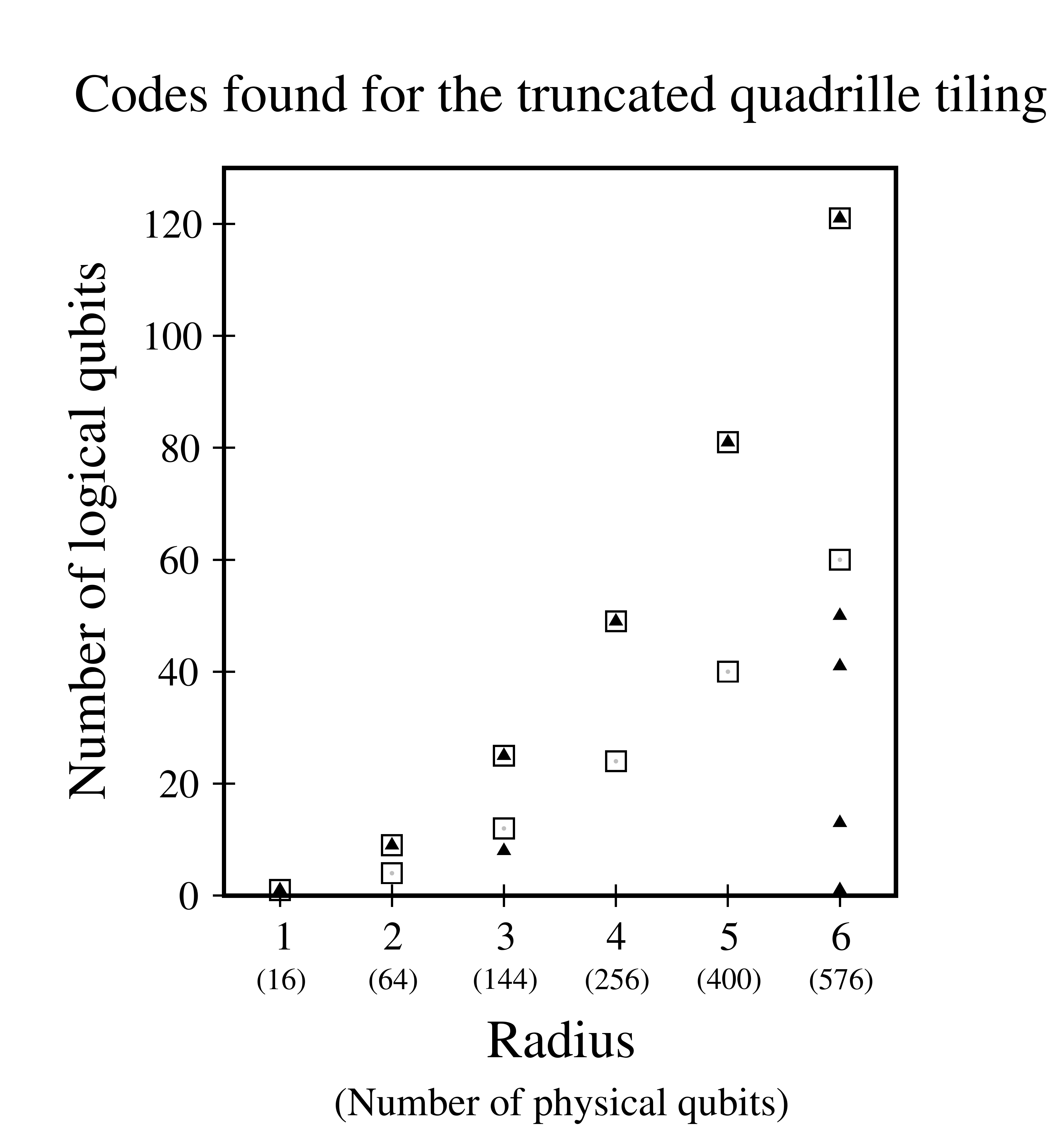}}}
\hspace{0.35in}
\subfloat[]{{\includegraphics[width=2.83in]{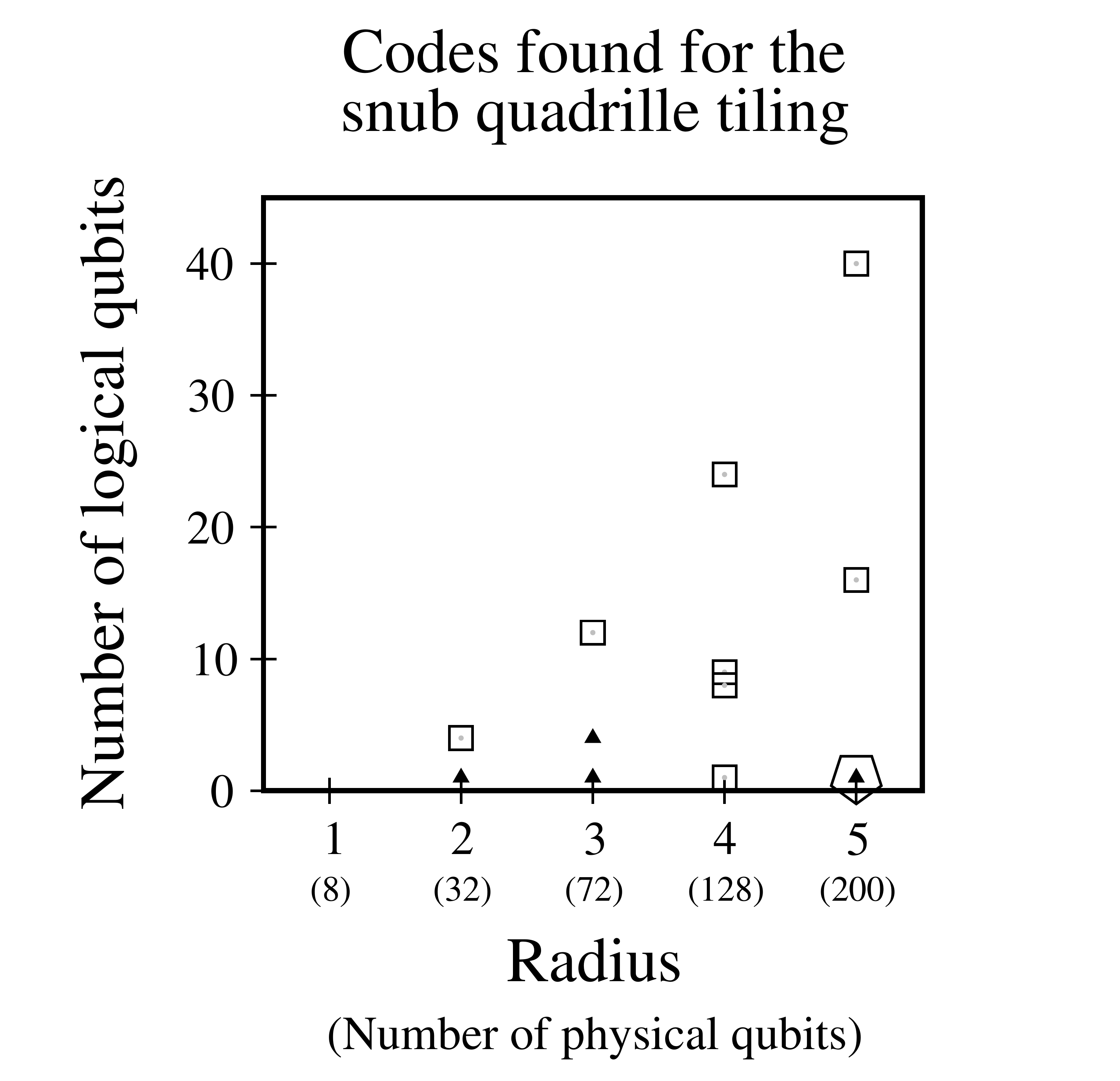}}}\\
\caption{
\label{figure:results-quadrille}
A figure containing plots of the results from scanning the (a) quadrille, (b) isosnub quadrille, (c) truncated quadrille and (d) snub quadrille tilings.  Every polygon in the plot corresponds to a code that was found with a distance equal to the number of sides of the polygon, so that triangles indicate distance three codes, squares indicate distance four codes, etc.  The position along the $x$-axis indicates the radius of the lattice where the code was found, where the radius is an integer defined to be the length of the lattice divided by the length of the smallest possible periodic lattice for the tiling;  it also indicates the number of physical qubits in the lattice where the code was found, which appears on the $x$-axis just under the value of the radius.  The position along the $y$-axis indicates the number of logical qubits with that distance in the code.
Note that in cases where multiple codes were found for the same radius and with the same number of logical qubits, multiple polygons are drawn, so that for example in plot (c) we see several cases in which a distance 3 code (triangle) and a distance 4 code (square) were found that were in lattices with the same radius and also had the same number of logical qubits.
}
\end{figure*}

\begin{figure*}
\centering
\subfloat[]{{\includegraphics[width=4.5in]{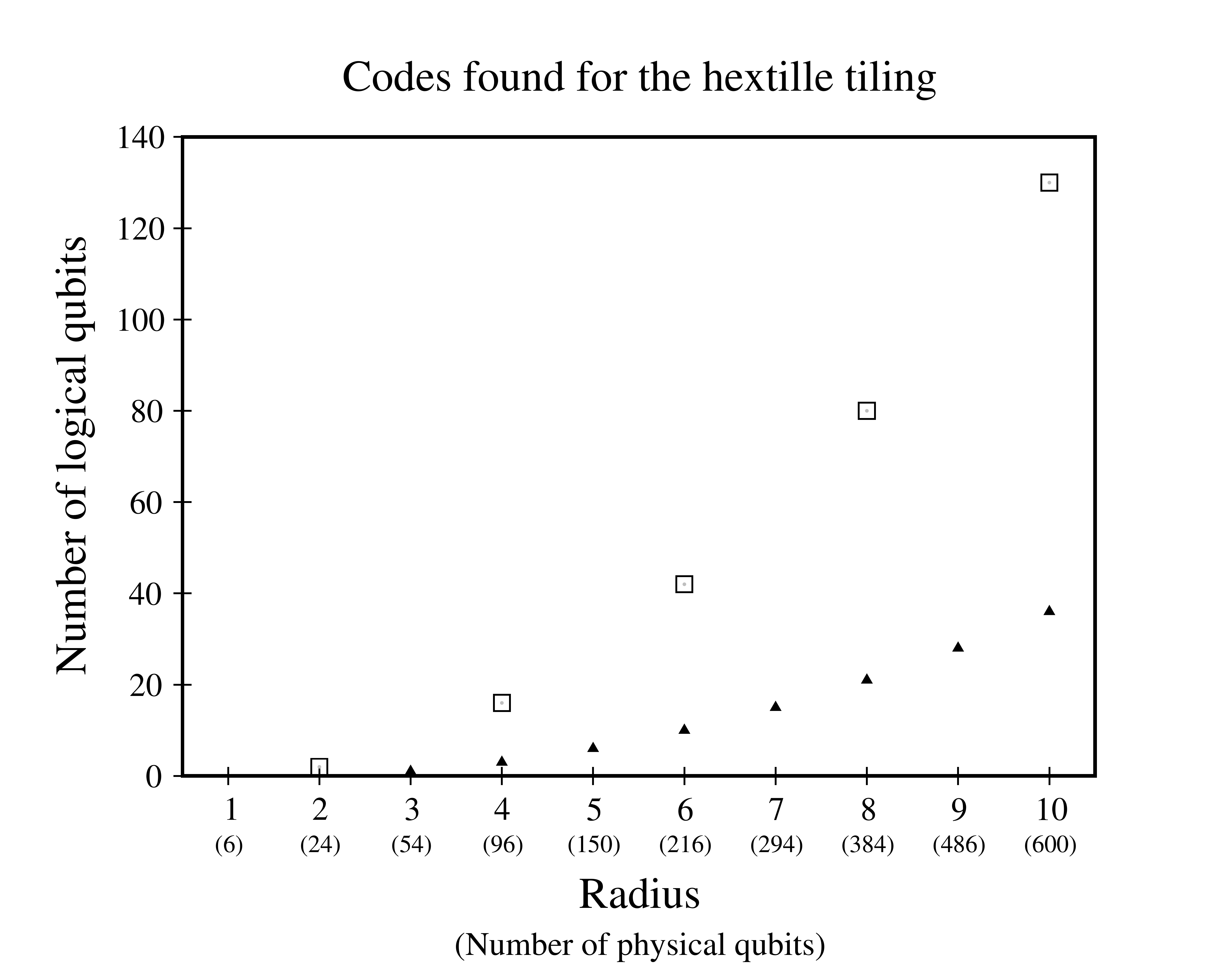}}}
\subfloat[]{{\includegraphics[width=2.17in]{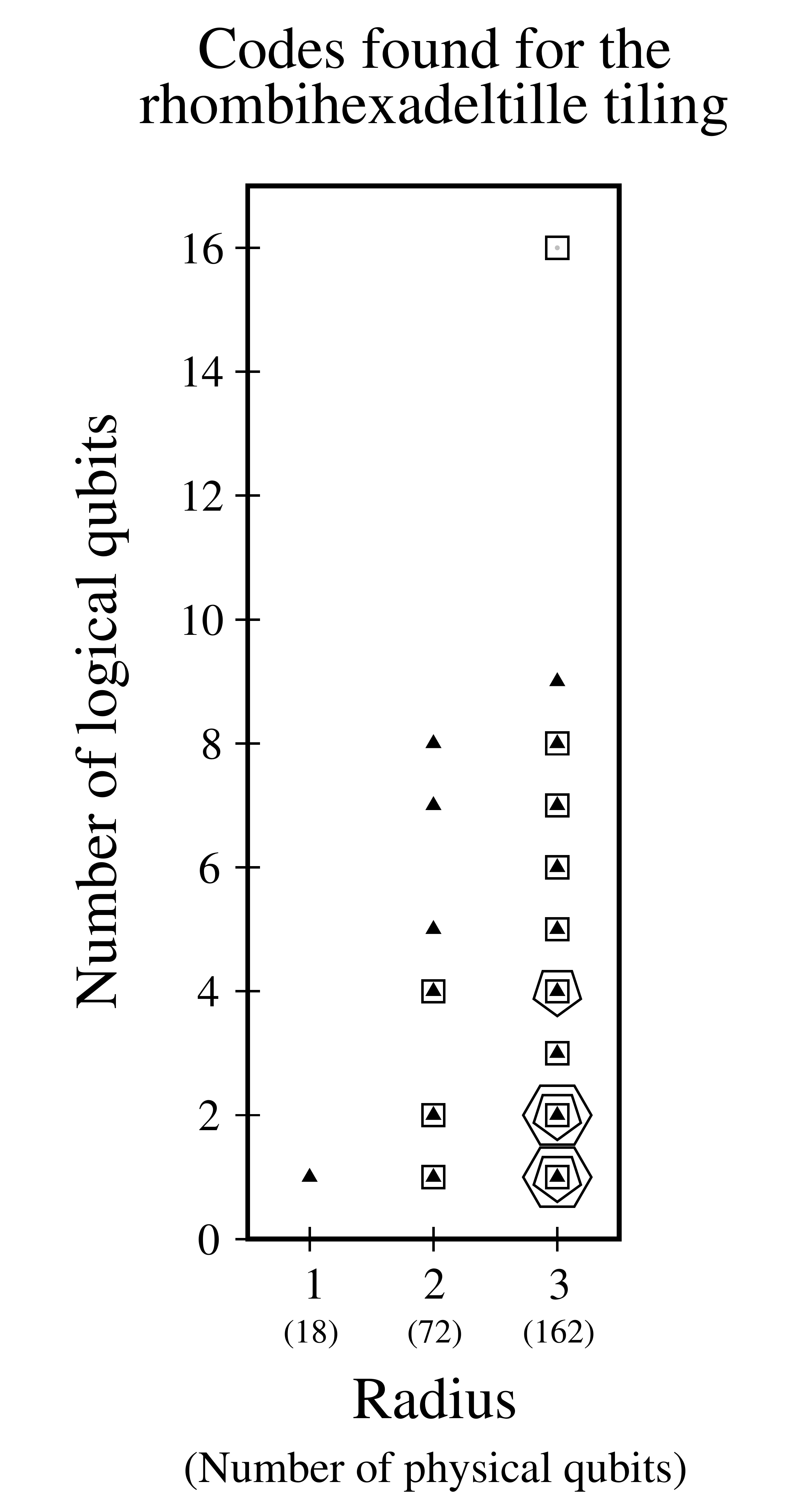}}} \\
\vspace{0.5in}
\subfloat[]{{\includegraphics[width=2.83in]{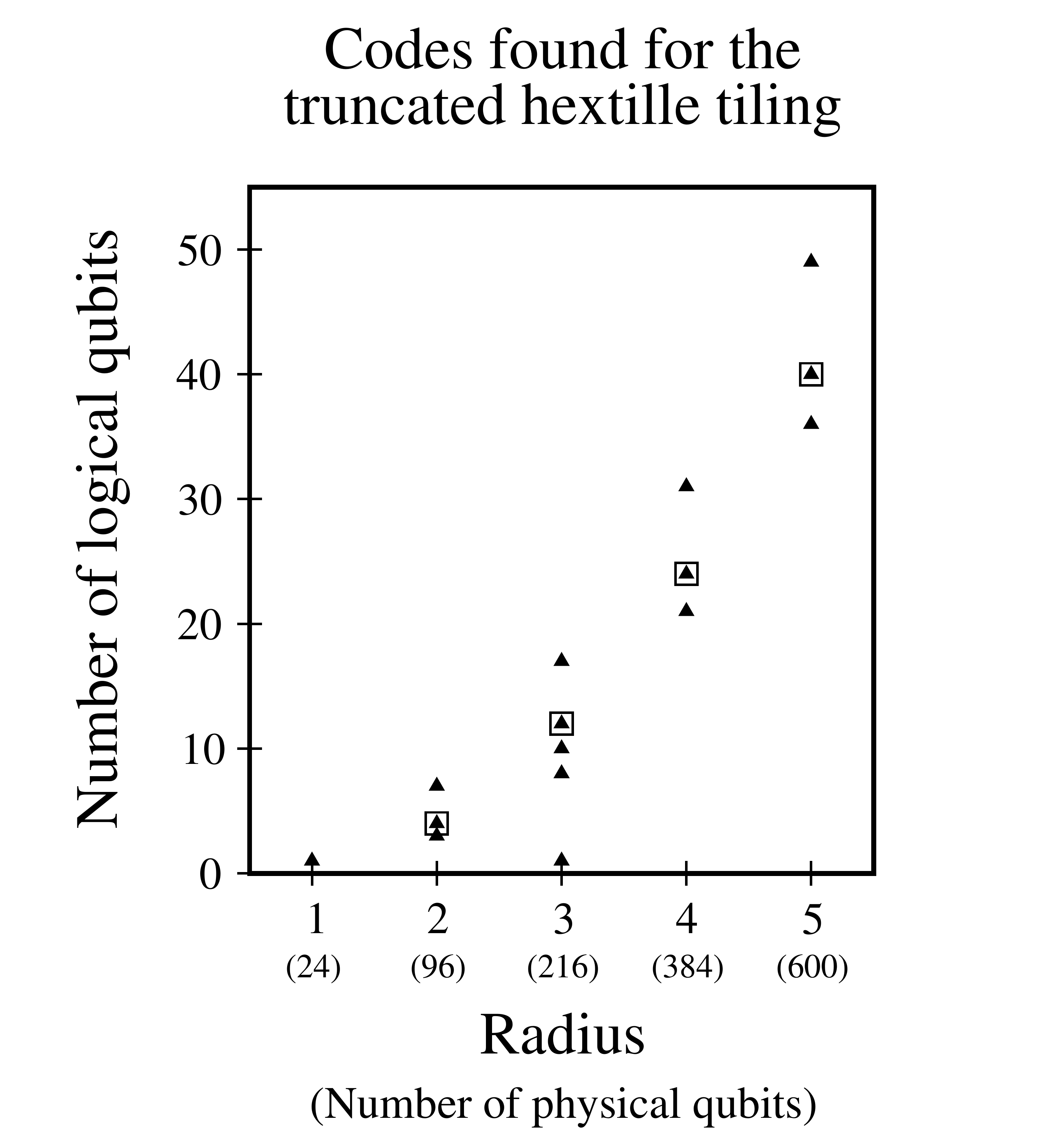}}}
\hspace{1.67in}
\subfloat[]{{\includegraphics[width=2.17in]{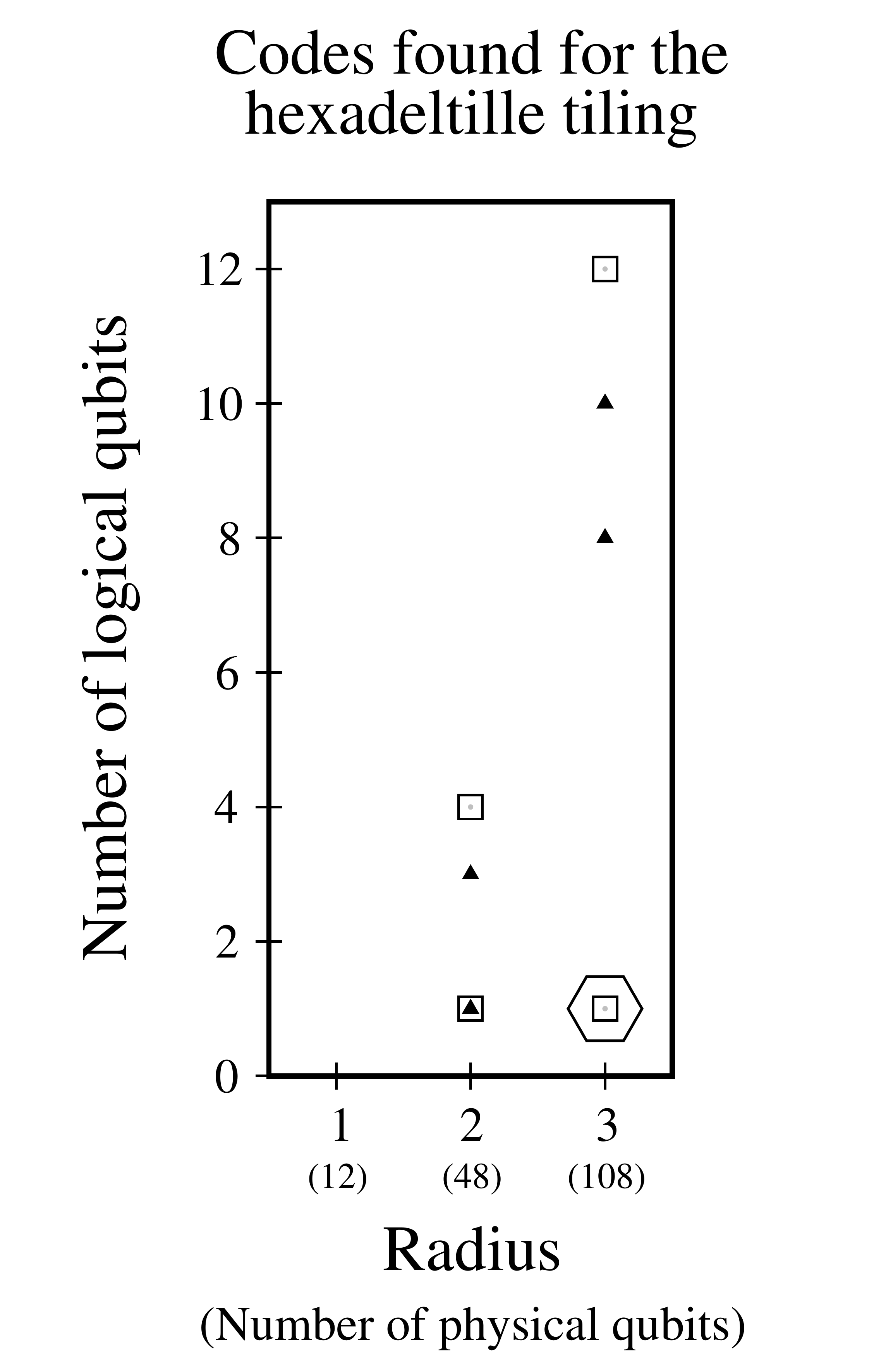}}} \\
\caption{
\label{figure:results-hextille}
A figure containing plots of the results from scanning the (a) hextille, (b) rhombihexadeltille, (c) truncated hextille and (d) hexadeltille tilings.  See the caption of Figure \ref{figure:results-quadrille} for an explanation of how to interpret these plots.
}
\end{figure*}

Observe that two kinds of trends appear frequently in the results:  codes that grow in distance but remain constant in the number of logical qubits as the radius increase, and codes that remain constant in distance but grow in the number of logical qubits as the radius increases.  The former trend appears in the quadrille, snub quadrille, isosnub quadrille, hexadeltille, and rhombihexadeltille tilings\footnote{Note that where the former trend was present, the maximum radius that we scanned was often quite limited;  this is due to the exponential explosion in the cost of finding the optimal code as a function of the distance of the code.}.  The latter trend appears in the truncated quadrille, snub quadrille, hextille, truncated hextille, hexadeltille, and rhombihexadeltille tilings.  In many of the tilings there are also codes that were found that do not seem to belong to an obvious trend.

In the follow subsections we will focus on some specifics of the results for each of the tilings.

\subsubsection{quadrille}

For the quadrille lattice, we only saw one labeling, illustrated in Figure \ref{figure:quadrille-code-labeling}, that resulted in an interesting code.  This labeling corresponds to the compass model code, and the algorithm correctly found that the distance of the code grows linearly with the radius of the lattice and is exactly equal to the square root of the number of qubits in the lattice.  This result is not terribly surprising, but it is good to see that our search technique employing the algorithm can correctly duplicate known results.

\begin{figure}
\includegraphics[width=3in]{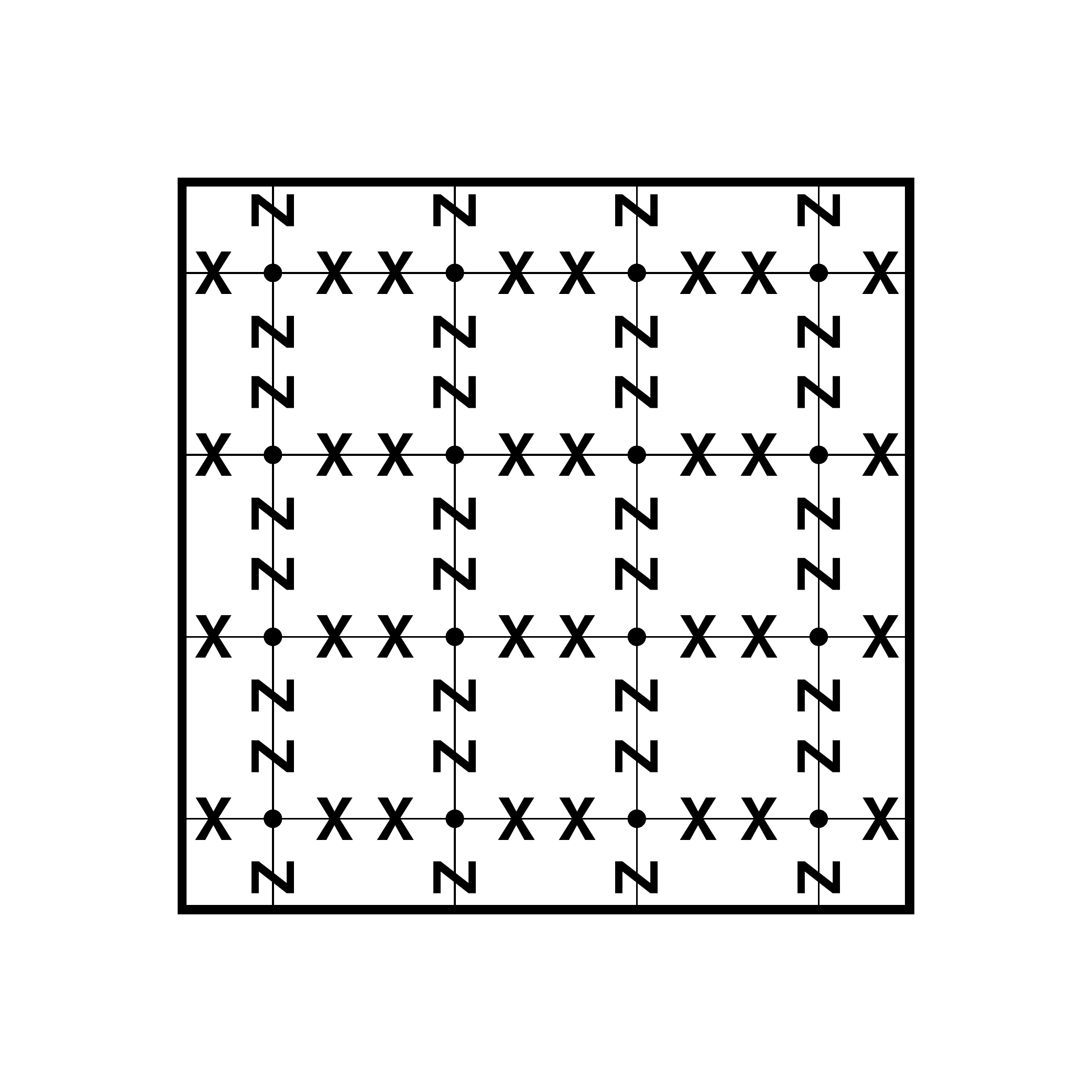}
\caption{
\label{figure:quadrille-code-labeling}
A figure illustrating the single labeling of the quadrille tiling (on a radius 2 lattice) that results in a useful quantum code.
}
\end{figure}
\subsubsection{truncated quadrille}

There are three kinds of codes that appear in this tiling where the number of qubits increases with the radius:  two where the distance is fixed at 4, and one where the distance is fixed at 3.  For the best two of these three kinds of codes, the number of logical qubits ($l$) is related to the radius ($r$) by $l=(2r-1)^2$.    Since the number of physical qubits ($n$) is given by $n=(4r)^2$, the number of logical qubits per physical qubit is thus given by $\frac{l}{n}=\paren{\half-\frac{1}{r}}^2$, a quantity which converges to $\frac{1}{4}$ as $r\to\infty$.  There were four labelings with this property that we saw in our search:  two with distance 3 qubits (illustrated in Figure \ref{figure:truncated-quadrille-code-3-labelings}), and two with distance 4 qubits (illustrated in Figure \ref{figure:truncated-quadrille-code-4-labelings}).

\begin{figure*}
\subfloat[distance 3 code labelings]{
\label{figure:truncated-quadrille-code-3-labelings}
\includegraphics[width=3.5in]{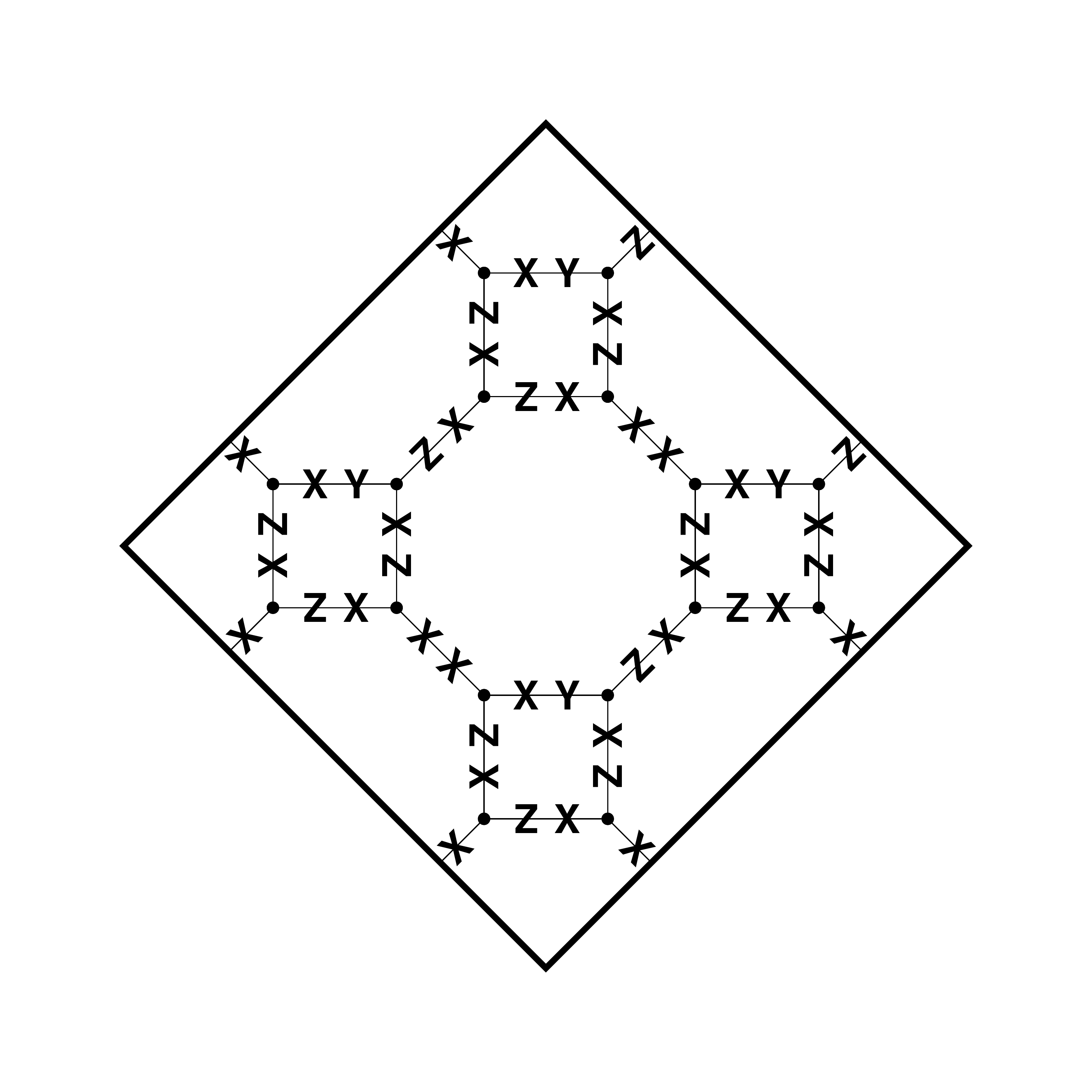}
\includegraphics[width=3.5in]{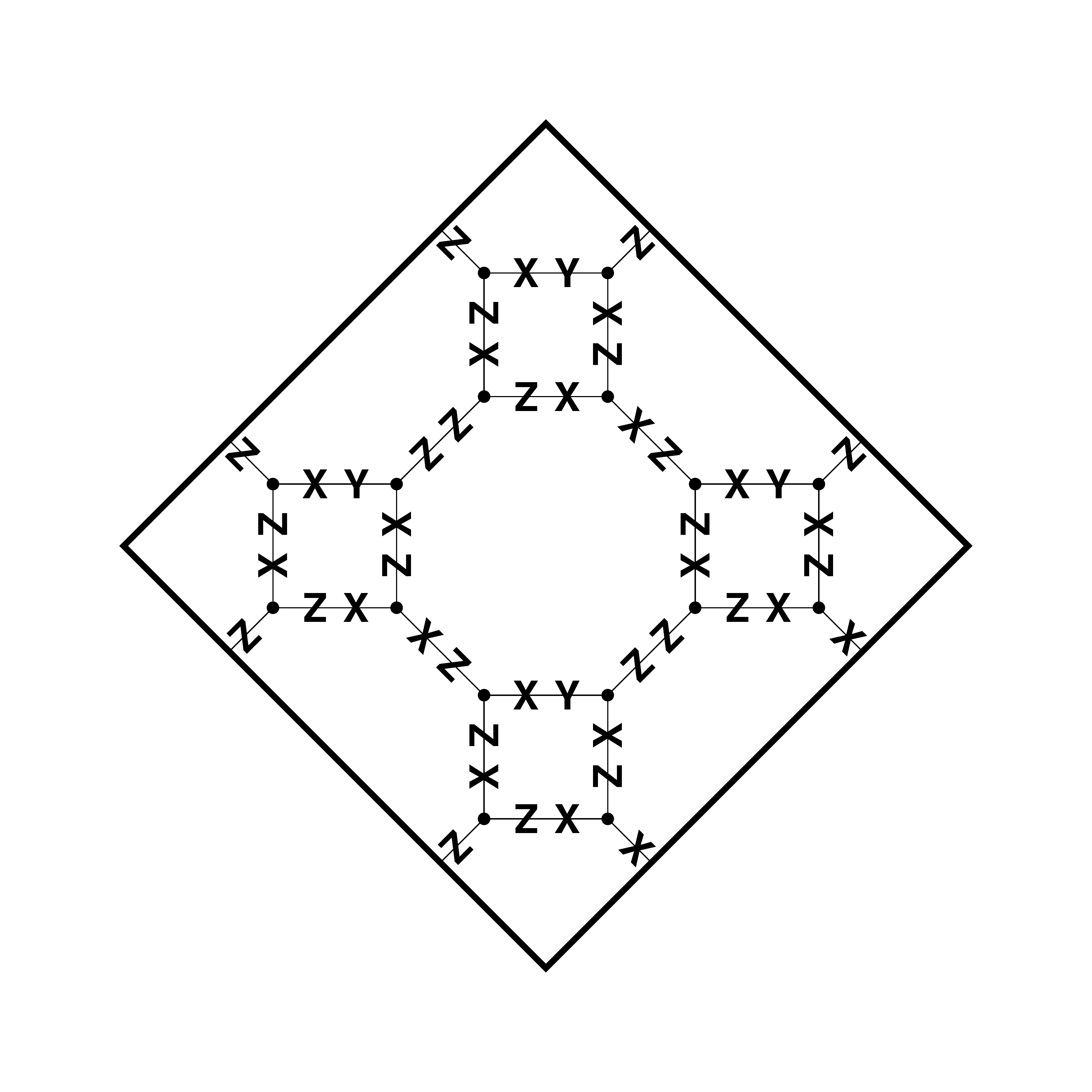}
}\\
\subfloat[distance 4 code labelings]{
\label{figure:truncated-quadrille-code-4-labelings}
\includegraphics[width=3.5in]{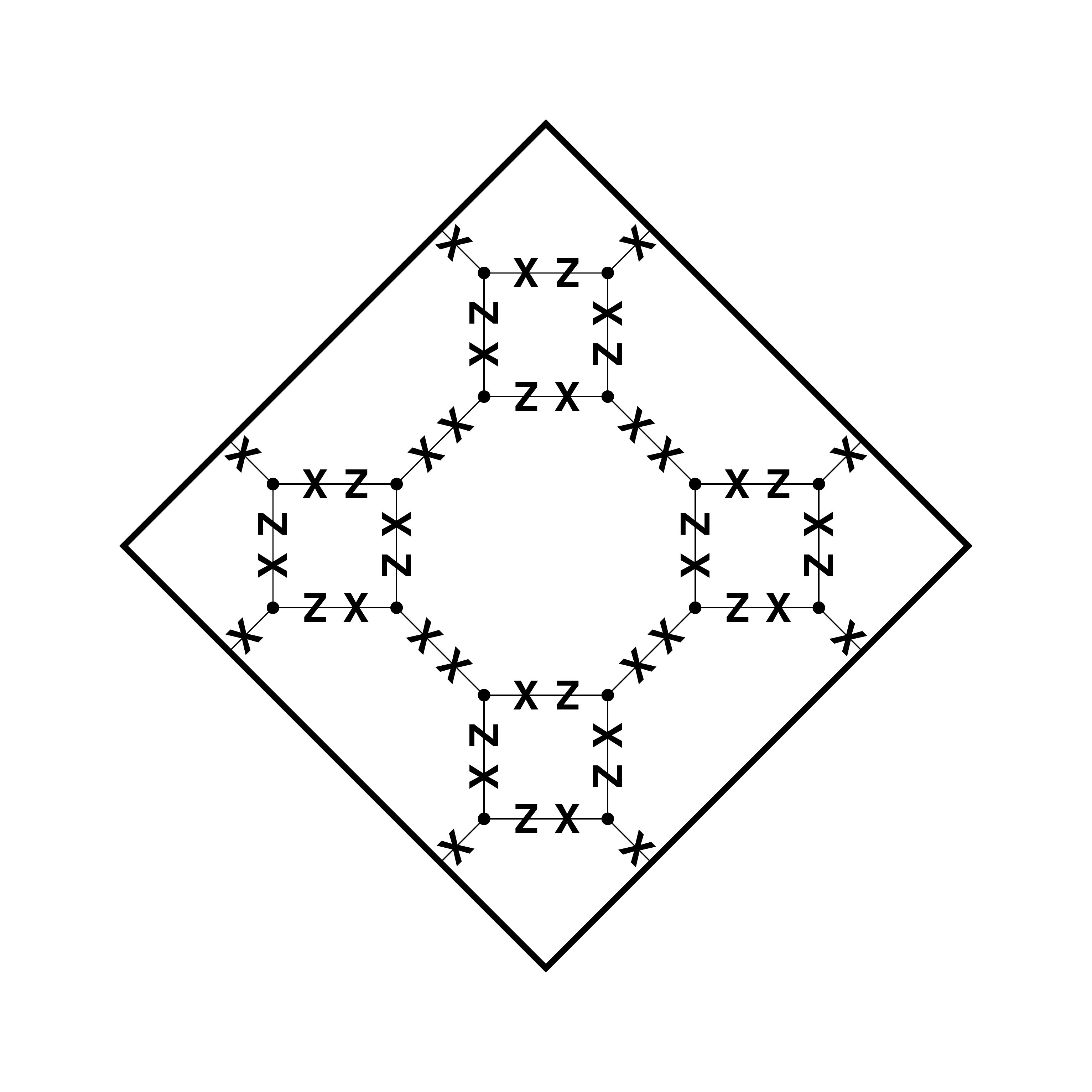} 
\includegraphics[width=3.5in]{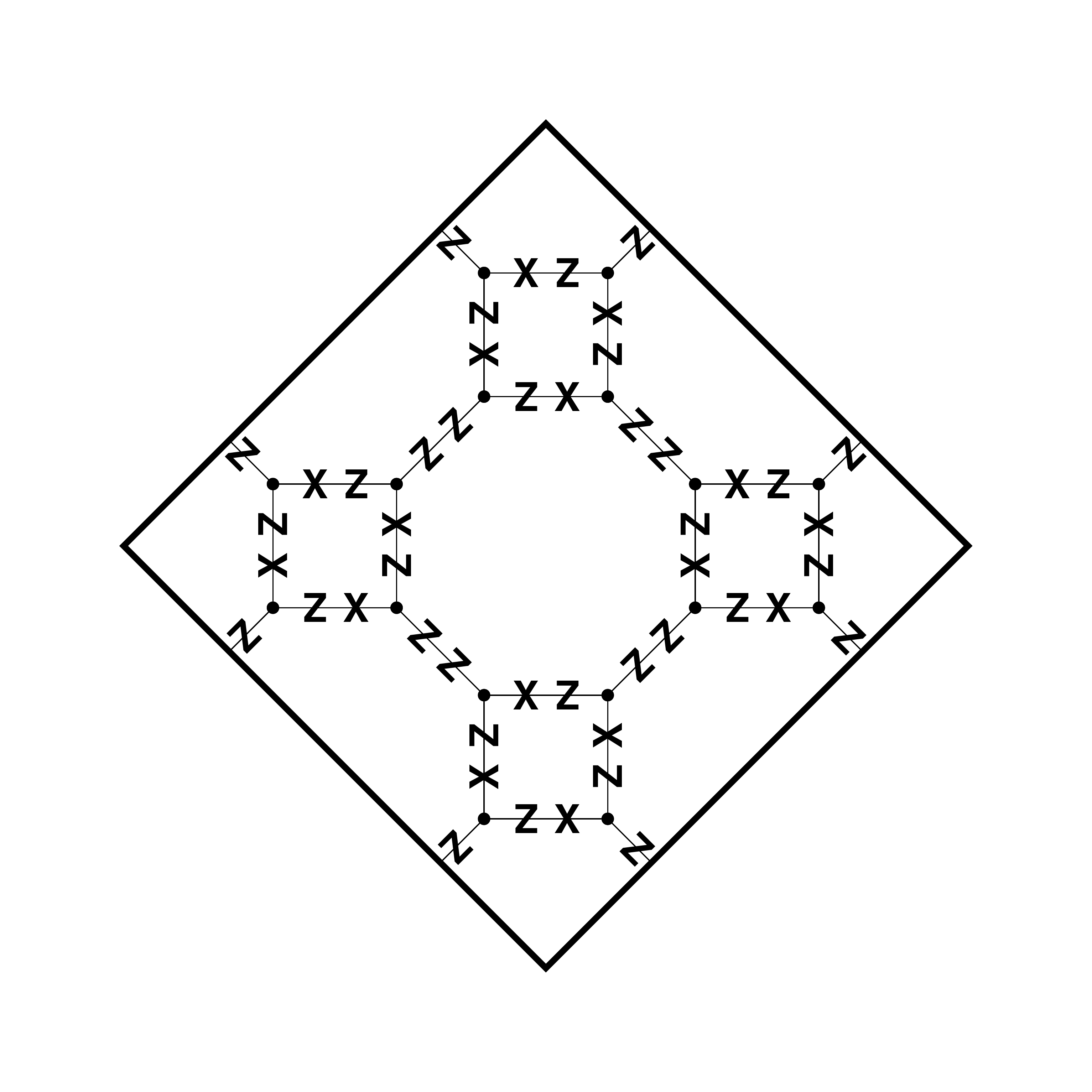}}\\
\caption{
\label{figure:truncated-quadrille-code-labelings}
A figure illustrating four labelings of the truncated quadrille tiling (on a radius 1 lattice) that result in a quantum code with distance 3/4 (in respectively \ref{figure:truncated-quadrille-code-3-labelings}/\ref{figure:truncated-quadrille-code-4-labelings}), and a number of distance logical qubits proportional to the square of the radius of the labeling.
}
\end{figure*}
\subsubsection{snub quadrille}


There are two kinds of interesting codes found in this filing.  First, we saw exactly one labeling that has the property that the distance is four and the number of logical qubits ($l$) is given by $l=2r(r-1)$, where $r$ is the radius of the lattice.  Since the number of physical qubits ($n$) is given by $n=8r^2$, this means that the number of logical qubits per physical qubit is given by $\frac{l}{n}=\frac{2r(r-1)}{8r^2}=\frac{1}{4}\paren{1-\frac{1}{r}}\to\frac{1}{4}$ as $r\to\infty$.  This labeling is illustrated in Figure \ref{figure:snub-quadrille-code-4-labeling}.

\begin{figure}
\includegraphics[width=3.5in]{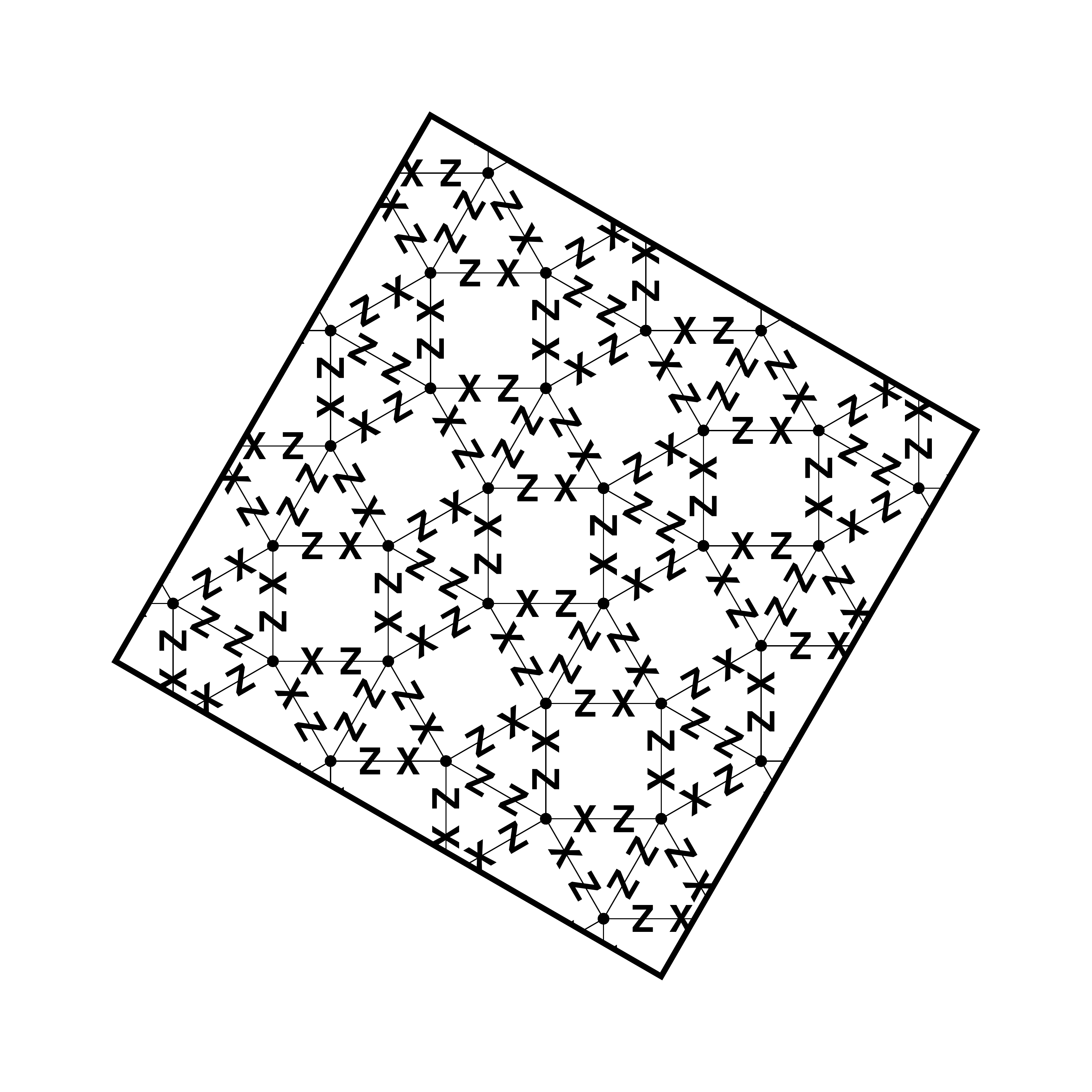} 
\caption{
\label{figure:snub-quadrille-code-4-labeling}
A figure illustrating a labeling of the snub quadrille tiling (on a radius 2 lattice) that results with distance 4 and a number of distance logical qubits proportional to the square of the radius of the labeling.
}
\end{figure}

Second, more usefully, we saw twelve labelings which result in a code that has one qubit whose distance grows with the size of the lattice.  Two of these labelings are illustrated in Figure \ref{figure:snub-quadrille-code-r-labeling}.

\begin{figure*}
\includegraphics[width=3.5in]{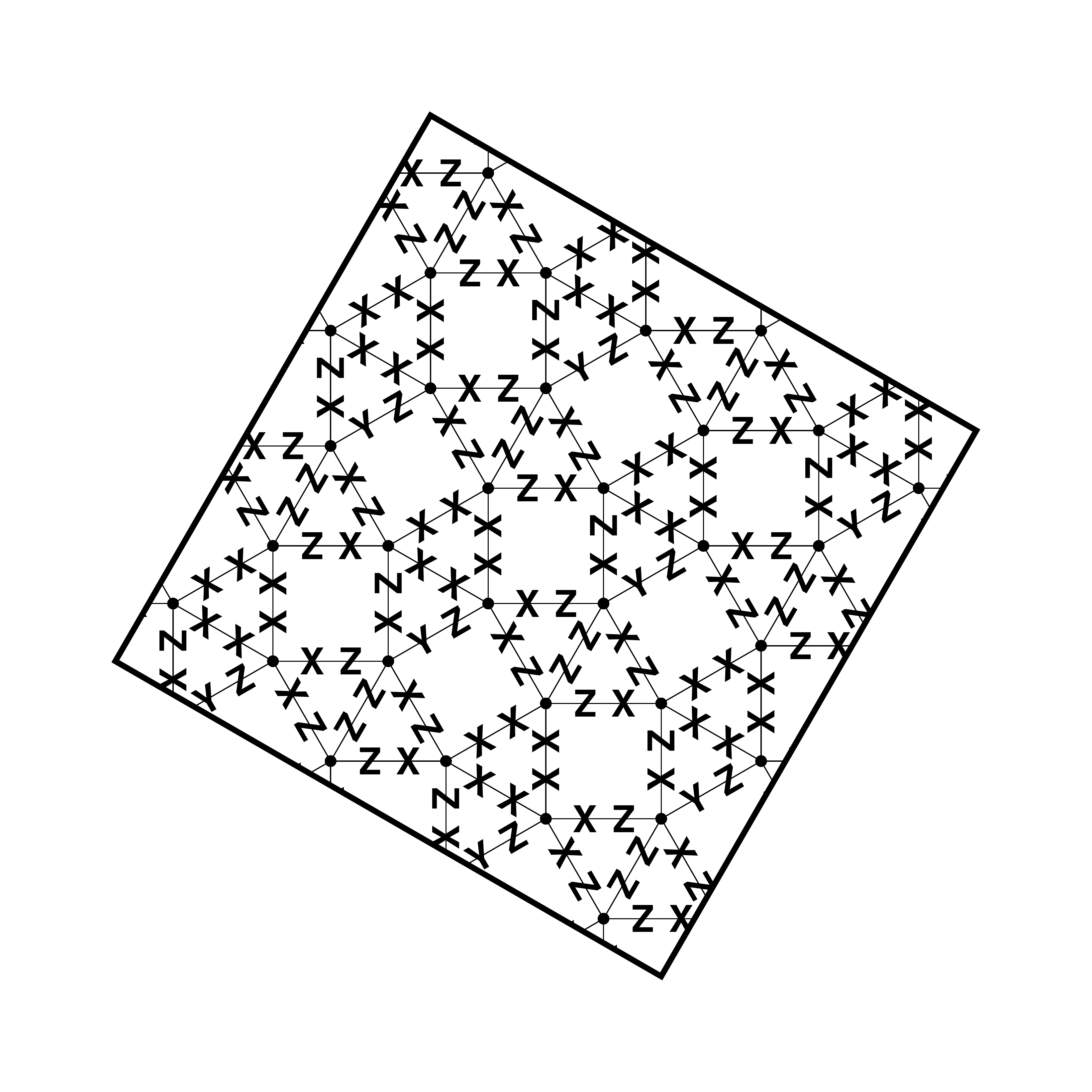} 
\includegraphics[width=3.5in]{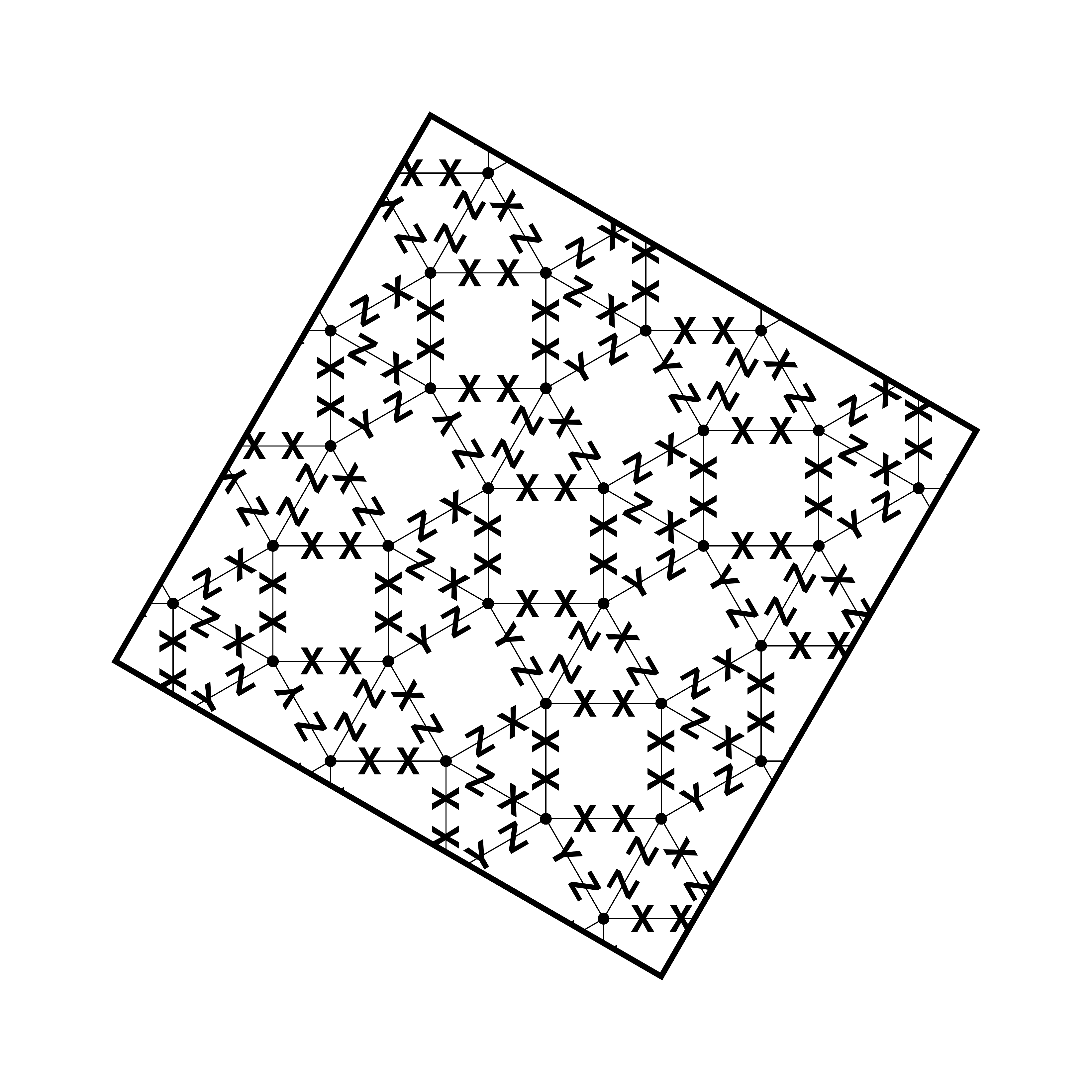} 
\caption{
\label{figure:snub-quadrille-code-r-labeling}
A figure illustrating two of the labelings of the snub quadrille tiling (on a radius 2 lattice) that result in a quantum code with a single qubit whose distance grows with the radius of the lattice.
}
\end{figure*}
\subsubsection{isosnub quadrille}

We only saw two labelings of the isosnub lattice that result in useful codes, both of which only have a single qubit that seems (assuming that the trend seen in Figure \ref{figure:results-hextille} can be extrapolated) to have a distance that grows with the radius of the lattice.  These two labeling are illustrated in Figure \ref{figure:isosnub-quadrille-code-r-labeling}.

\begin{figure*}
\includegraphics[width=3.5in]{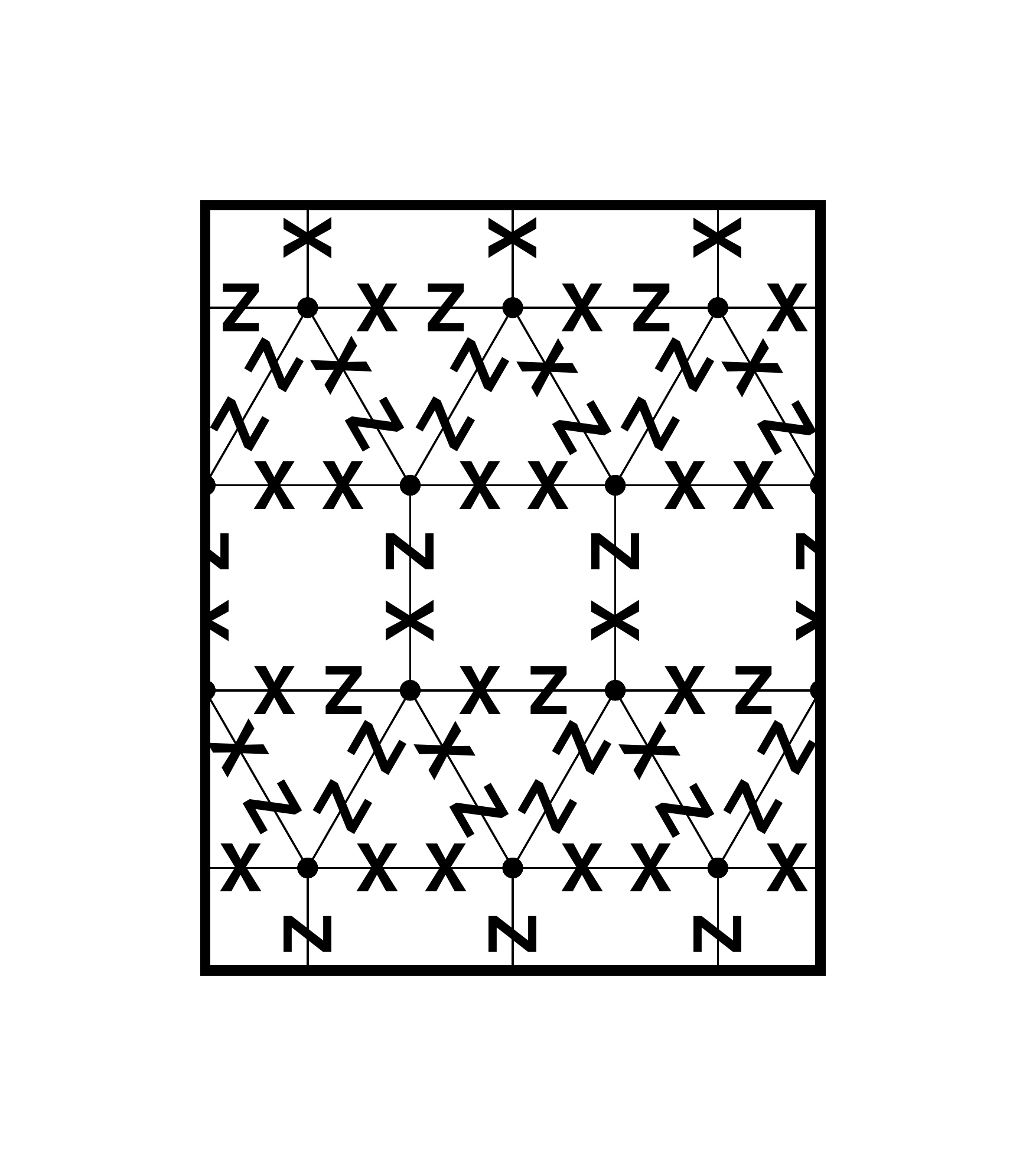} 
\includegraphics[width=3.5in]{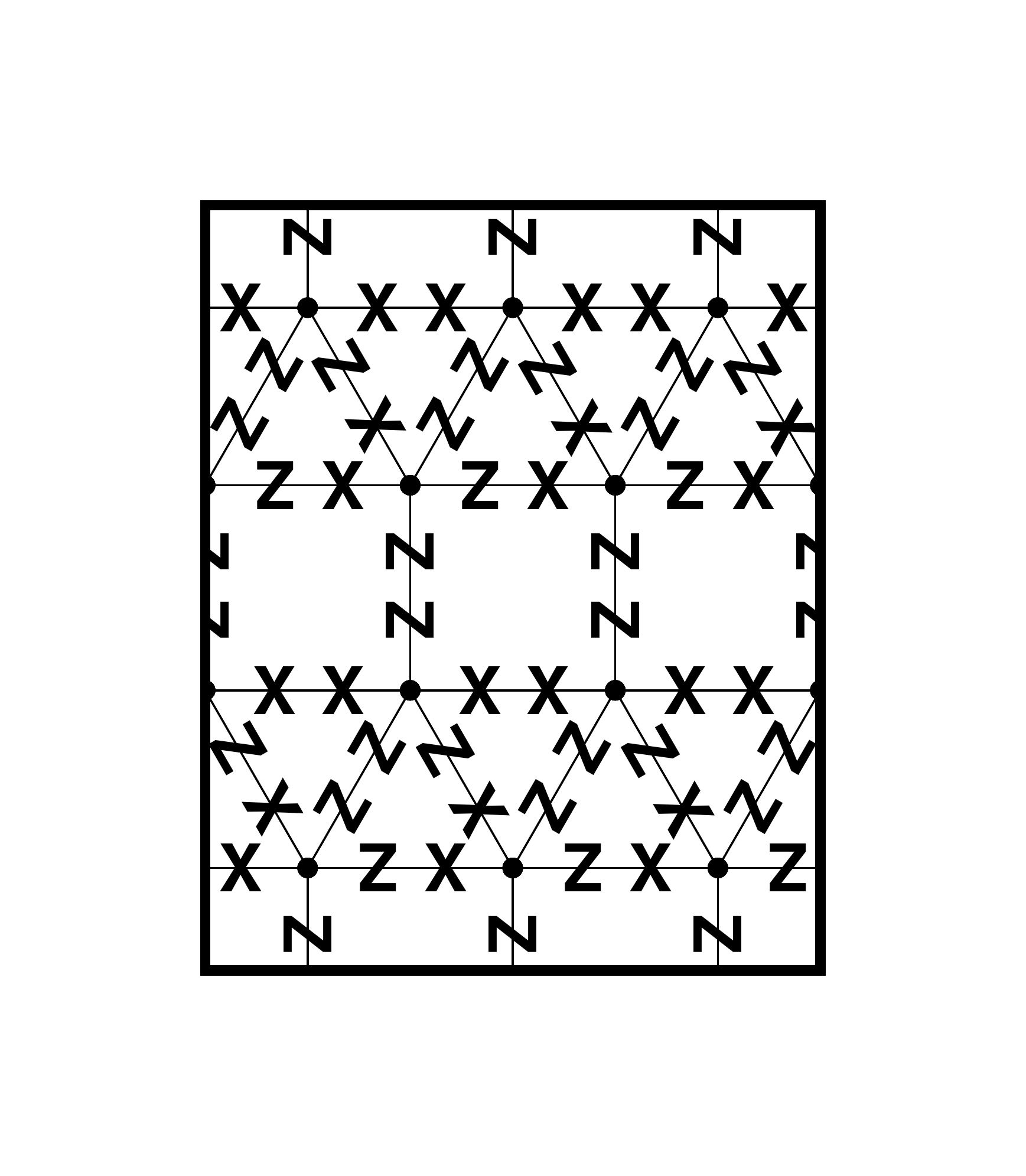} 
\caption{
\label{figure:isosnub-quadrille-code-r-labeling}
A figure illustrating two of the labelings of the isosnub quadrille tiling (on a radius 1 lattice) that result in a quantum code with a single qubit whose distance grows with the radius of the lattice.
}
\end{figure*}
\subsubsection{deltille}

We scanned this tiling up to a radius of eight;  no interesting codes were found in any of the 122 labelings.
\subsubsection{hextille}

In this tiling we saw four labelings which resulted in two kinds of interesting codes:  two of the labelings (illustrated in Figure \ref{figure:hextille-code-3-labelings}) resulted in codes of distance 3 that were present for every value of the radius, and two of the labelings (illustrated in Figure \ref{figure:hextille-code-4-labelings}) resulted in codes of distance 4 that were only present for \emph{even} values of the radius.  The former resulted in codes which had a number of logical qubits ($l_3$) given by $l_3=\frac{(r-1)(r-2)}{2}$, where $r$ is the radius, and the latter resulted in codes which had a number of logical qubits ($l_4$) given by $l_4=r(r+3)$.  Since the number of qubits ($n$) is given by $n=6r^2$, we have that the number of logical qubits per physical qubit for the distance and distance 4 codes were given respectively by $d_3=\frac{l_3}{n}=\frac{1}{12}\paren{1-\frac{1}{r}}\paren{1-\frac{2}{r}}$ and $d_4=\frac{l_4}{n}=\frac{1}{6}\paren{1+\frac{3}{r}}$;  as $r\to\infty$, we have that $d_3\to\frac{1}{12}$ and $d_4\to\frac{1}{6}$.

It is interesting to observe that there is no distance/qubit count trade-off in this tiling.  As long as the radius is even, the distance 4 code is superior in \emph{both} distance \emph{and} logical qubit count over the distance 3 code.

\begin{figure*}
\subfloat[distance 3 code labelings] {
\label{figure:hextille-code-3-labelings}
\includegraphics[width=3.5in]{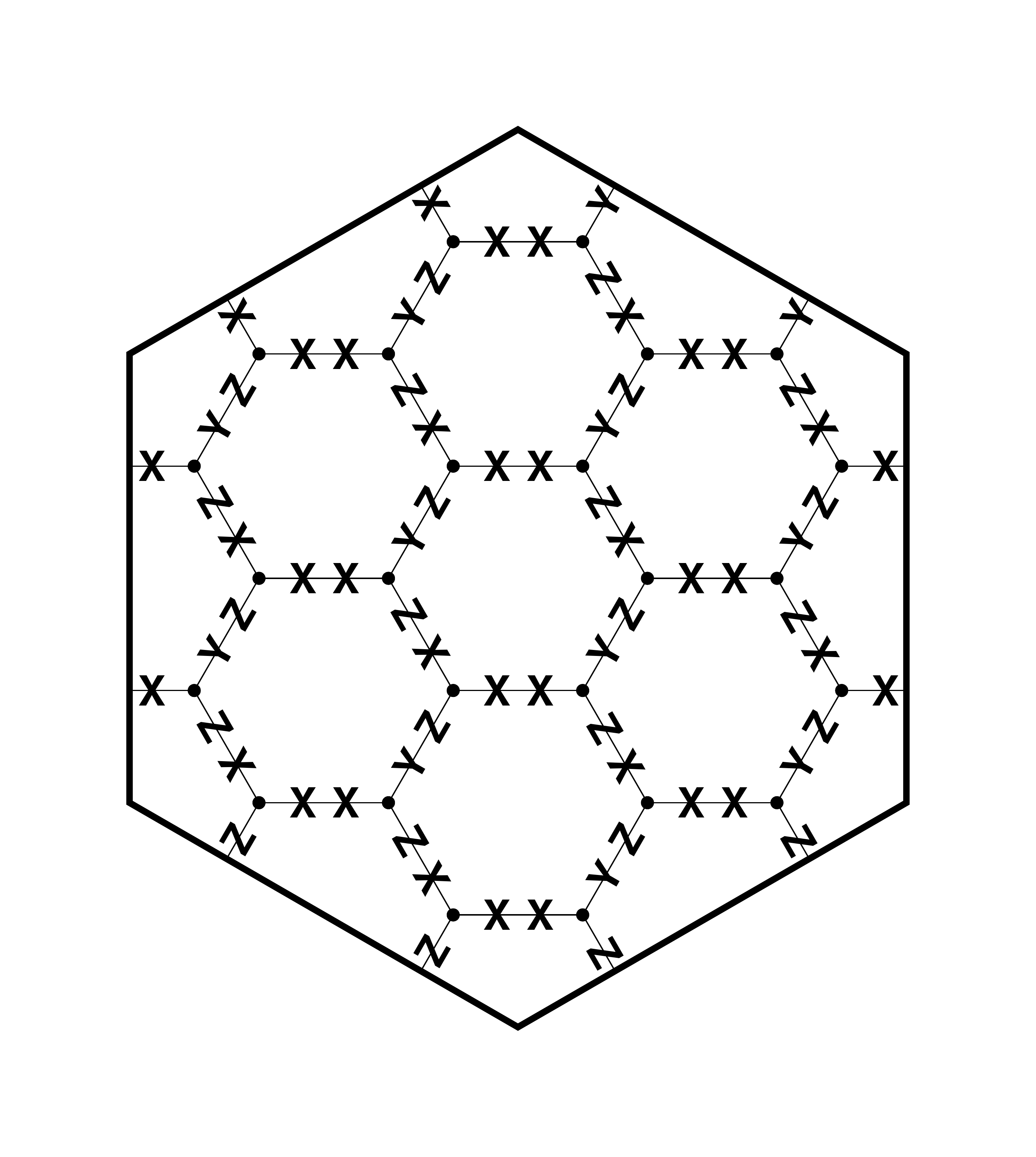} 
\includegraphics[width=3.5in]{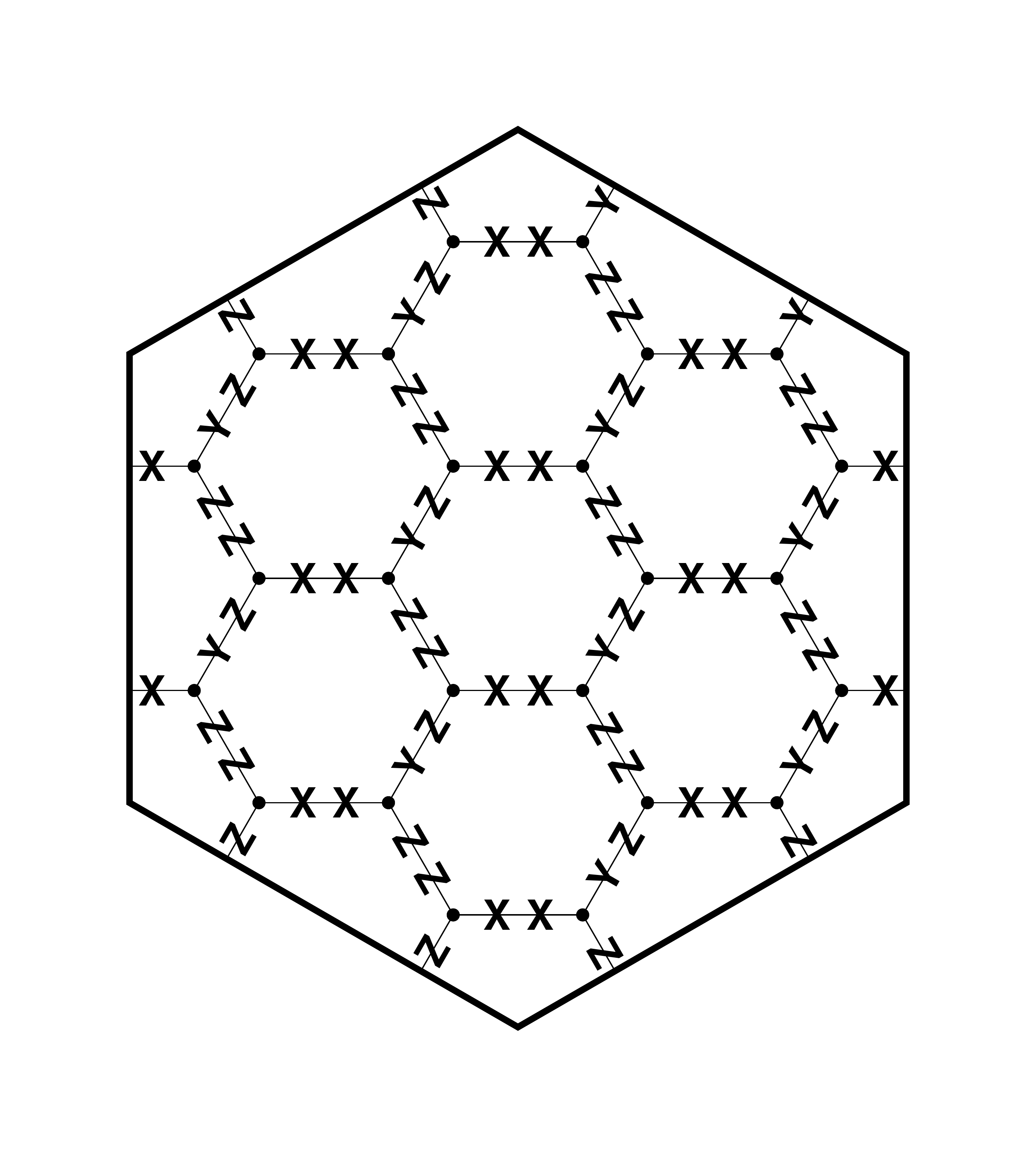} 
}\\
\subfloat[distance 4 code labelings] {
\label{figure:hextille-code-4-labelings}
\includegraphics[width=3.5in]{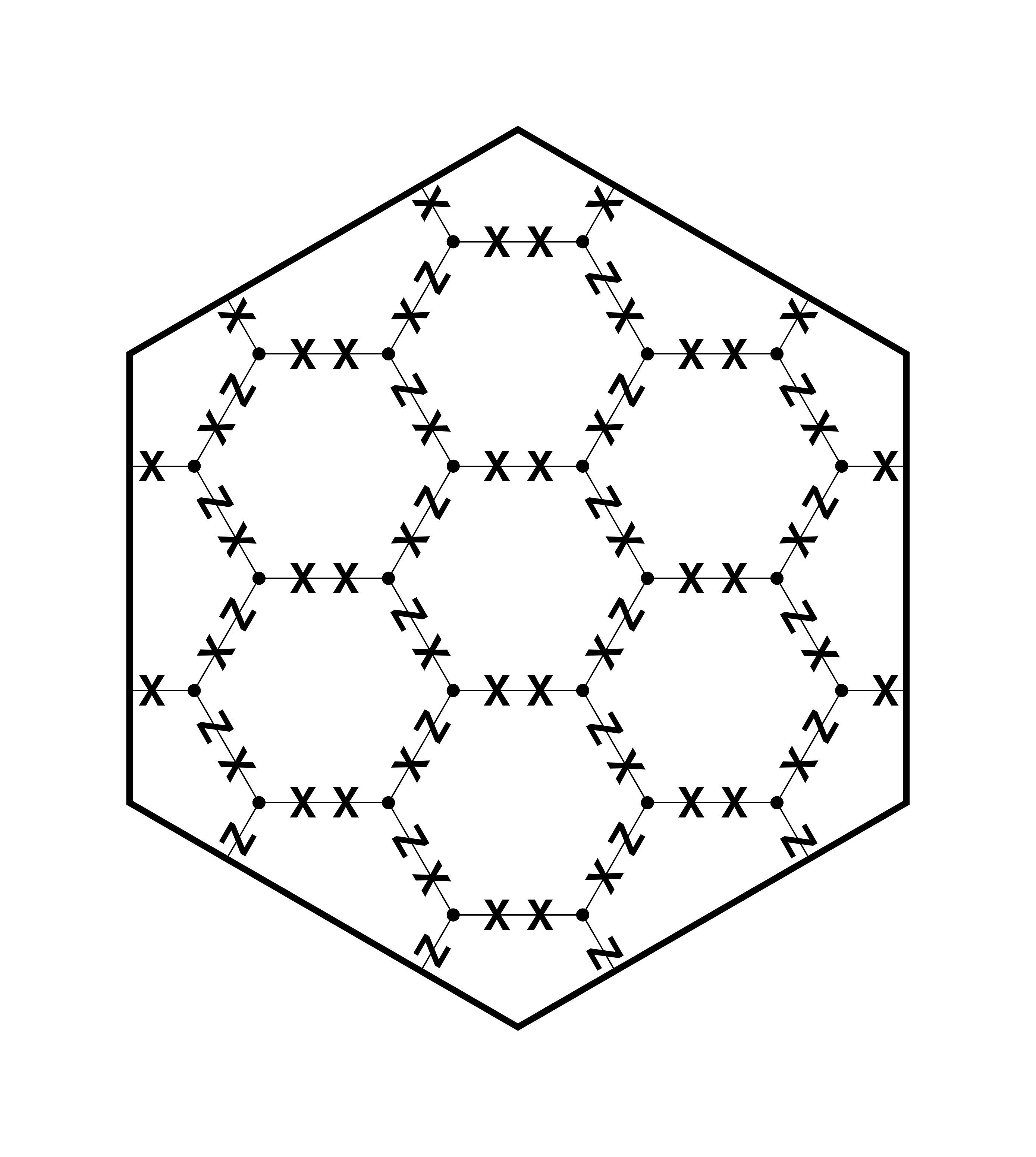} 
\includegraphics[width=3.5in]{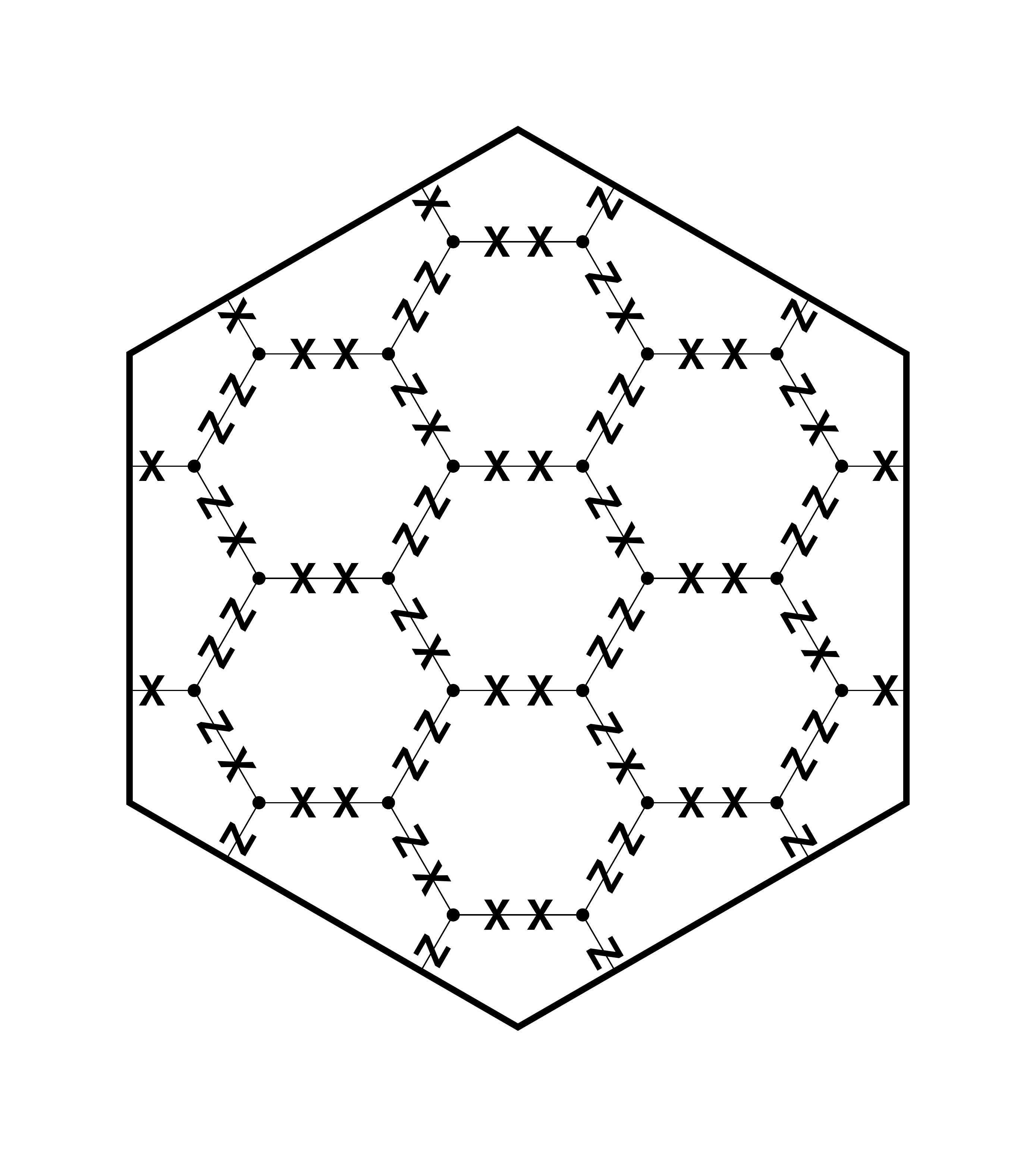} 
}
\caption{
\label{figure:hextille-code-labelings}
A figure illustrating four labelings of the hextille tiling (on a radius 2 lattice) that result in a quantum code with distance 3/4 (in respectively \ref{figure:hextille-code-3-labelings}/\ref{figure:hextille-code-4-labelings}), and a number of distance logical qubits proportional to the square of the radius of the labeling.
}
\end{figure*}
\subsubsection{truncated hextille}

In the truncated hextille there are four kinds of codes where the number of qubit increases with the radius:  three with the distance fixed at 3, and one with the distance fixed at 4.

The best of the distance 3 codes has the number of logical qubits ($l_3$) given by $l_3=2r^2-1$, where $r$ is the radius of the code.  Since the number of physical qubits ($n$) is given by $n=24r^2$, the number of logical qubits per physical qubit is thus given by $\frac{l_3}{n}=\frac{1}{12}\paren{1-\frac{1}{24r}}^2$, a quantity which converges to $\frac{1}{12}$ as $r\to\infty$.  The two labelings we saw which give rise to this code are illustrated in Figure \ref{figure:truncated-hextille-code-3-labelings}.

The best of the distance 4 codes has the number of logical qubits ($l_4$) given by $l_4=2r(r-1)$, and the number of logical qubits per physical qubit is thus given by $\frac{l_4}{n}=\frac{1}{12}\paren{1-\frac{1}{r}}$, a quantity which converges to $\frac{1}{12}$ as $r\to\infty$.  We see from this analysis that although the best distance 4 code contains fewer logical qubits than the best distance 3 code, they both converge to the same number of logical qubits per physical qubit in the large radius limit.  One of the nine labelings we saw which give rise to this distance 4 code are illustrated in Figure \ref{figure:truncated-hextille-code-4-labelings}.

\begin{figure*}
\includegraphics[width=3.5in]{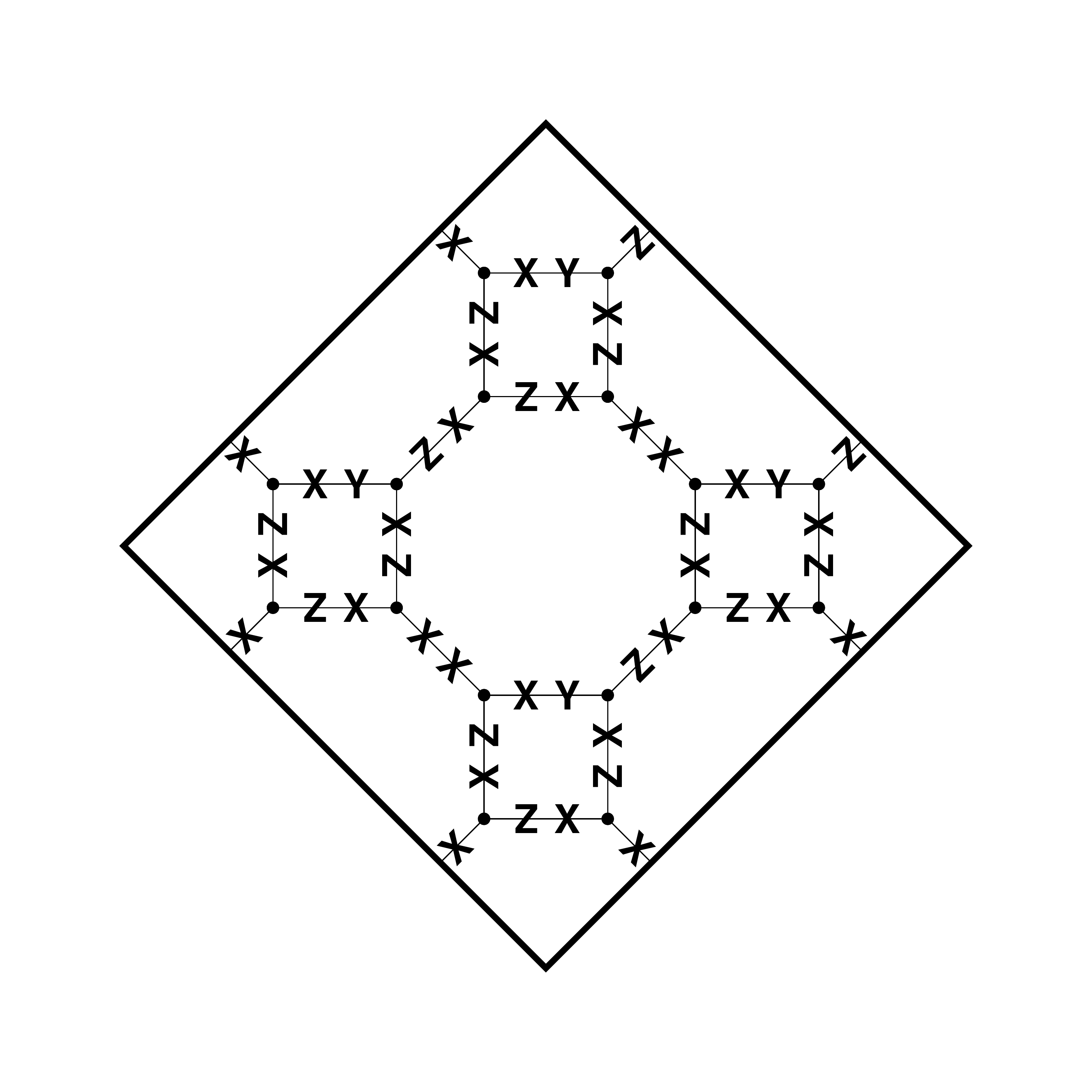} 
\includegraphics[width=3.5in]{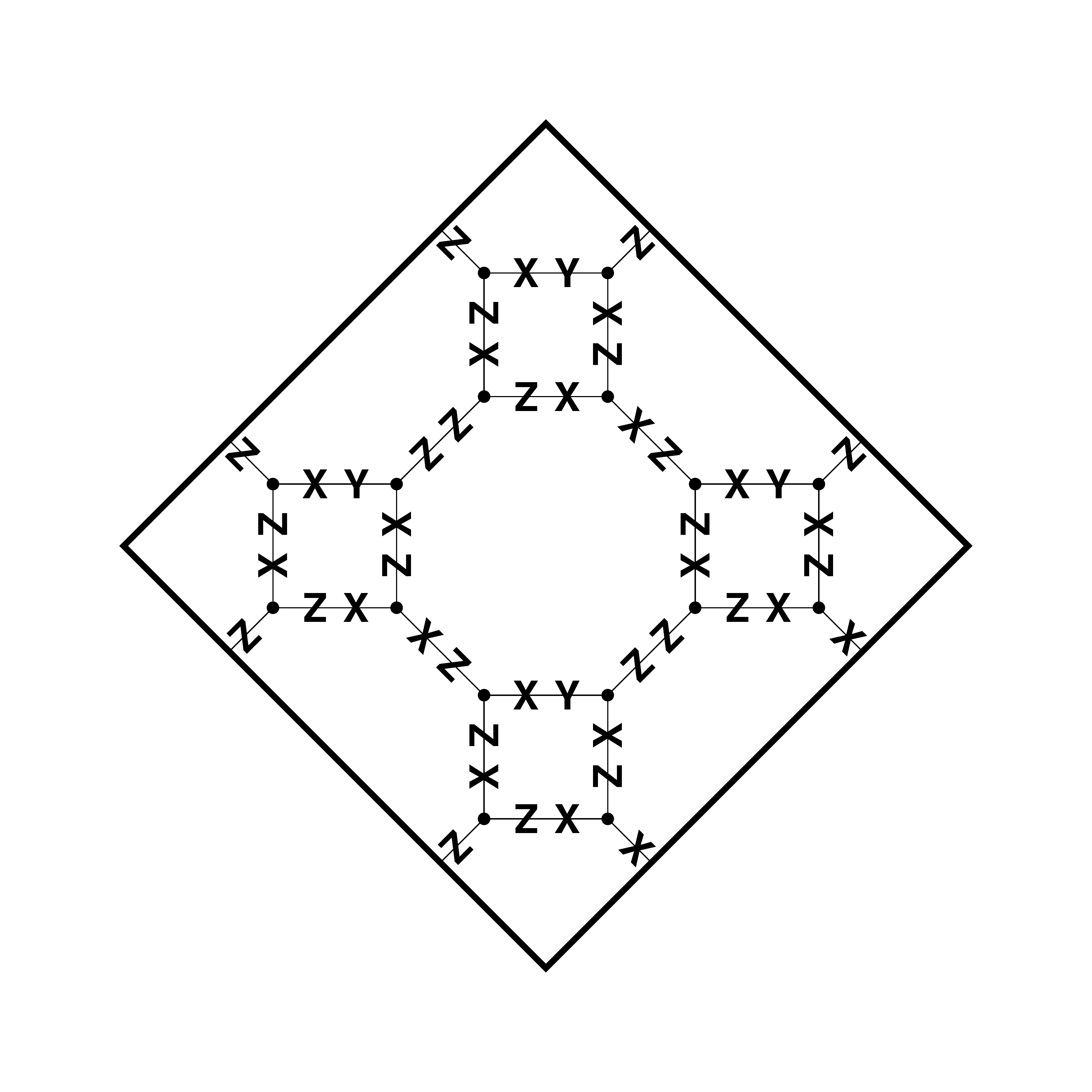} 
\caption{
\label{figure:truncated-hextille-code-3-labelings}
A figure illustrating two labelings of the truncated hextille tiling (on a radius 1 lattice) that result in a quantum code with distance 3 and a number of distance logical qubits proportional to the square of the radius of the labeling.
}
\end{figure*}

\begin{figure}
\includegraphics[width=3.5in]{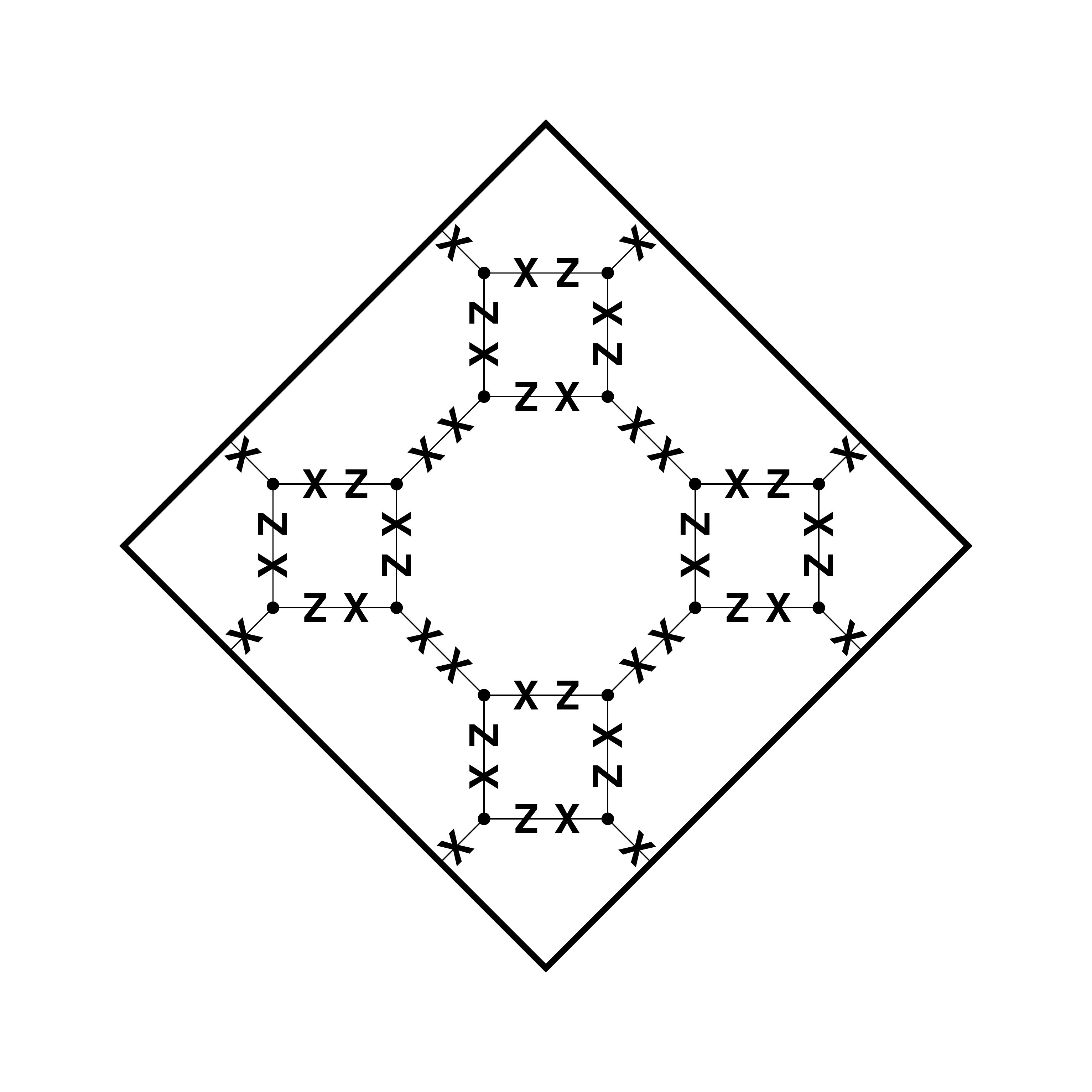} 
\caption{
\label{figure:truncated-hextille-code-4-labelings}
A figure illustrating two labelings of the truncated hextille tiling (on a radius 1 lattice) that result in a quantum code with distance 4 and a number of distance logical qubits proportional to the square of the radius of the labeling.
}
\end{figure}
\subsubsection{hexadeltille}

There are many codes that appear in the hexadeltille tiling, but it is difficult to draw conclusions about trends due to the limit on the size of the lattices that were scanned.  The good news, though, is that the reason why scanning larger radii was difficult is because there is a code in this tiling with a qubit whose distance grows with the radius of the lattice.  One of the nine labelings that we saw with this property is illustrated in Figure \ref{figure:hexadeltille-code-labeling}. 

\begin{figure}
\includegraphics[width=3in]{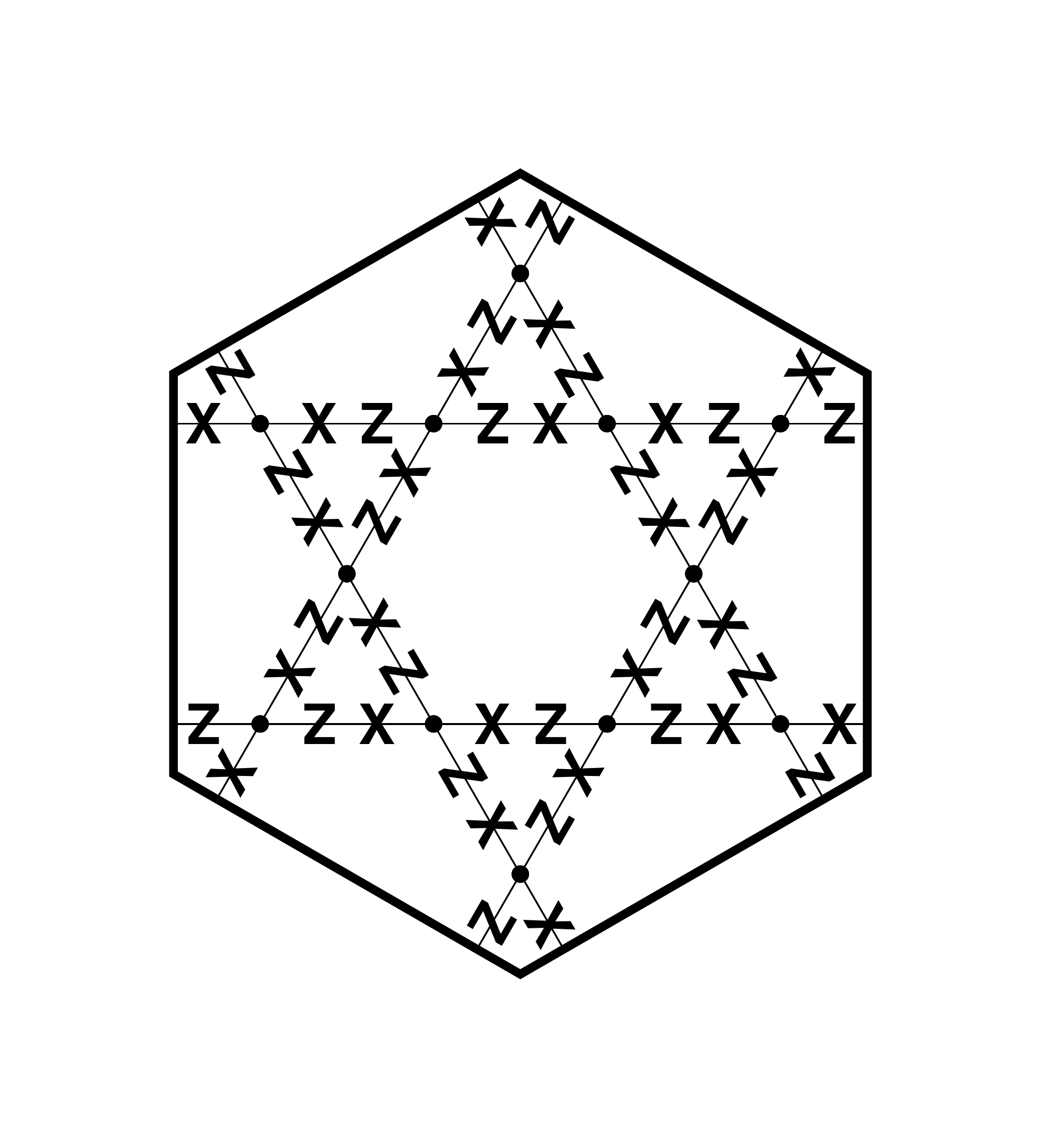}
\caption{
\label{figure:hexadeltille-code-labeling}
A figure illustrating one of the labelings of the hexadeltille tiling (on a radius 1 lattice) that results in quantum code whose distance grows with the size of the tiling.
}
\end{figure}
\subsubsection{rhombihexadeltille}

This tiling is interesting because it had many more labelings that resulted in codes  than all of the other tilings combined;  specifically, for the rhombihexadeltille tiling we saw 48,807 labelings that resulted in useful codes, whereas for all of the other tilings combined we saw only 421 labelings that resulted in useful codes.  This is even more remarkable considering that the largest lattice we were able to scan for the rhombihexadeltille tiling was smaller than that the for most of the other tilings.

As can be seen in Figure \ref{figure:results-hextille}, this tiling is also interesting because it features so many different kinds of codes, including both codes that seem to grow in the number of logical qubits with radius and codes that grow in distance with size.  It is the only tiling that features a lattice that contains a labeling resulting in a code for every distance up to 6.

The rhombihexadeltille tiling is the only tiling we have seen which has code both with a distance greater than 4 and multiple qubits;  we saw six labelings which resulted in codes with distance 6 and two qubits, and four labelings which resulted in codes with distance 5 and four qubits.  In Figure \ref{figure:rhombihexadeltille-code-56-labelings} we show an example of each of these labelings.

\begin{figure*}
\subfloat[distance 5 code labeling]{
\label{figure:rhombihexadeltille-code-5-labeling}
\includegraphics[width=4.75in]{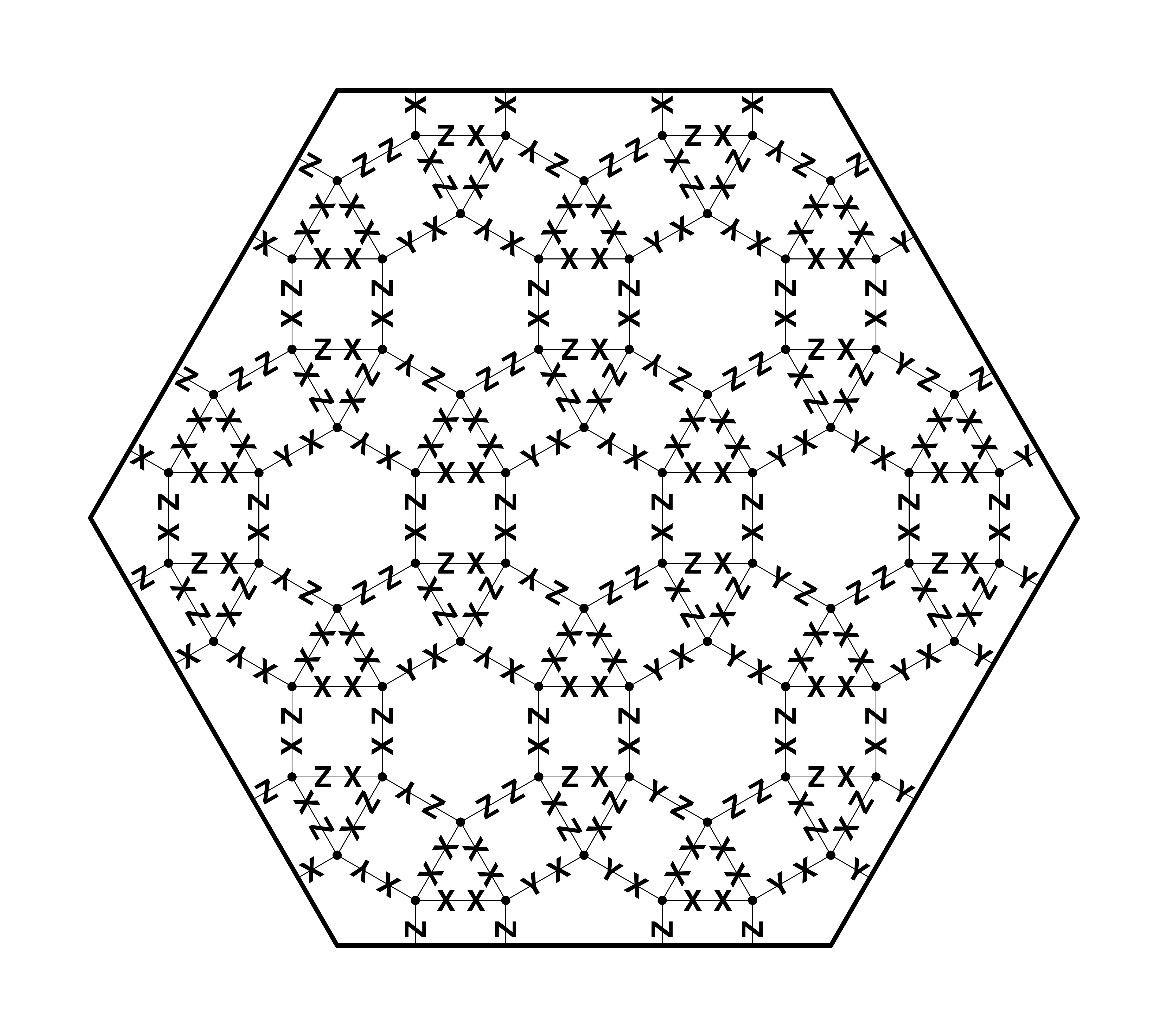} 
}\\
\subfloat[distance 6 code labeling]{
\label{figure:rhombihexadeltille-code-6-labeling}
\includegraphics[width=4.75in]{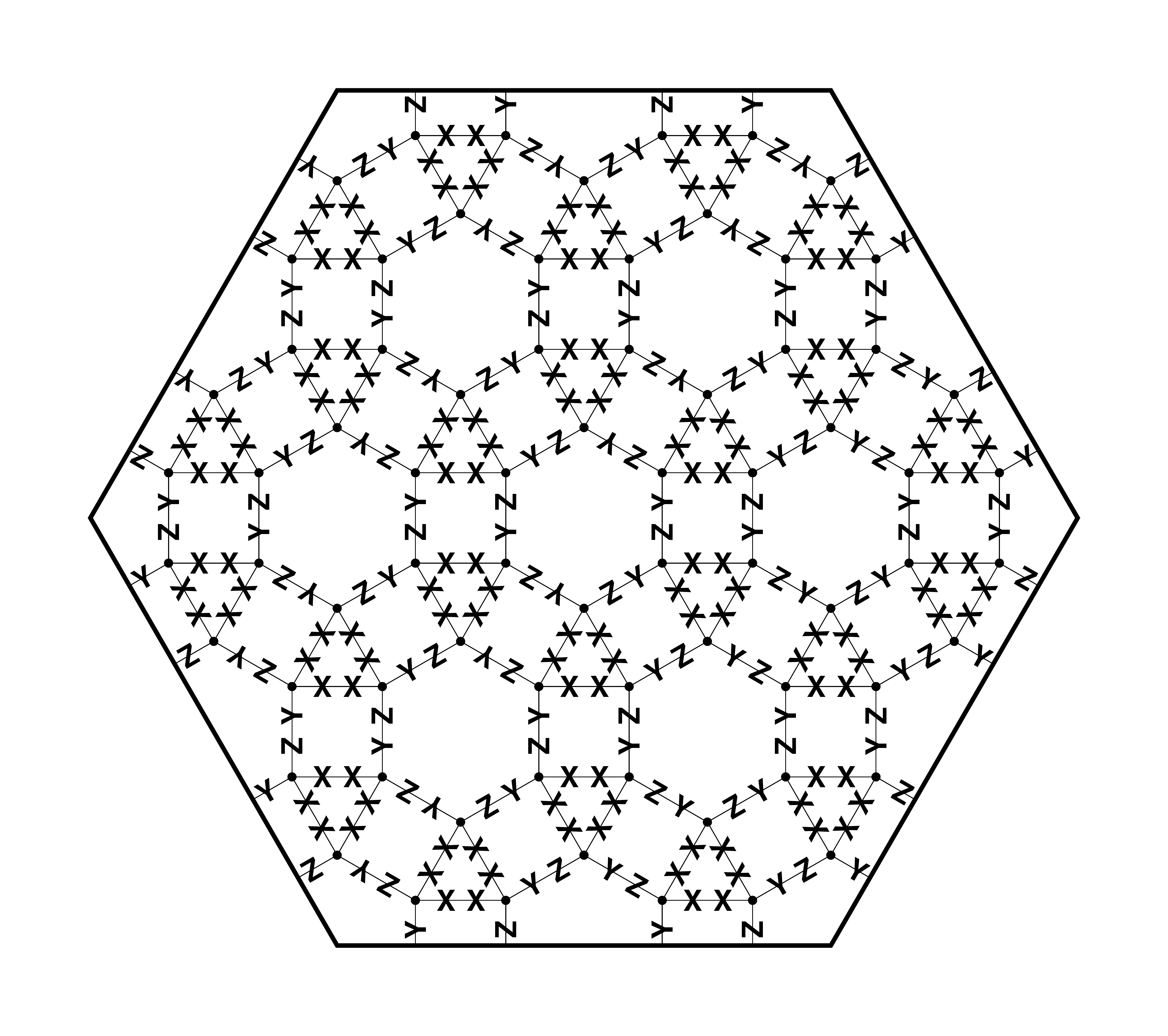} 
}
\caption{
\label{figure:rhombihexadeltille-code-56-labelings}
A figure illustrating labelings of the rhombihexadeltille tiling (on a radius 2 lattice) that result in a quantum code with distance 5 and four qubits (\ref{figure:rhombihexadeltille-code-5-labeling}) and a quantum code with distance 6 and two qubits (\ref{figure:rhombihexadeltille-code-6-labeling}) when applied to a radius 3 lattice.
}
\end{figure*}

Two of labelings resulted in codes with the highest number of qubits -- 16 logical qubits at distance 4 for a lattice of radius three.  These two labelings are illustrated in Figure \ref{figure:rhombihexadeltille-code-4-labelings}.

\begin{figure*}
\includegraphics[width=4.75in]{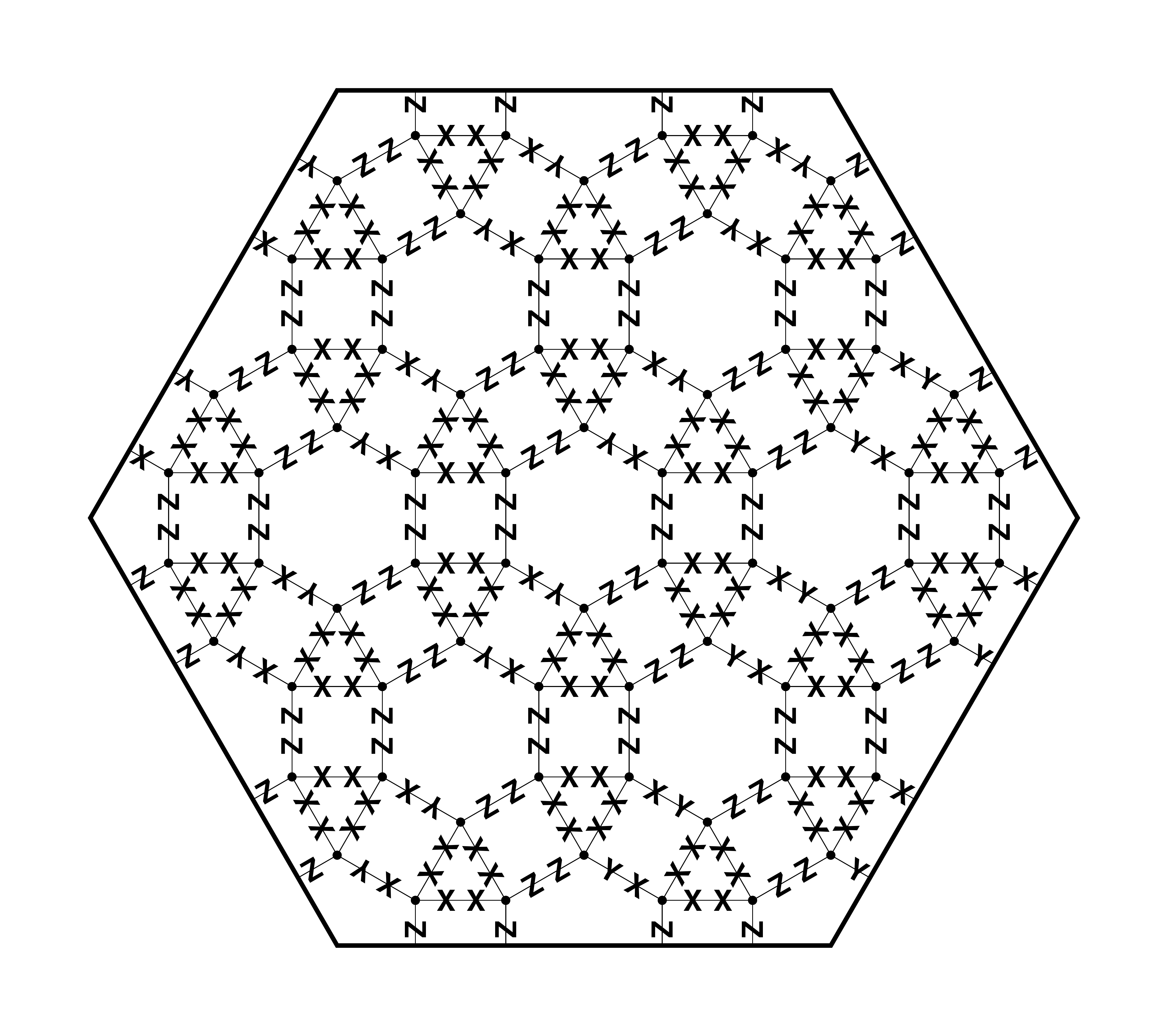} 
\\
\includegraphics[width=4.75in]{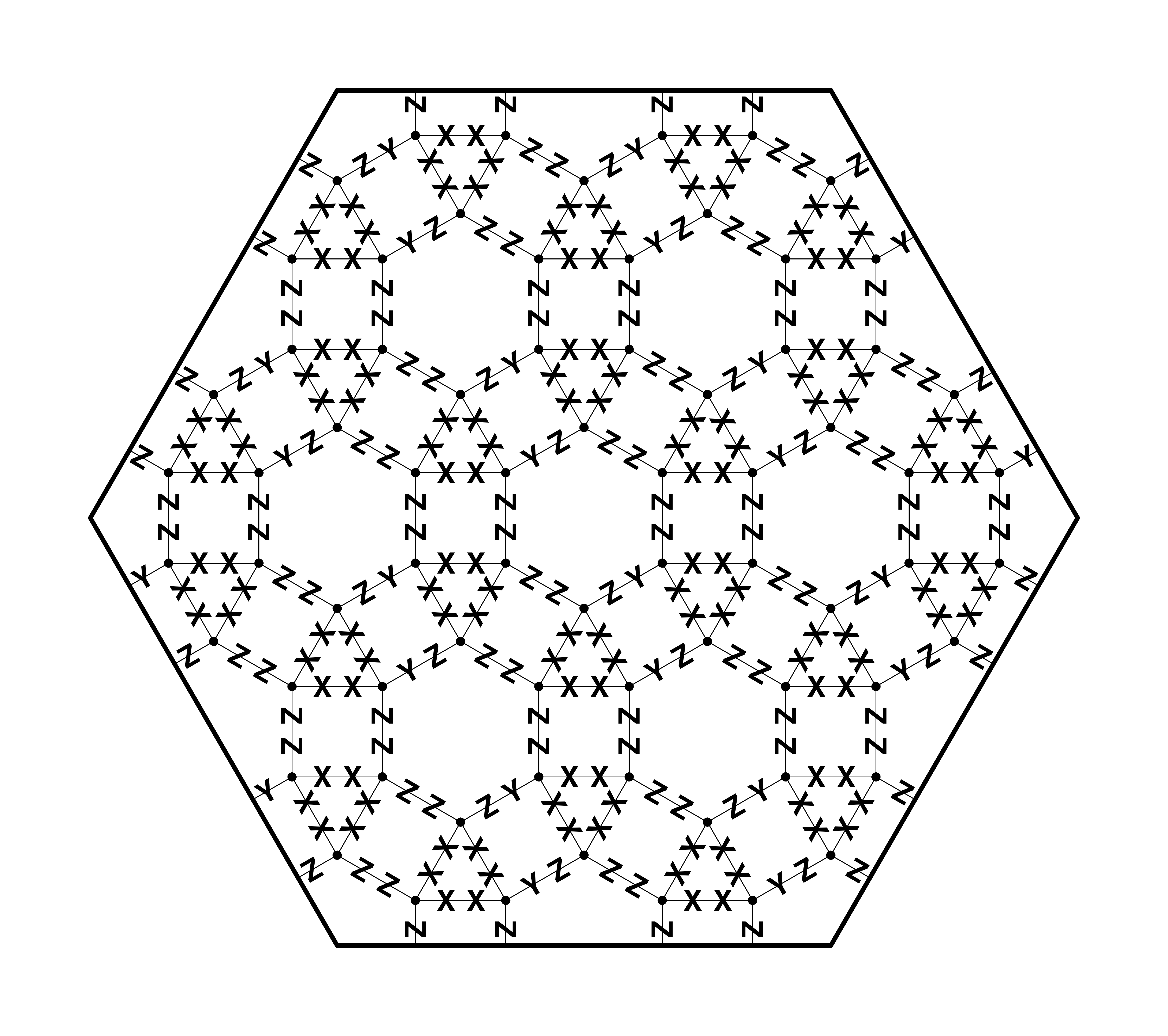} 
\caption{
\label{figure:rhombihexadeltille-code-4-labelings}
A figure illustrating two labelings of the rhombihexadeltille tiling (on a radius 2 lattice) that result in a quantum code with distance 4 and sixteen when applied to a radius 3 lattice.
}
\end{figure*}
\subsection{Discussion} \label{sec:discussion}

There are few surprises in our results.  For example, the best code that we found that maximized the logical qubit distance per physical qubit was the compass model code in the quadrille tiling, which is already well-known.  Furthermore, all of the codes obeyed the upper bounds $kd\in O(n)$ and $d^2\in O(n)$ --- where $k$ is the number of logical qubits in the code, $d$ is the distance of the code, and $n$ is the number of physical qubits implementing the code --- that were derived in \cite{Bravyi:10a} for codes having spatially local generators.

Some of the observed differences between the tilings are an artifact of the search space.  For example, every code found on the hextille tiling could also be implemented on the deltille tiling, but although we found two kinds of codes for the hextille tiling we found no codes for the deltille tilings.  This is because our search space included no way for the deltille tilings to ``knock out'' the middle qubits in each hexagonal tiling, and furthermore the hextille tiling search space included two classes of vertices which could have independent labelings whereas the deltille tiling search space had only one class of vertices.

Although it is not clear how many of the codes we found will have practical applications, the success of this search demonstrates the feasibility of using brute-force computation to find useful codes within a constrained search space.
\section{Conclusion}
\label{sec:conclusion}

In this paper we have presented an algorithm for computing the optimal quantum subsystem code that can be implemented using a given set of measurements.  We have shown that although this algorithm requires exponential time in the worst case, this exponential is a function of the code distance, and so the algorithm terminates (relatively) quickly when the optimal code has low distance.  Because of this, the algorithm can be used to perform a brute-force search through a space of possible measurements in order to see which give rise to ``useful'' (high-distance) codes.  We demonstrated the feasibility of this approach by applying the algorithm to search for codes implemented on systems with lattice structures corresponding to nine of the eleven convex vertex-uniform tilings, and on all but one of these nine tilings we found useful codes.

This algorithm should prove helpful in two kinds of ways in particular.  First, it can be applied in an exploratory setting to do the tedious work of computing the code resulting from a set of measurements so that the researcher can experiment with new ideas for choices of measurement to see how well they work.  Second, it can be applied to hone a `rough' idea for how a code might be implemented (such as a particular lattice configuration) into a concrete idea by scanning through the possible choices of the degrees of freedom to see if any result in useful codes; of course, cleverness can often come up with an answer more quickly than a computationally intensive search, but it is good to have the alternative of brute-force computation to fall back on when brute-force cleverness fails.

\bibliography{CodeQuest}

\end{document}